%% file: paperROCarxiv.tex
\documentclass[a4paper]{article}

\usepackage{amsthm,amsmath,amsfonts,amssymb,mathrsfs}
\usepackage{graphicx}
\usepackage{setspace}
\usepackage{rotate}

\usepackage{authblk}
\usepackage{natbib}

\usepackage{palatino}
\usepackage{epstopdf}
\usepackage{ifthen}
\newboolean{draft}
\newboolean{arxiv}

\setboolean{draft}{false} %
\setboolean{arxiv}{true} %

\newcommand{\ifDraft}[1]{\ifthenelse{\boolean{draft}}{#1}{}}
\newcommand{\ifelseDraft}[2]{\ifthenelse{\boolean{draft}}{#1}{#2}}

\newcommand{\ifArXiV}[1]{\ifthenelse{\boolean{arxiv}}{#1}{}}
\newcommand{\ifelseArXiV}[2]{\ifthenelse{\boolean{arxiv}}{#1}{#2}}

\usepackage{anysize}
\marginsize{1.5cm}{1.5cm}{1.5cm}{1.5cm}

\newlength{\refwidth} 
\usepackage{placeins} %

\usepackage{hyperref}

\usepackage[english]{babel}  %
\usepackage{todonotes}

\include{macros}

\include{macrosROC}

\newtheorem{lemma}{Lemma}
\newtheorem{cor}{Corollary}
\newtheorem{prop}{Proposition}

\graphicspath{{./}{./figures/}{./pix-auto/}}

\title{ROC curves for spatial point patterns and presence-absence data}

\author[1]{Adrian Baddeley}
\author[2]{Ege Rubak}
\author[3,4]{Suman Rakshit}
\author[5]{Gopalan Nair}
\affil[1]{School of Population Health, Curtin University, Perth, Australia}
\affil[2]{Department of Mathematical Sciences, Aalborg University, Aalborg, Denmark}
\affil[3]{School of Electrical Engineering, Computing, and Mathematical Sciences, Curtin University, Perth, Australia}
\affil[4]{Curtin Biometry and Data Analytics, Centre for Crop and Disease Management, Curtin University, Perth, Australia}
\affil[5]{School of Mathematics \& Statistics, University of Western Australia, Perth, Australia} %

\date{}

\begin{document}
\maketitle

\setcounter{tocdepth}{2}

\input{abstract}

\maketitle

\ifDraft{\centerline{\LARGE >> \today << }}

\ifDraft{{\tableofcontents}\pagebreak}

\input{mainbody}

\appendix

\addcontentsline{toc}{section}{APPENDICES}
\include{appendices}

\include{supplementCore}

\end{document}

%% file: macros.tex
\newcommand{\R}{{\mathbb R}}
\newcommand{\EE}{{\mathbb E}}
\newcommand{\PP}{{\mathbb P}}
\newcommand{\Prob}[1]{\PP\{ {#1} \}}
\newcommand{\var}[1]{\mbox{{\sf var}}[ {#1} ]}

\newcommand{\indicate}[1]{\boldmaths{1}\{ {#1} \}}
\newcommand{\dee}[1]{\, {\rm d}{#1}}

\newcommand{\boldmaths}[1]{{\ensuremath\boldsymbol{#1}}}
\newcommand{\bx}{\boldmaths{x}}
\newcommand{\by}{\boldmaths{y}}
\newcommand{\bz}{\boldmaths{z}}
\newcommand{\bX}{\boldmaths{X}}

\newcommand{\btheta}{\boldmaths{\theta}}

%% file: macrosROC.tex
\newcommand{\bpi}{\boldmaths{\pi}}

\newcommand{\namestyle}[1]{\mbox{\textsf{#1}}}
\newcommand{\FP}{\namestyle{FP}}
\newcommand{\TP}{\namestyle{TP}}
\newcommand{\AUC}{\namestyle{AUC}}
\newcommand{\FPhat}{\widehat{\FP}}
\newcommand{\TPhat}{\widehat{\TP}}

\newcommand{\FPalt}{\widetilde{\FP}}

\newcommand{\hiSymbol}{>}
\newcommand{\loSymbol}{<}
\newcommand{\hi}[1]{{#1}^\hiSymbol}
\newcommand{\lo}[1]{{#1}^\loSymbol}

\newcommand{\symsub}[2]{{#1}_{\mbox{\scriptsize \textsc{#2}}}}
\newcommand{\Fsub}[1]{\symsub{F}{#1}}
\newcommand{\Fpos}{\Fsub{P}}
\newcommand{\Fneg}{\Fsub{N}}
\newcommand{\Ssub}[1]{\symsub{S}{#1}}
\newcommand{\Spos}{\Ssub{P}}
\newcommand{\Sneg}{\Ssub{N}}

\newcommand{\lamtrue}{\lambda_0}

\newcommand{\hyp}[1]{{\mathcal H}_{#1}}

%% file: abstract.tex
\begin{abstract}
  Receiver Operating Characteristic (ROC) curves have recently been
  used to evaluate the performance of models for spatial presence-absence
  or presence-only data.
  Applications include species distribution modelling %
  and mineral prospectivity analysis. %
  We clarify the interpretation of the ROC curve in this context.
  Contrary to statements in the literature, ROC
  does not measure goodness-of-fit of a spatial model, and its
  interpretation as a measure of predictive ability is weak;
  it is a measure of ranking ability,
  insensitive to the precise form of the model.
  To gain insight we draw connections between ROC
  and existing statistical techniques for spatial point pattern data.
  The area under the ROC curve (AUC) is related to hypothesis tests
  of the null hypothesis that the explanatory variables have no effect.
  The shape of the ROC curve has a diagnostic interpretation.
  This suggests several new techniques, which extend the scope of application
  of ROC curves for spatial data, to support 
  variable selection and model selection, analysis of segregation
  between different types of points, adjustment for a baseline,
  and analysis of spatial case-control data.
  The new techniques are illustrated with several real example datasets.
  Open source \textsf{R} code implementing the techniques is available in the development version of
  our package \texttt{spatstat} \citep{baddturn05,baddrubaturn15} and will be included in the
  next public release.
\end{abstract}

%% file: mainbody.tex
\section{Introduction}

Receiver Operating Characteristic (ROC) curves
have long been used to assess the performance of decision rules,
classifiers and hypothesis tests \citep{krzahand09,namdago02}.
The area under the ROC curve, AUC, provides a single numerical summary
of performance, useful for comparing several competing tests or classifiers.

In recent decades, these tools have been used 
to evaluate the performance of models for spatial presence-absence
and presence-only data
in spatial ecology (species distribution models, \citealp{fran09})
and in exploration geoscience
(mineral prospectivity indices or mineral potential maps,
\citealp{Porwal2010,fabbchun08,fordetal19,khad25thesis}).
In these fields of application, the interpretation of ROC and AUC has been controversial.
AUC is variously claimed to be
a measure of \emph{goodness-of-fit} of a predictive model \citep{fielbell97},
a measure of \emph{predictive power} \citep{lobojimereal07,aust07,fran09}
or a tool for \emph{validation} of the model \citep{nykaetal15}.
Some writers correctly warn that the ROC for spatial data depends on the choice of the
spatial study domain \citep{jime12} and ignores the overall abundance
or prevalence of the species \citep{manewillorme01}.

The statistical interpretation of ROC curves for spatial data needs clarification.
It is not necessarily valid to translate the original definition
of the ROC curve (for classifying individuals in a population) into the context
of spatial presence-absence data (for classifying pixels in a spatial grid)
because pixels are not independent individuals. At the very least, the interpretation of ROC
is different in the spatial context.

This article aims to clarify the interpretation of ROC and AUC
for spatial presence-absence data,
to identify their methodological strengths and weaknesses,
and to develop improvements and extensions, including a new extension
to spatial point pattern data.

In the literature on spatial data cited above,
the term ``ROC curve'' refers exclusively
to one particular form of the ROC curve,
based on probabilities predicted by a model (which we call the ``M-ROC'').
In this paper we emphasise that many alternative applications of ROC
curves are possible
and useful for spatial data.

We find that the ROC and AUC for models of spatial presence-absence data
do not measure goodness-of-fit of a predictive model. Rather, they measure
``badness-of-fit'' of the \emph{null} model,
i.e.\ the model in which
the probability of presence is the same at all locations.
The M-ROC is the most optimistic of all variants of the ROC curve,
is liable to the effects of over-fitting, and has poor ability to detect
certain departures from the fitted model.

The interpretation of ROC and AUC as measures of predictive power
is also weaker than generally believed.
In the case of a single explanatory variable, ROC and AUC are almost
completely unable to discriminate between different models or methodologies
when applied to the same data; they only measure the ability of the
\emph{explanatory variable} to rank or segregate the study region into
areas of high and low probability of presence.
Furthermore, 
this ability to segregate space always refers to a specific region of space;
the ROC and AUC can change markedly if we restrict attention to a sub-region
of the original study region, or if we consider an adjacent spatial region,
giving rise to instances of Simpson's Paradox. ROC curves cannot be used
to predict behaviour on other spatial regions, or to predict the response
to changes in covariate values (such as predicting the effect of climate
change on a species range).
Nevertheless, the fact that the M-ROC does not depend
on the precise form of the model can be an advantage
for some tasks such as variable selection.

To gain insight, we connect the ROC and AUC 
to other statistical techniques for spatial data,
including hypothesis tests (Berman-Waller-Lawson, Kolmogorov-Smirnov),
P--P plots (capture efficiency curve, fitting rate curve),
and rate estimation (resource selection function, mineral prospectivity index,
mineral potential map,
continuous Boyce index, point process intensity).
AUC is related to a test of the null hypothesis
that the covariates have no effect, rather than a goodness-of-fit test.
These connections improve understanding of the correct interpretation,
statistical behaviour, and strengths and weaknesses of ROC and AUC,
and help to identify conditions where they perform well. The strong connections
also imply that these various techniques cannot be treated as
independent pieces of corroborating evidence in favour of a model.

We develop new applications of the ROC and AUC
to spatial presence-absence data. These new techniques support
nonparametric analysis of dependence on a covariate,
variable selection and model selection, analysis of segregation
between different types of points, adjustment for a baseline,
and analysis of spatial case-control data.
We develop ``fitted'' and ``theoretical'' versions of the ROC curve
for model validation. We provide variance formulae
and strategies for avoiding overfitting.
The new techniques are illustrated with several real example datasets.
We discuss fields of potential application,
including spatial ecology and geological prospectivity.

We also develop the ROC and AUC for \emph{spatial point pattern} data,
that is, data providing the exact spatial locations of the objects of interest.
In geological prospectivity analysis, the primary survey data are
the exact spatial coordinates of the mineral deposits, and these are usually
discretised to obtain presence-absence data for analysis.
This discretisation step is unnecessary: ROC and AUC can be calculated
directly from the exact spatial coordinate data. Analysis based
on the exact coordinates is often simpler than analysis
based on discretised pixel data. These results also imply that ROC curves
derived from presence-absence data using different choices of pixel size
are consistent, asymptotically as pixel size tends to zero,
because the limit is the spatial point pattern case. On the other hand
the results also indicate that the ROC and AUC depend implicitly on the
spatial coordinate system.

We find that the \emph{shape} of the ROC curve has a diagnostic interpretation.
Roughly speaking, if the ROC curve is concave or convex, then thresholding
of the fitted probabilities is appropriate, and otherwise it is inappropriate.

The plan of the paper is as follows.
Section~\ref{S:background:ROC} recalls the definition
of ROC and AUC in a non-spatial context.
Section~\ref{S:background:spatial}
introduces spatial point pattern data, presence-absence data, spatial covariates, and
parametric statistical modelling of point pattern and presence-absence data
as a function of covariates.
Section~\ref{S:ROC:covariate} develops the conceptually simplest application
of the ROC curve to spatial presence-absence data and spatial
point pattern data, for a single explanatory variable. This is not the
most common version of the ROC curve but it gives clearer insight into
strengths and weaknesses.
Section~\ref{S:ROC:model} describes the ``model ROC curve''
that is typically used to assess the performance of a species distribution
model or a mineral prospectivity index.
It also investigates new kinds of ROC curves \emph{predicted by}
a species distribution model or prospectivity index.
Section~\ref{S:tests} shows that AUC and other measures of ``performance''
are related to tests of the null hypothesis that the explanatory variables
have no effect, rather than tests of goodness-of-fit.
In Section~\ref{S:rhohat} we explain how to correctly ``read'' the ROC curve
using its connection with the resource selection function
(prospectivity index, intensity function).
Section~\ref{S:extensions} develops new extensions of the ROC curve
for spatial data, including ROC curves with non-uniform weights and
non-uniform baselines, and applies them to model selection.
Section~\ref{S:modelcheck} proposes new techniques for model checking,
by comparing the empirical and model-predicted ROC curves,
and by assessing convexity.
We end with a Discussion in Section~\ref{S:discussion}.

\section{Background: ROC curves}
\label{S:background:ROC}

Here we briefly recall the definition of ROC and AUC in a
general context \citep{krzahand09}.

\subsection{ROC curve}
\label{S:background:ROC:def}

Consider a classification rule or hypothesis test rule
based on the value of a statistic $S$. An observation is allocated to
the `Positive' group (the null hypothesis is rejected) if $S > t$,
and allocated to the `Negative' group (the null hypothesis is Not rejected)
if $S \le t$, where $t$ is a threshold value to be chosen.
The probability of a true positive (the power of the test)
is
\begin{equation}
  \label{eq:TP}
\TP(t) = \Prob{S > t \mid \mbox{true status is Positive}}
\end{equation}
while the probability of a false positive (the size of the test)
is
\begin{equation}
  \label{eq:FP}
\FP(t) = \Prob{S > t \mid \mbox{true status is Negative}}.
\end{equation}
A good classifier is one which achieves a large value of $\TP(t)$
for any given value of $\FP(t)$.

The ROC curve is loosely defined as a plot of $\TP(t)$ against $\FP(t)$
for all possible thresholds $t$. This may not define a continuous curve
since $\TP$ and $\FP$ may be discontinuous; in that case it seems to be common practice
to use a piecewise linear interpolant. 
We shall instead define the ROC curve as %
the graph of $R(p)$ against $p$ for all $0 \le p \le 1$, where
\begin{equation}
  \label{eq:roc}
  R(p)=\TP( \FP^{-1}(p) ), \quad\quad 0 \le p \le 1
\end{equation}
where $\FP^{-1}(p) = \max\{ t: \FP(t) \ge p \}$ is the left-continuous
inverse function of $\FP$. Then $R(p)$ is a left-continuous, non-decreasing
function of $p$.
Good performance of the classifier is indicated by an ROC curve
lying well above the diagonal line $y=x$.

In some applications it is appropriate to reverse the direction of thresholding,
that is, to treat smaller values of $S$ as more favorable to the Positive group.
The classification rule allocates a subject to the Positive group
if $S \le t$. This ``reversed ROC'' curve will be denoted $\lo R(p)$.

\subsection{Empirical ROC from a finite dataset} 
\label{S:background:ROC:empirical}

Assume we observe data $(s_1,y_1), \ldots, (s_J, y_J)$ for $J$ individuals,
where $s_j$ is the value of the statistic $S$ for the $j$th individual
and $y_j$ is the indicator of true status,
(i.e.\ $y_j=1$ if the $j$th individual truly belongs to the Positive group, and
$y_j=0$ if Negative).
The empirical ROC curve based on these data is the (interpolated) graph of 
the empirical true positive rate
\begin{equation}
  \label{eq:TPhat:finite}
  \TPhat(t) = \frac{
    \sum\nolimits^{J}_{j=1} y_j \ \hat y_j
  }{
    \sum\nolimits^{J}_{j=1} y_j
  }
  = \frac{
    \#\{ j: \; y_j = 1, \; s_j > t \}
  }{
    \#\{ j: \; y_j = 1  \}
  }
\end{equation}
against the empirical false positive rate.
\begin{equation}
  \label{eq:FPhat:finite}
  \FPhat(t) =  \frac{
    \sum\nolimits^{J}_{j=1} (1-y_j) \ \hat y_j
  }{
    \sum\nolimits_{j=1}^{J} (1-y_j)
  }
  = \frac{
    \#\{ j: \; y_j = 0, \; s_j > t \}
  }{
    \#\{ j: \; y_j = 0  \}
  }.
\end{equation}
Equivalently, the empirical ROC curve is the plot of 
$\widehat R(p) = \TPhat(\FPhat^{-1}(p))$ against $p \in [0,1]$.
This should not be conflated with the true (``theoretical'') ROC curve $R(p)$
defined in \eqref{eq:roc}.

Of course,
\eqref{eq:TPhat:finite}
and
\eqref{eq:FPhat:finite}
are not the only possible estimators of the
functions $\TP$ and $\FP$; kernel estimates
\citep{hallhynd03,hallhyndfan04,zouhall00,zouEtal97} are discussed in
Section~\ref{ss:ROC:covariate:casecontrol}.

\subsection{Measures of `performance'}
\label{S:background:ROC:performance}

When a single numerical value of `performance' of a classifier is required, a popular
choice is the area under the ROC curve,
\begin{equation}
  \label{eq:AUC}
  \AUC = \int_0^1 R(p) \dee p = \int_{-\infty}^\infty \TP(t) \dee \FP(t).
\end{equation}
The last integral on the right hand side is a Stieltjes integral.
Clearly $0 \le \AUC \le 1$, and larger values of $\AUC$ are conventionally
interpreted as implying greater discriminatory power.
If the ROC falls on the diagonal line $R(p) \equiv p$ then $\AUC = 1/2$.
Consequently, values of $\AUC$ close to $1/2$
are interpreted as indicating a complete lack of discriminatory power.
The area under the ``reverse'' ROC curve $\lo R(p)$
will be denoted $\lo\AUC$ and satisfies $\lo\AUC = 1 - \AUC$.

An alternative to AUC is the one-sided Youden criterion
(\citealp{youd50}, \citealp[p.\ 30]{krzahand09})
defined as the greatest positive deviation of the ROC curve above the diagonal,
\begin{equation}
  \label{eq:Youden}
  J = \max_t (\TP(t) - \FP(t))_{+} = \max_p (R(p) - p)_{+}
\end{equation}
where $x_{+} = \max(0, x)$ denotes the positive part.
This originated in medical diagnostics and has been used in mineral prospectivity analysis
\citep{ruopetal08,chenwu19,baddetal21thresh,khad25thesis}.

\subsection{Relation of ROC to comparison of probability distributions}
\label{S:background:ROC:P-P}

The ROC curve is closely related to the \emph{P--P plot} \citep{wilkgnan68},
a standard tool for comparing two probability distributions. While the ROC curve is
usually conceived as a comparison between the \emph{disjoint} populations of positive
and negative individuals, the P--P plot is a comparison between any
two probability distributions, including distributions in overlapping populations.

Consider two random variables
$X$ and $Y$ with cumulative distribution functions
$F_X(t) = \Prob{X \le t}$ and $F_Y(t) = \Prob{Y \le t}$.
The P--P plot is a plot of $F_Y(t)$ against $F_X(t)$ for all $t$,
or equivalently, the graph of $H(p)$ against $p$
for $0 \le p \le 1$, where
\begin{equation}
  \label{eq:P-P}
  H(p) = F_Y(F_X^{-1}(p)),
\end{equation}
where $F_X^{-1}(p) = \inf\{ x: F_X(x) \ge p\}$ is the right-continuous
inverse function of $F_X$.

In the context of ROC curves, let $\Spos$ and $\Sneg$ be random variables
denoting the discriminant scores for randomly selected members
of the positive and negative groups
respectively. These random variables have cumulative distribution functions
$\Fpos(t) = 1 - \TP(t)$ and $\Fneg(t) = 1 - \FP(t)$. Accordingly,
the ROC curve is
\begin{equation}
  \label{eq:rocF}
  R(p)= \TP(\FP^{-1}(p)) = 1-\Fpos( {\Fneg}^{-1}(1-p) ),
  \quad\quad 0 \le p \le 1.
\end{equation}
and conversely the P--P plot is the ``reversed'' ROC curve
\begin{equation}
  \label{eq:roc-F}
  \lo R (p)= \Fpos( {\Fneg}^{-1}(p)) = 1 - \TP(\FP^{-1}(1-p)),
  \quad\quad 0 \le p \le 1
\end{equation}
so that $\lo R(p) = 1 - R(1-p)$.
That is, the P--P plot and ROC curve are equivalent when the axes are reversed.

\subsection{Probability integral transformation}

The connection to P--P plots shows that ROC curves 
are related to the probability integral
transformation. If the distributions of $\Sneg$ and $\Spos$ are continuous,
for example, then $\lo R(p)$ is the cumulative distribution function
\begin{equation}
  \label{eq:Rneg=cdf}
  \lo R(p) = \Prob{ \Fneg(\Spos) \le p }
\end{equation}
of the values $U = \Fneg(\Spos)$, and similarly 
\begin{equation}
  \label{eq:R=cdf}
  R(p) = \Prob{ \Fneg(\Spos) \ge 1 - p }
  = \Prob{ 1 - \Fneg(\Spos) \le p }
\end{equation}
is the c.d.f.\ of the values $V = 1 - \Fneg(\Spos)$.
If the distributions of $\Sneg$ and $\Spos$ were identical
(and continuous) then $U$ and $V$ would be uniformly distributed
between 0 and 1.

From this it can be derived that $\AUC$ has the elegant interpretation
\begin{equation}
  \label{eq:AUC=prob}
   \AUC = \Prob{\Spos > \Sneg} 
\end{equation}
where $\Spos$ and $\Sneg$ are values of $S$ for independent, randomly-selected
members of the Positive and Negative groups, respectively.
We also have $\AUC = \EE[U] = 1 - \EE[V]$.

\section{Background: spatial data and models}
\label{S:background:spatial}

We are concerned with the spatial distribution of `individuals',
which may be individual animals or organisms of a particular species,
or mineral deposits of a particular kind. 
We assume the individuals have negligible size at the scale of interest.
The original observations may be either a list of the exact spatial coordinates
of the observed individuals, or an array of binary variables
indicating the presence or absence of individuals in the cells of a grid.

\subsection{Spatial point pattern}
\label{S:background:spatial:point}

In a ``mapped'' spatial point pattern 
\citep{digg83,digg14,ripl81,cres91,illietal08,baddrubaturn15}
the available data are the exact spatial locations of the individuals
observed within a given sampling region $W \subset \R^d$. The point pattern is
a finite set
$
  \bx = \{x_1, \ldots, x_n\}, \quad n \ge 0, \quad x_i \in W
$
where the number of points $n$ is not fixed in advance. It will be assumed
that two individuals cannot exist at the same location.

The left panel of Figure~\ref{F:bei}
shows a typical mapped point pattern dataset giving the 
locations of 3604 trees of the species
\textit{Beilschmiedia pendula} (Lauraceae)
in a 1000-by-500 metre rectangular sampling region
in the tropical rainforest of Barro Colorado Island
\citep{hubbfost83,condhubbfost96,cond98}.

\begin{figure}[!h]
  \centering
  \centerline{
    \includegraphics*[width=0.45\refwidth,bb=15 120 385 310]{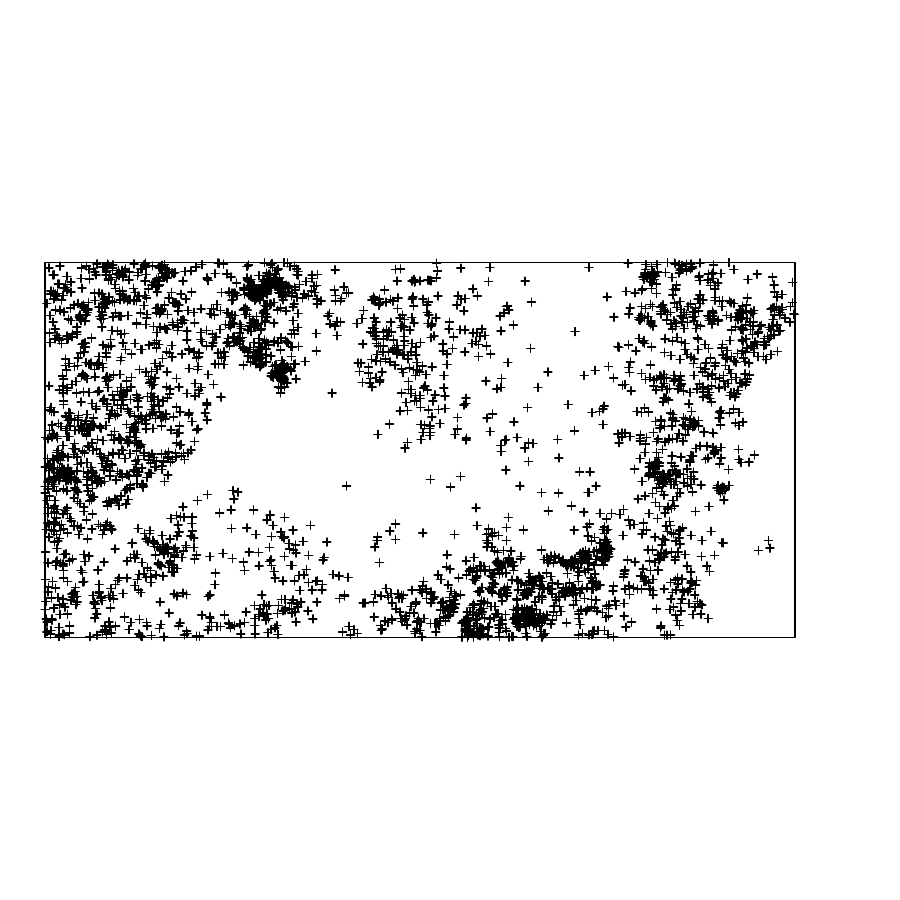}
    \includegraphics*[width=0.45\refwidth,bb=13 120 385 310]{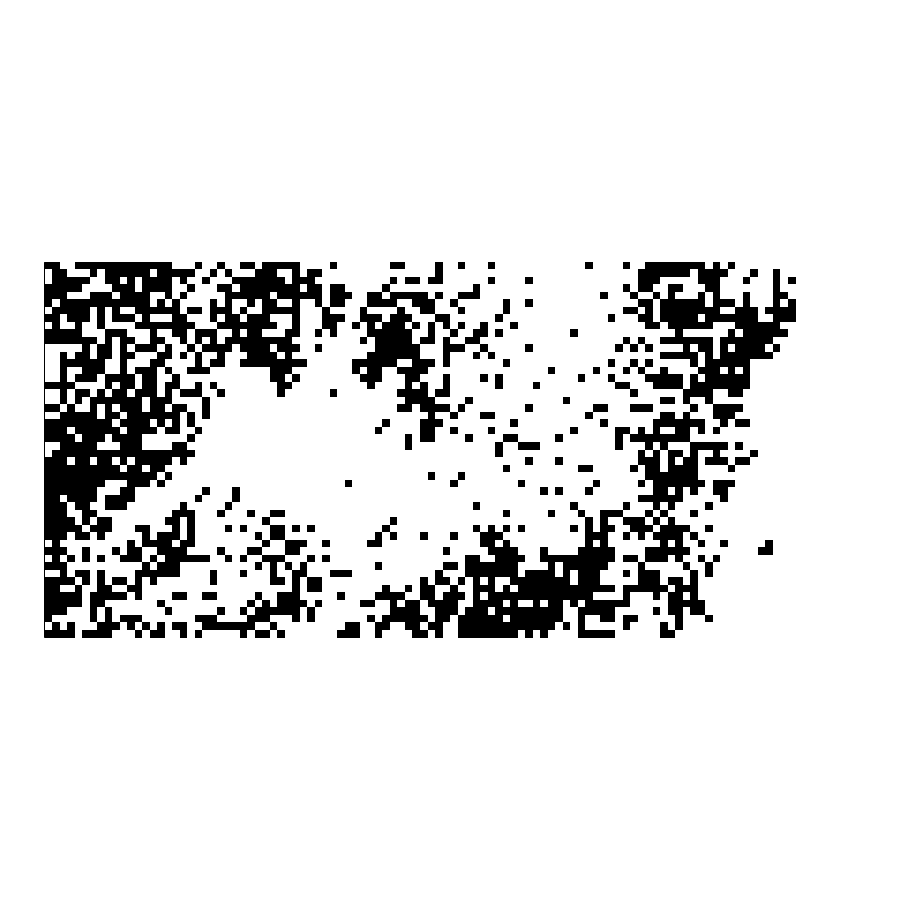}
  }
  \caption{\textit{Beilschmiedia} tree data.
    \emph{Left:}
    Exact tree locations ($+$) in a 1000 by 500 metre survey region.
    \emph{Right:}
    Presence-absence indicators for a grid of 10-metre square pixels.
    Black indicates presence.
  }
  \label{F:bei}
\end{figure}

\subsection{Presence-only and presence-absence data}
\label{S:background:spatial:presence}

In ``presence-only'' and ``presence-absence'' data \citep{fran09}
the study region is subdivided into cells. 
In each cell, individuals may be either ``present'' (observed to be present
in the cell), ``absent'' (determined to be absent from the cell) or
``undetermined'' (not conclusively determined to be either present
or absent). Presence-absence data record 
``presence'' or ``absence'', for those cells whose status is known.
Presence-only data record ``presence'' or ``non-presence'' (absence or
undetermined status) for every cell in the study region.

Assume the study region $W$ is partitioned into cells
$Q_1, \ldots, Q_J$ of equal area $a$.
For each cell $Q_j$ (or for a subset of cells for which observations
are available) let $y_j$ be the \emph{presence indicator}
that equals 1 if the cell $Q_j$ is observed to contain
any individuals of the species (``presence'')
and 0 otherwise (implying ``absence'' for presence-absence data;
``absence'' or ``unknown'' for presence-only data).
The data consist of the vector of indicator variables
$
  \by = (y_1 \; y_2 \; \ldots \; y_J)^\top.
$

A mapped point pattern $\bx$ can be discretised to obtain presence-absence data.
Let $n_{j}= n(\bx \cap Q_{j})$ denote the number of points of the observed
point pattern $\bx$ that fall inside the $j$th cell.
Then $y_{j}=\indicate{n_{j} >0}$ 
takes the value $1$ if the $j$th cell contains any points of $\bx$,
and $0$ if it does not.

The right panel of Figure~\ref{F:bei} shows indicators of presence or absence of
\textit{Beilschmiedia pendula} trees in a grid of 10-metre-square pixels
superimposed on the previous Figure. In this illustrative example,
the presence-absence data were derived by discretising
the exact coordinate data shown in the left panel of Figure~\ref{F:bei}
as described above.

In ecological applications the exact coordinates are typically not recorded,
and the presence-absence or presence-only data are the original observations;
while in many applications to geological prospectivity,
exact coordinates are indeed discretised
to produce pixel-based presence-absence data for further analysis.

In the right panel of Figure~\ref{F:bei} there are 1753 pixels
in which \textit{B.\ pendula} is present,
while 3604 trees are depicted in the left panel. %

\subsection{Multitype point patterns and case-control data}
\label{S:background:spatial:casecontrol}

A ``multitype'' point pattern is a mapped spatial point pattern in which
each point is labelled as belonging to one of several different types.
Examples include survey maps of trees in a forest labelled by species,
catalogues of galaxies labelled by type or colour,
and microscope images of cell populations labelled by cell type.
In spatial case-control data in epidemiology, there are two types of points
which represent respectively disease cases and (a sample from)
the population at risk \citep{jarndiggchet02,diggetal07}.

Figure~\ref{F:mucosa} shows 
the locations of centres of cell profiles in a histological section of the
gastric mucosa (mucous membrane of the stomach) of a rat.
The window is oriented so that the lower edge is closest to the
stomach wall while the upper edge is interior to the stomach.
The cells are classified into two types: there are 86 \emph{ECL cells}
(enterochromaffin-like cells) and 807 \emph{other} cells.
The intensities of each type of cell are clearly non-uniform.
One hypothesis of interest is whether the ratio of ECL cells
to other cells is constant over the region; equivalently, whether the 
spatially-varying intensities of ECL cells and other cells are
proportional.
An alternative hypothesis is that ECL cells are relatively more abundant
at locations close to the stomach wall.
The data were originally collected by Dr Thomas Berntsen,
and have been discussed and analysed in \citet[pp.\ 2, 169]{mollwaag04}.

\begin{figure}[!hb]
  \centering
  \centerline{
    \includegraphics*[width=0.52\refwidth,bb=15 75 425 355]{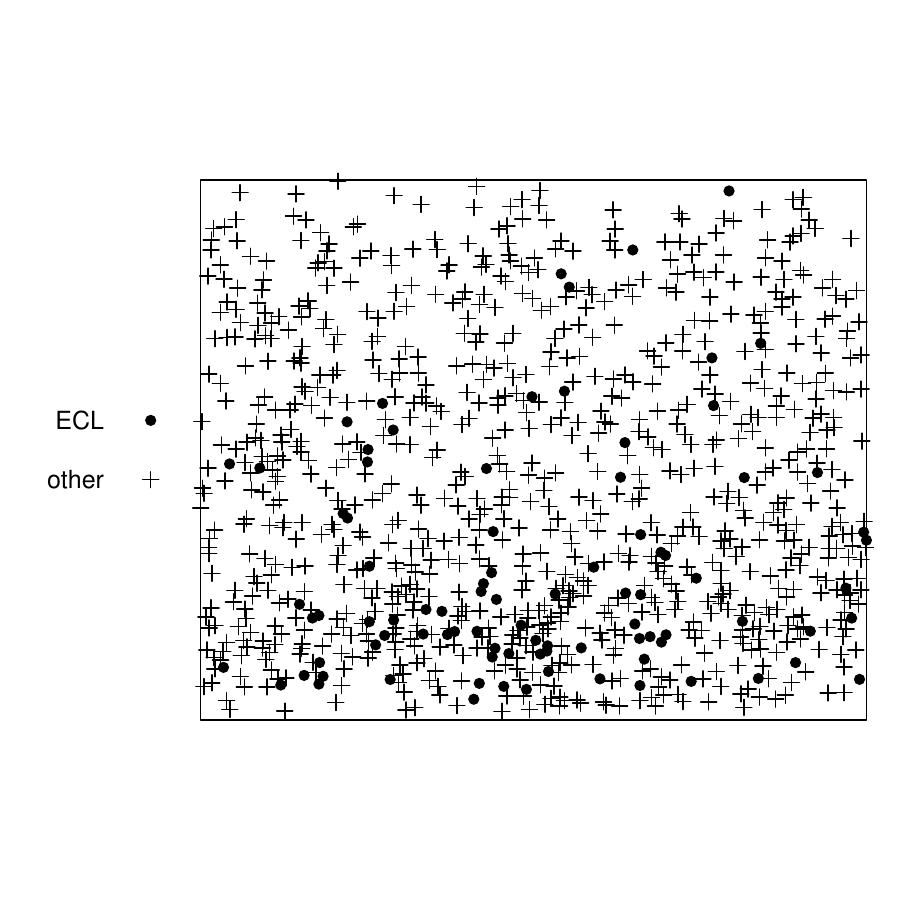}
  }
  \caption{
    Gastric mucosa data, showing locations of enterochromaffin-like cells
    ($\bullet$) and other cells ($+$). Lower edge is closest to stomach wall.
    Collected by Dr Thomas Berntsen.
  }
  \label{F:mucosa}
\end{figure}

\subsection{Spatial covariates}
\label{S:background:spatial:covariates}

A spatial covariate (or ``evidence layer'')
is a real-valued function $Z(u), u \in W$ defined on the study region.
Examples in geological science
include isotope abundance ratio, magnetic field strength, and terrain slope.

Figure~\ref{F:bei.extra} shows contour maps of two spatial covariates,
terrain elevation and terrain slope,
which are available for the survey region of the
\textit{Beilschmiedia pendula} data.

\begin{figure}[!h]
  \centering
\centerline{
    \includegraphics*[width=0.4\refwidth,bb=15 120 385 310]{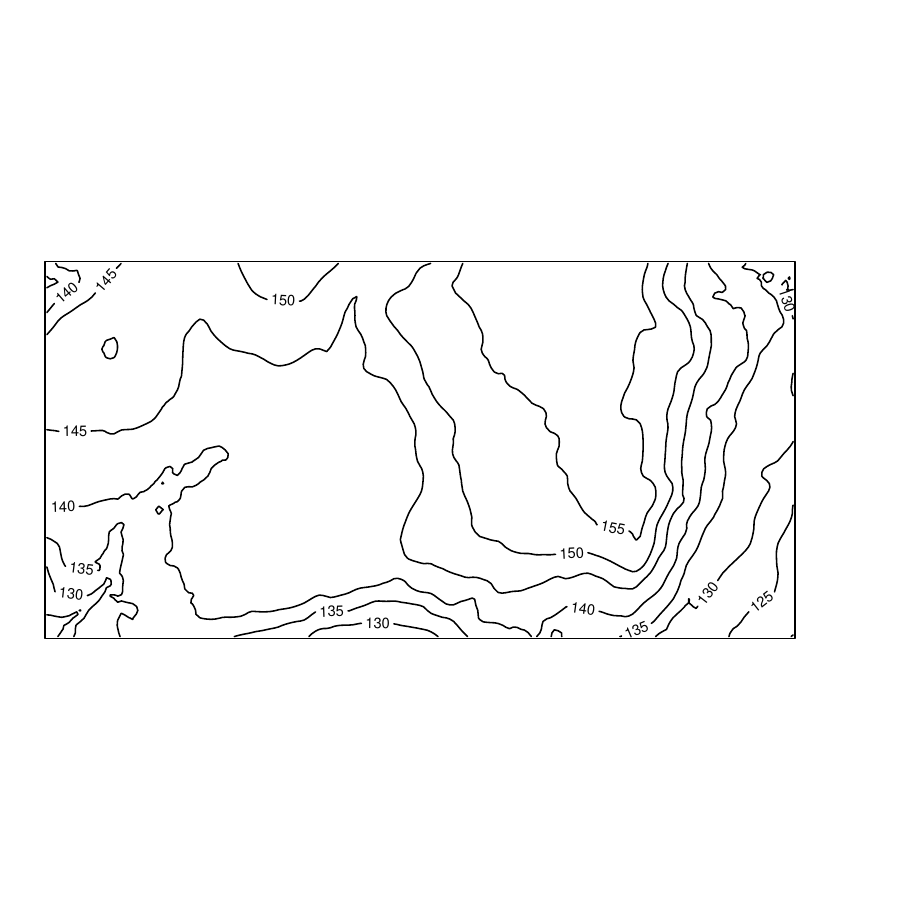}
    \includegraphics*[width=0.4\refwidth,bb=15 120 385 310]{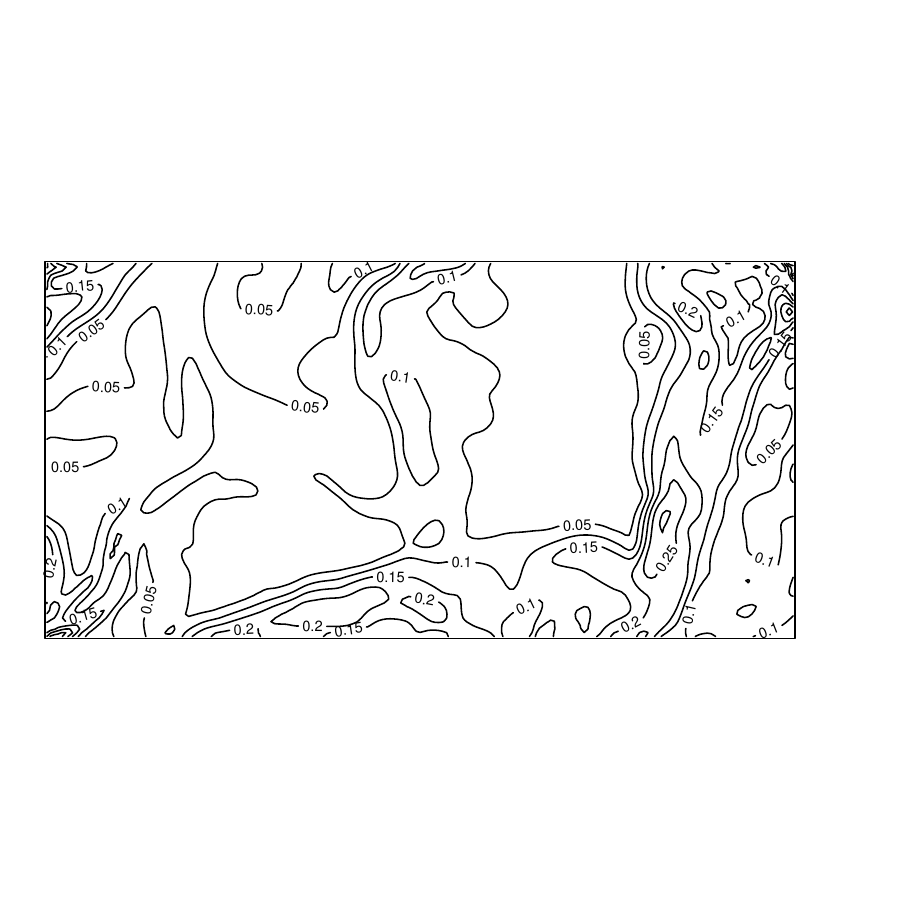}
  }
  \caption{
    Contours of terrain elevation in metres (\emph{Left})
    and contours of terrain slope (\emph{Right})
    for the \textit{Beilschmiedia pendula} data.
  }
  \label{F:bei.extra}
\end{figure}

The left panel of 
Figure~\ref{F:murchison} shows data from a geological survey
of the Murchison region in Western Australia \citep{watkhick90}
first analysed by \citet{knoxgrov97}.
Crosses indicate the locations of known gold deposits.
Lines represent geological faults,
and grey-shaded polygons
represent a particular rock type, greenstone outcrop.
The main aim is to predict the spatially-varying abundance
of gold deposits from the more easily observable fault
pattern and the greenstone outcrop map.
Figure~\ref{F:murchison} strongly suggests that proximity to faults
is predictive for gold prospectivity.

\begin{figure}[!h]
  \centering
  \centerline{
    \includegraphics*[width=0.35\refwidth,bb=40 10 390 420]{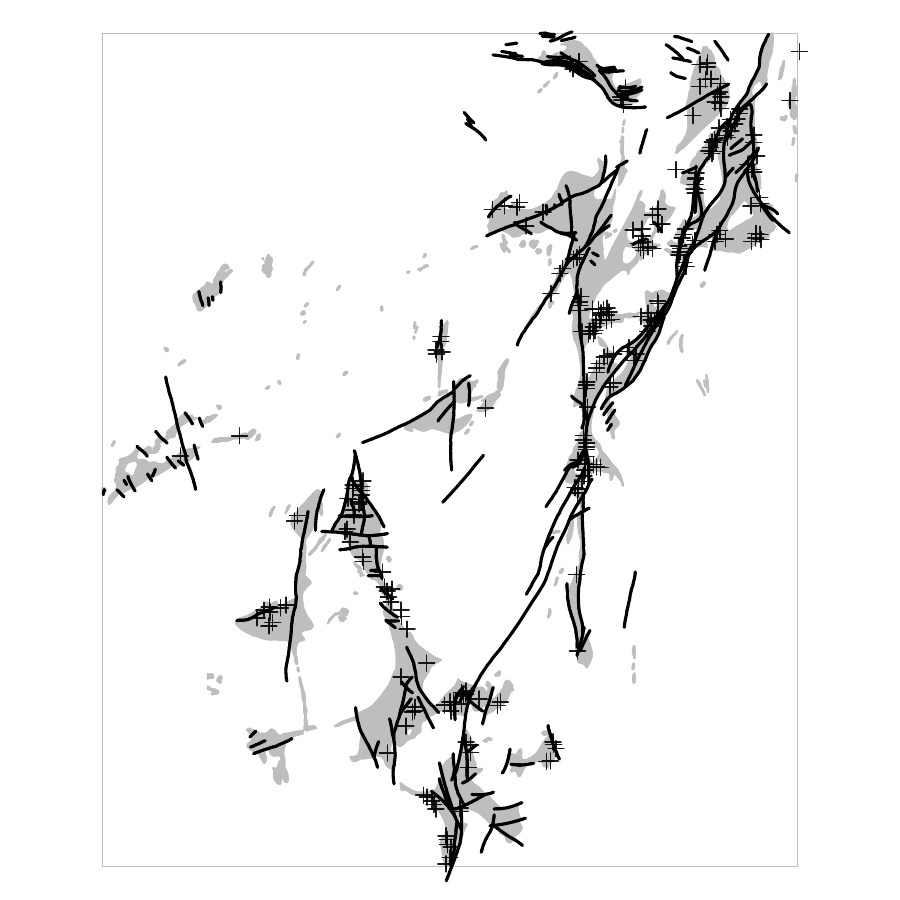}
    \includegraphics*[width=0.34\refwidth,bb=40 20 365 410]{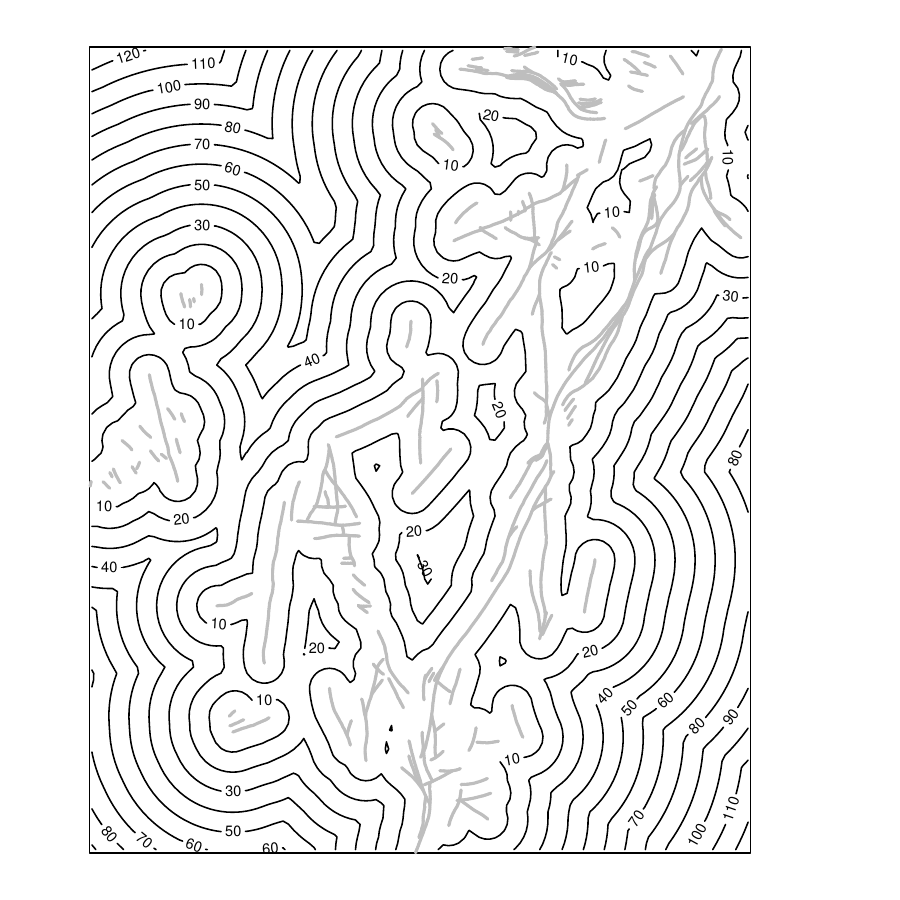}
  }  
  \caption{Murchison data.
    \emph{Left:} gold deposits (+), geological faults (---)
    and greenstone outcrop (grey shading) in a survey region 330 by 400
    kilometres across. Vector (spatial coordinate) data,
    rounded to the nearest metre.
    \emph{Right:}
    Contours of distance (in \textit{km}) to the nearest geological fault
    in the Murchison data. Faults shown as light grey lines.
  }
  \label{F:murchison}
\end{figure}

The covariate information consists of the greenstone
region and the geological faults. For modelling purposes these must be
converted to spatial functions $Z(u), u \in W$. 
For the geological fault lines, a common choice is the
\emph{distance function} $D(u)$ defined for any location $u$
as the shortest distance from $u$ to the nearest geological fault
\citep{knoxgrov97,badd18iamg}.
The right panel of Figure~\ref{F:murchison} shows contours of this function.
For the greenstone polygons, the indicator function defined by
$G(u) = 1$ if location $u$ lies inside the greenstone and $G(u) = 0$ otherwise,
will be used in some of our examples.

\subsection{Models for spatial data which depend on covariates}

\subsubsection{Presence-absence data}

To investigate whether the \textit{B.\ pendula} trees
exhibit a preference for higher terrain elevation (for example),
a standard approach using the presence-absence data
is to fit a logistic regression model
for the probability of presence as a function of terrain elevation,
\begin{equation}
  \label{eq:logistic}
  \log\frac{\pi_j}{1-\pi_j} = \log a + \beta_0 + \beta_1 Z_j
\end{equation}
where, for a given pixel $Q_j$,
the probability of presence of the species is $\pi_j$,
and $Z_j$ is the terrain elevation. Here $\beta_0,\beta_1$ are parameters
to be estimated,
and the offset $\log a$ (where $a$ is pixel area) ensures comparability
of results obtained using different pixel sizes.
This approach was developed independently by
\citet{lewi72survey,tuke72,agte74,bril78,kvam83} and others.
See \citet{badd18iamg} for a recent survey.

A surprising common misconception in the applied literature is that
logistic regression is a nonparametric technique \citep[p.\ 24]{kvam06},
i.e.\ that it does not assume any particular form of the relationship
between $p$ and $Z$.
On the contrary, the logistic regression equation \eqref{eq:logistic}
specifies a linear relationship between the predictor $Z$ and the
log odds of presence.  A truly nonparametric model would take the form
\begin{equation}
  \label{eq:nonparam:presence}
  \frac{\pi_j}{1-\pi_j} = a \rho(Z_j)
\end{equation}
where $\rho(z)$ is an unknown function to be estimated.
This model, and its important connections with ROC curves,
will be discussed in Section~\ref{S:rhohat}.
In real applications, the model typically involves several explanatory
variables, so that in the equations above,
$Z_j$ and $Z(u)$ may be $m$-dimensional vectors
representing the values of $m$ scalar valued covariates.

\subsubsection{Mapped spatial point patterns}

For mapped point patterns, the corresponding approach is to
fit a Poisson point process model in which the point process intensity
is a loglinear function of the covariate. The observed point pattern $\bx$
is assumed to be a realisation of a Poisson point process $\bX$ in $W$,
with intensity (expected number of points per unit area)
\begin{equation}
  \label{eq:loglinear}
  \log\lambda(u) = \beta_0 + \beta_1 Z(u) , \quad\quad u \in W
\end{equation}
at any given spatial location $u \in W$,
where $Z(u)$ is the covariate. 
The close connection between the logistic regression
and the Poisson point process model is explained by
\citet{baddetal10,wartshep10}.
The logistic regression model \eqref{eq:logistic} is asymptotically
equivalent to the loglinear Poisson model \eqref{eq:loglinear}
for small pixel size.

Again, the loglinear relationship \eqref{eq:loglinear} is a simplifying model
assumption, which may be false in a real application.
A nonparametric model would take the form
\begin{equation}
  \label{eq:nonparam:cts}
  \lambda(u) = \rho(Z(u)) , \quad\quad u \in W, 
\end{equation}
where $\rho$ is a function to be estimated. Again we shall find that
$\rho$ is closely related to the ROC curve.

\subsubsection{Interpretation of models}

A properly-formulated species distribution model (SDM) or binary regression or
point process model is a kind of ``law'' expressing
species' preference for particular habitat, etc.
and can conceivably be extrapolated from one study region to another.

In principle, such a model predicts the abundance of data points
at any given location where covariate values are known. However,
in many applications to spatial data,
the main goal of modelling is to identify the conditions
which \emph{maximise} the abundance of data points.
In ecology, attention is often focused on the ``home range''
or ``favorable conditions'' for a species. Models are used
in order to understand species requirements, manage species,
predict change (e.g.\ response to climate change), and assess risks.
Covariates are chosen for their explanatory power in a model
\citep{elitleat09}.
In exploration geology, the goal is to identify the \emph{highest} spatial
concentration of mineral deposits, and there is no interest whatsoever in
modelling the lack of deposits in non-prospective places.
Covariates may be chosen for operational reasons 
rather than for scientific understanding \citep{fordetal19}.

This paper does not make assumptions about the particular models
(such as species distribution models, logistic regressions
and point process models)
or predictive procedures (such as machine learning \citep{Yeomans2018PhD} or
fuzzy predictors \citep{Zhang2015}) that may be in use.
It applies to the use of ROC and AUC for evaluating the performance
of all such models and predictive procedures.

\subsubsection{Spatial case-control data}
\label{S:background:model:casecontrol}

A spatial case-control point pattern can be analysed
with or without conditioning on the observed locations
\citep{digg90,diggrowl94}. 
In the unconditional analysis, the data are treated as a realisation 
of a bivariate spatial point process, requiring a model for the
random spatial locations of the points. The spatial
relative risk (the spatially-varying probability of a case) is determined by
the ratio of the intensities of cases and controls. In the conditional analysis,
the locations $x_i$ are treated as fixed, and the type labels $m_i$
(where $m_i = 1$ denotes a case and $m_i = 0$ denotes a control)
are modelled as a realisation of a binary random field at the sites $x_i$.
The spatially-varying relative risk at a location $x_i$ is the probability that $m_i = 1$.
When the analysis includes covariates, an advantage of the
conditional analysis is that the values of the covariates are required
only at the locations $x_i$. 

In this article we assume that the analysis
of case-control data is conditional on the locations.
Thus, we have $J$
irregularly placed spatial locations $x_j$, $j=1,\dots,J$ where
we model the conditional probability of the point at location $x_j$
being a case, $\pi_j = P(m_j=1)$.
This model may be directly specified
(e.g.\ as having the same form as the pixel based model in
\eqref{eq:logistic}), or it may be derived from an underlying point process model
assumed to have
generated the marked point pattern; see
\citet[pp.\ 358--359]{baddrubaturn15};
\citet{digg90,diggrowl94}.

\section{ROC for a spatial covariate}
\label{S:ROC:covariate}
\label{S:simpsons}

The ROC curve is often regarded
as a tool for assessing the performance of a model.
However, it is possible to construct an ROC curve
without involving a model, using only the raw data
(spatial point pattern or presence-absence data)
and using any spatial covariate as the discriminant variable.
We call this the ``covariate ROC'' (C-ROC) and describe it below.

Many aspects of ROC curves are easier to understand and explain
in this context where the discriminant is a fixed covariate,
than in the more complicated case where the discriminant is
derived from a fitted model.
Many of the conclusions of this section also apply to the ``model ROC''
discussed in subsequent sections.

\subsection{Definition}

Assume there is a single, given, real-valued spatial covariate $Z$,
which may be one of the original explanatory variables, or
may have been constructed from them,
e.g.\ the distance to the nearest fault line in the Murchison example
of Section~\ref{S:background:spatial:covariates}. The objective is to assess
whether $Z$ has utility for predicting the abundance and distribution
of points, without assuming any specific model for the dependence on $Z$.
In applications, the use of a single covariate
may be an initial or interim step in an analysis involving many
explanatory variables.

\subsubsection{Presence-absence data}
\label{ss:rocData.pixel}

In the case of presence-absence data
(Section~\ref{S:background:spatial:presence}),
the data consist of the presence indicators $y_j$
for pixels $Q_j$, $j = 1,\ldots,J$.
For any given threshold $t \in \R$, consider the classifier 
$\hat y_j = \indicate{z_j > t}$.
The empirical C-ROC curve is defined as the 
empirical ROC curve (Section~\ref{S:background:ROC:empirical})
based on these data. The spatial covariate $Z$ plays the role of the
statistic $S$ in equations \eqref{eq:TPhat:finite} and \eqref{eq:FPhat:finite},
and the empirical C-ROC curve is 
the graph of $\TPhat(t)$ against $\FPhat(t)$, where
\begin{equation}
  \label{eq:TPFPhat:covar}
  \TPhat(t)
  =
  \frac{
    \#\{ j: \; y_j = 1, \; z_j > t \}
  }{
    \#\{ j: \; y_j = 1  \}
  }
  ,
  \quad\quad\quad
  \FPhat(t)
  =
  \frac{
    \#\{ j: \; y_j = 0, \; z_j > t \}
  }{
    \#\{ j: \; y_j = 0  \}
  }
\end{equation}
This empirical version of the C-ROC curve will be denoted $R_{Z, \by}(p)$
where $\by = ( y_1,\ldots,y_J)$ is the vector of presence indicators.

This version of the ROC curve is based solely on a spatial covariate,
without reference to a model. It has been used at least once
in mineral prospectivity analysis \citep{goodetal93} but we have
not found other references.

In \eqref{eq:TPFPhat:covar} we note that $\TPhat$ is
computed from the observed data at `presence' pixels 
and $\FPhat$ from `absence' pixels.
This is consistent with the implicit assumption
in Section~\ref{S:background:ROC} that the Positive and Negative populations
are disjoint. 
However, in the spatial context, where the status of each pixel
can be either positive or negative (presence or absence),
it might be more appropriate to replace
$\FPhat(t)$ in \eqref{eq:TPFPhat:covar} by
\begin{equation}
  \label{eq:FPalt:finite}
  \FPalt(t) = \frac{1}{J} \sum_j \hat y_j
  = \frac{1}{J} \sum_j \indicate{z_j > t}
\end{equation}
the fraction of \emph{all} pixels that are predicted to contain a presence.
This is simpler to interpret, since it is the fraction of area
in the survey region where the covariate value $z_j$ exceeds the nominated
threshold $t$.
The discrepancy between $\FPhat(t)$ and $\FPalt(t)$ is small when pixel size
is small. In applications to mineral prospectivity, this approach is useful
because it avoids the difficult task of finding ``true negative'' pixels
\citep{nykaetal15}.

Under suitable conditions 
the empirical true positive rate $\TPhat(t)$ in \eqref{eq:TPFPhat:covar}
is a pointwise consistent estimator of the ``true'' rate of true positives
\begin{eqnarray}
  \label{eq:TP:true:finite}
  \TP(t) &=& \frac{
             \sum\nolimits^{J}_{j=1} \pi_j \ \indicate{z_j > t}
             }{
             \sum\nolimits^{J}_{j=1} \pi_j
             },
\end{eqnarray}
where $\pi_j$ is the (unknown) true probability that pixel $j$ contains
a presence. Similarly,
the empirical false positive rate $\FPhat(t)$ in \eqref{eq:TPFPhat:covar} and
the alternative version $\FPalt(t)$ in \eqref{eq:FPalt:finite} are
both pointwise consistent estimators of the true rate of false positives
\begin{eqnarray}
  \label{eq:FP:true:finite}
  \FP(t) &=&  \frac{
             \sum\nolimits^{J}_{j=1} (1-\pi_j) \ \indicate{z_j > t}
             }{
             \sum\nolimits_{j=1}^{J} (1-\pi_j)
             }.
\end{eqnarray}
A plot of \eqref{eq:TP:true:finite} against \eqref{eq:FP:true:finite}
shall be called the ``true'' or ``theoretical'' C-ROC curve
for the covariate $Z$, and denoted $R_{Z,\bpi}(p)$ where
$\bpi = (\pi_1,\ldots,\pi_J)$ is the vector of presence probabilities.

For a fitted model with estimated presence probabilities $\hat\pi_j$ the ``model-predicted'' C-ROC
curve, $R_{Z,\hat\bpi}$ is defined similarly, with $\bpi$ replaced by $\widehat\bpi$
in \eqref{eq:TP:true:finite} and \eqref{eq:FP:true:finite}.

The calculations above treat every pixel as having equal weight,
and could be inappropriate in applications where pixels do not have equal area.
Examples include geographical projections of observations on the Earth's surface
where the projection does not have unit Jacobian, such as the projection
to latitude-longitude coordinates. Appropriate adjustments to the C-ROC
curve are described in Section~\ref{S:ROC.baseline}.

\subsubsection{Spatial point pattern data}
\label{ss:rocData.cts}

In the case of mapped spatial point pattern data
(Section~\ref{S:background:spatial:point}),
the data consist of the exact coordinates of the observed points,
$\bx = \{x_1,\ldots,x_n\}$.
The continuous-space counterpart of $\TPhat(t)$ in \eqref{eq:TPFPhat:covar} is
\begin{equation}
  \label{eq:TPhat:cts}
    \TPhat(t) = \frac 1 n \sum\limits_{i=1}^{n} \indicate{Z(x_{i}) > t} ,
\end{equation}
the fraction of data points $x_i$
at which the covariate value $Z(x_i)$ exceeds the threshold.
Equivalently, $1-\TPhat(t)$ is the c.d.f.\ of the
covariate value $Z(x_I)$ at a randomly-selected data point $x_I$
(where $I$ is uniformly distributed on $\{1,2,\ldots,n\}$).
The continuous-space counterpart of both $\FPhat(t)$ in \eqref{eq:TPFPhat:covar}
and $\FPalt(t)$ in \eqref{eq:FPalt:finite} is 
\begin{eqnarray}
  \label{eq:FPhat:cts}
  \FPhat(t) = \frac 1 {|W|} \int_W \indicate{Z(u) > t} \dee u,
\end{eqnarray}
the fraction of \emph{area} in the study region where the covariate value
exceeds the threshold.
Here $|W|$ is the area of the study region $W$. Equivalently,
$1-\FPhat(t)$ is the c.d.f.\ of the covariate value $Z(U)$
at a random location $U$ uniformly distributed in $W$.
The graph of $\TPhat(t)$ against $\FPhat(t)$ is the empirical C-ROC curve
$R_{Z,\bx}(p)$.

The continuous space versions \eqref{eq:TPhat:cts}--\eqref{eq:FPhat:cts}
can be derived as the limit of the discrete versions \eqref{eq:TPFPhat:covar}
for the presence-absence case as the pixel size $a$ tends to zero,
under regularity assumptions on the covariate $Z$,
and assuming all points $x_i$ are distinct.
The alternative form \eqref{eq:FPalt:finite} also converges to
\eqref{eq:FPhat:cts}.

The continuous-space version of the C-ROC curve
\eqref{eq:TPhat:cts}--\eqref{eq:FPhat:cts}
resolves issues inherent in the presence-absence case,
about the choice of pixel size, and the loss of information due to
discretisation. It can be used to establish connections
with other methodology for spatial point patterns.

Under suitable conditions, the empirical true positive rate
$\TPhat(t)$ in \eqref{eq:TPhat:cts}
is a consistent estimator of the ``true'' or ``theoretical'' true positive rate
\begin{eqnarray}
  \label{eq:TP:true:cts}
  \TP(t) &=& \frac{
             \int_W \indicate{Z(u) > t} \lambda(u) \dee u
             }{
             \int_W                     \lambda(u) \dee u
             }
\end{eqnarray}
where $\lambda(x)$ is the intensity function of the point process $\bX$.
The empirical false positive rate \eqref{eq:FPhat:cts} is non-random
and can be treated as the ``true'' or ``theoretical'' false positive rate
$\FP(t)$. The graph of $\TP(t)$ against $\FP(t)$ is the ``theoretical''
C-ROC curve $R_{Z,\lambda}(p)$ for the point process $\bX$.

For a point process model with fitted intensity $\widehat\lambda(x)$ the ``model-predicted'' C-ROC
curve, $R_{Z,\widehat\lambda}(p)$ is defined similarly, with $\lambda$ replaced by $\widehat\lambda$
in \eqref{eq:TP:true:cts} (and $\FP(t)$ unchanged).

\subsubsection{Case-control data}
\label{S:ROC.casecontrol}
\label{ss:casecontrol}

For spatial case-control data, it is straightforward to define the C-ROC
of a spatial covariate $S$. 
The point patterns of cases and controls are combined into a
single pattern of points numbered $1, \ldots, J$;
the indicator variable $y_j$ equals $1$ if 
point $j$ is a disease case, and $0$ if it is a control;
then the ROC is based on \eqref{eq:TPhat:finite}--\eqref{eq:FPhat:finite},
effectively treating the cases and controls as samples from the
positive and negative populations, respectively.
The true positive rate \eqref{eq:TPhat:finite} is the fraction of
disease cases where $S > t$; the false positive rate \eqref{eq:FPhat:finite}
is the fraction of controls where $S > t$. The ROC curve is the reverse
P--P plot of the distributions of $S$ in the cases and the controls.
The case-control ROC will be denoted
$R_{S, \bx, \by}(p)$ where
$\bx$ and $\by$ are the point patterns of cases and controls respectively.
The technique can be used to investigate
``relative risk'' in spatial epidemiology \citep{bith91,diggrowl94,hazedavi09}
and ``segregation'' of different types of points \citep{dixo94,diggzhendurr05}.

\subsection{Examples}

Here we compute C-ROC curves for the example datasets
introduced in Section~\ref{S:background:spatial}.

\subsubsection{\textit{Beilschmiedia} trees}

\begin{figure}[!htb]
  \centering
  \centerline{
    \hfill
    \includegraphics*[width=0.3\refwidth,bb=0 0 410 420]{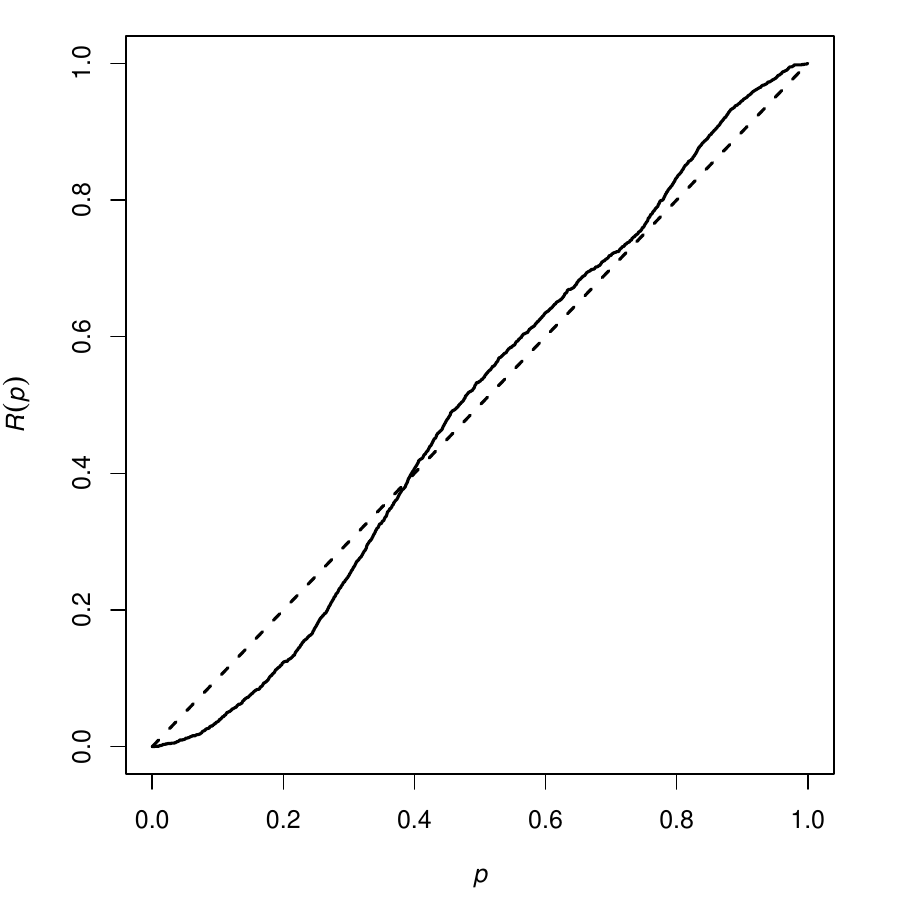}
    \hfill
    \includegraphics*[width=0.3\refwidth,bb=0 0 410 420]{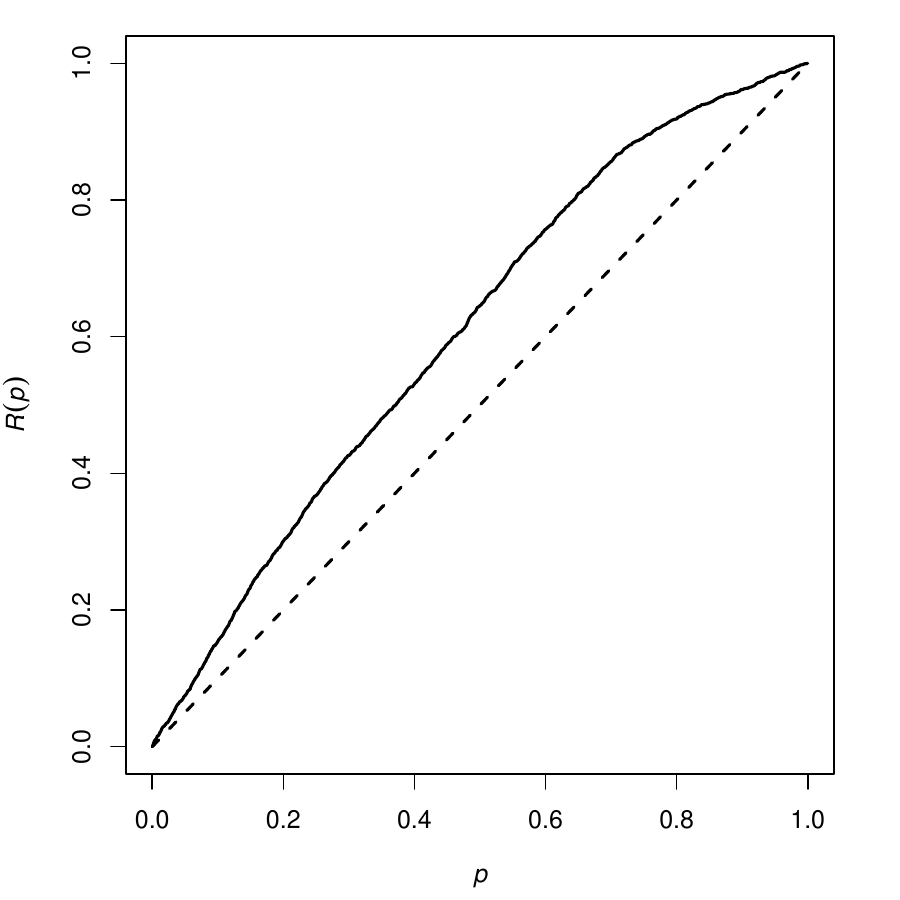}
    \hfill
    \includegraphics*[width=0.3\refwidth,bb=0 0 410 420]{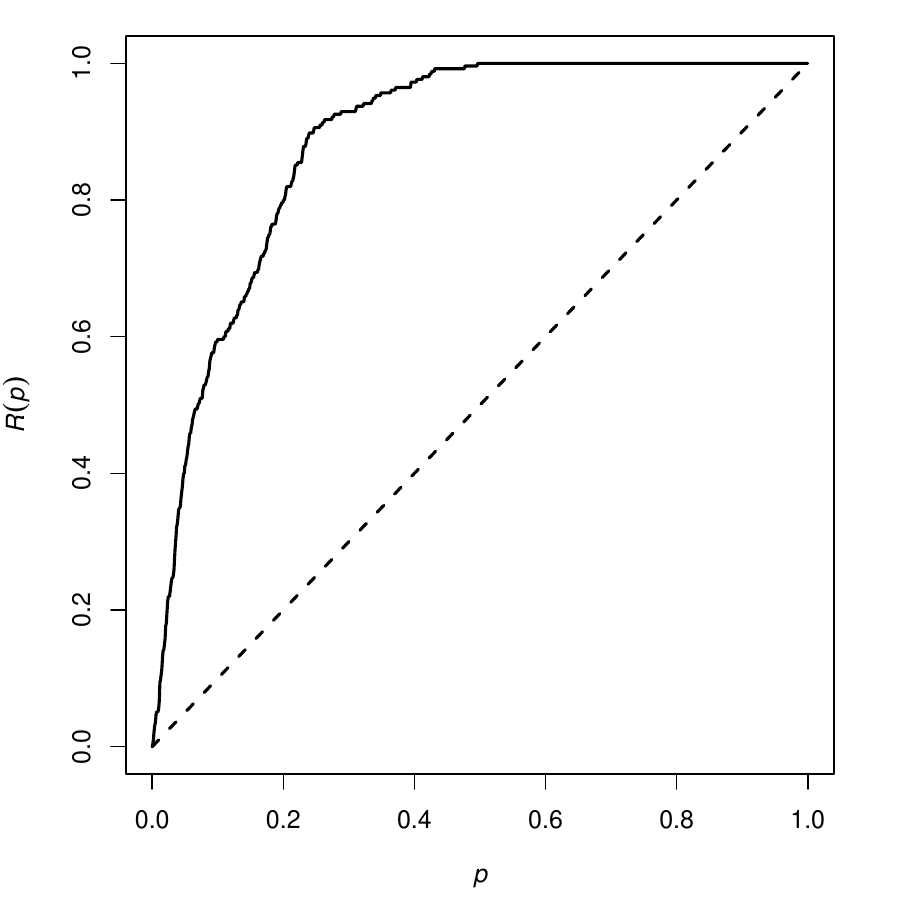}
    \hfill
  }
  \caption{\emph{Left and middle}: Empirical C-ROC curves computed for the \textit{Beilschmiedia} data
    using the presence-absence indicators in 10-metre pixels treating higher values of the
    covariate as favorable to trees where the covariate is either terrain elevation
    (\emph{Left}) or terrain slope (\emph{Middle}). 
    \emph{Right}:
    Empirical C-ROC curve  based on the 
    distance $D(u)$ to the nearest geological fault
    for the Murchison gold data, interpreting short distances
    as more favorable to gold.
  }
  \label{F:bei10ROCmur}
\end{figure}
The left and middle panel of Figure~\ref{F:bei10ROCmur} shows the empirical C-ROC curves
for terrain elevation and terrain slope in the \textit{Beilschmiedia} data
(Figure~\ref{F:bei.extra}),
computed from the presence-absence indicators in
10-metre pixels (right panel of Figure~\ref{F:bei}),
treating higher elevations and steeper slopes as more favorable
to the \textit{Beilschmiedia} trees.
The curve for terrain elevation has an
`S' shape and the computed AUC is 0.51, suggesting a weak
or inconsistent association between elevation and tree abundance.
The curve for terrain slope is a more satisfactory convex curve with
AUC equal to 0.61, suggesting that steeper slopes are
slightly more favorable to trees.
However, using the approach described in Section~\ref{S:rhohat},
  a more accurate explanation is that very flat areas are
  relatively unfavorable, while all other slopes are equally favorable.
  Flat areas include watercourses and swamps which would not be expected to
  support trees. This suggests that there is no evidence of a preference for
  particular terrain slopes.

\subsubsection{Murchison gold data}

The right panel of Figure~\ref{F:bei10ROCmur} shows the empirical C-ROC curve $\lo{\widehat R}_{D, \bx}(p)$ for the
Murchison gold data against distance to the nearest fault,
with short distances treated as favorable to gold. The
curve was computed using the continuous spatial coordinates
according to \eqref{eq:TPhat:cts}--\eqref{eq:FPhat:cts}.
A very similar result is obtained for the discretised presence-absence data.
Effectively, the figure plots the proportion of known gold deposits
within a certain distance from the nearest fault, against the corresponding 
fraction of survey area, for each threshold of distance to the nearest
fault. For example, 10\% of the survey area, $p = 0.1$,
contains about 50\% of known gold deposits, $R(p) = 0.495$.
This curve contains much of the essential information required for
geological prospectivity analysis \citep{badd18iamg}.
The steep initial part of the curve indicates that proximity to the faults
is highly prospective for deposits, in the sense that a large fraction of known
deposits occur in a relatively small fraction of survey area situated within
a relatively short distance from major faults.
However the curve does not show the distance threshold value.

\subsubsection{Gastric mucosa data}

\begin{figure}[!hbt]
  \centering
  \centerline{
    \includegraphics*[width=0.35\refwidth,bb=10 10 390 390]{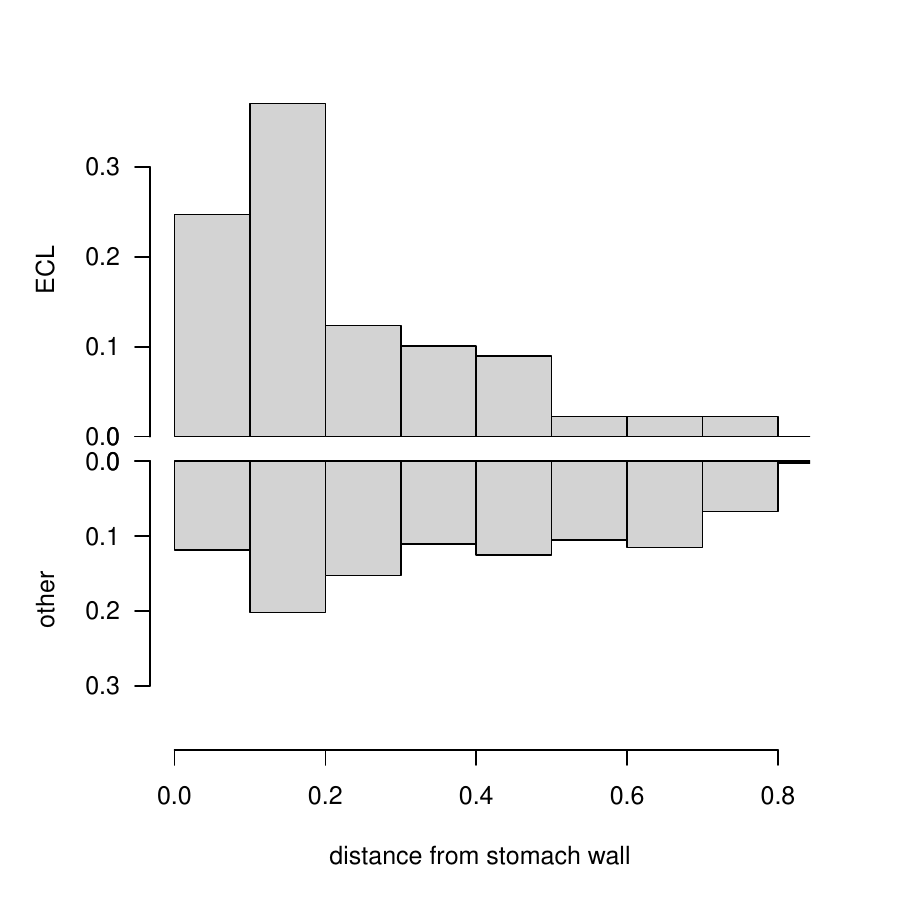}
    \includegraphics*[width=0.35\refwidth,bb=0 0 410 420]{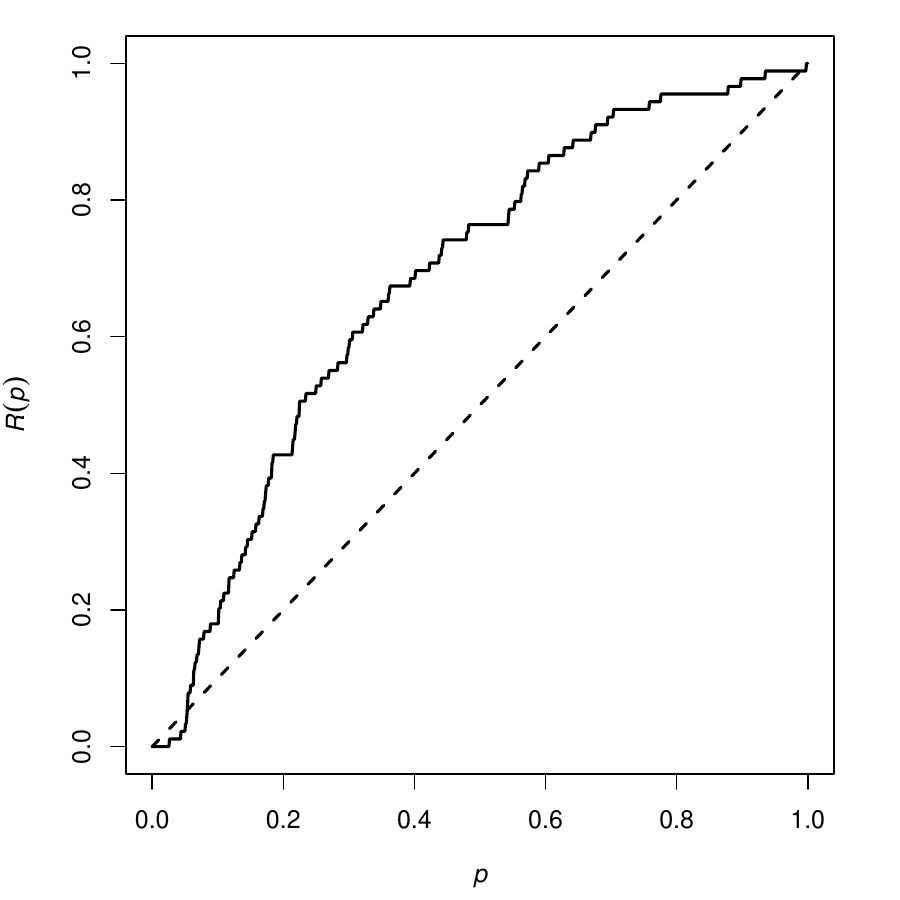}
  }
  \caption{
    Analysis of gastric mucosa data of Figure~\protect{\ref{F:mucosa}}
    using the vertical coordinate $y$
    which represents distance from the stomach wall.
    \emph{Left}:
    normalised histograms of $y$ coordinates of ECL cells (upper histogram)
    and other cells (lower histogram).
    \emph{Right:} C-ROC curve for $y$ coordinate
    for ECL cells relative to other cells,
    treating lower values of $y$ as favorable to ECL cells.
  }
  \label{F:mucosa.roc}
\end{figure}
The right panel of Figure~\ref{F:mucosa.roc} shows the empirical C-ROC curve for
the biological cells example of Figure~\ref{F:mucosa}.
In this example $y_j$ is the class of the $j$'th cell ($y_j=1$ for ECL and $y_j=0$ otherwise), $s_j$ is the vertical
coordinate (distance to the stomach wall) treating short distances as favorable to ECL cells (so
the inequalities in \eqref{eq:TPhat:finite} and \eqref{eq:FPhat:finite} are reversed).
The C-ROC curve lies well above the diagonal which gives strong evidence that
ECL cells have a greater preference for locations near the stomach wall
than do the other cells.
This is also apparent in the left panel of Figure~\ref{F:mucosa.roc} where the distribution of
distances for is much more concentrated at small values for ECL cells than for other cells (more
than 60\% of ECL cells versus
just over 30\% of other cells are within distance 0.2 of stomach wall).

\subsection{Variance}
\label{S:ROC:covariate:variance}

\subsubsection{Mapped point patterns and presence-absence data}

For a mapped point pattern in a study region $W$,
the false positive rate $\FP(t)$ is determined by
the spatial cumulative distribution function of $Z$ over $W$,
and is not subject to random variability. If the point process $\bX$
generating the observed data is a Poisson process (or at least if, conditionally
on the total number of points, the point locations are i.i.d.) then
the conditional variance of $\hat\TP(t)$ for fixed $t$ given the total number
of points $n$ is the binomial variance
\begin{equation}
  \label{e:varTP:mapped}
  \var{\widehat \TP(t) \mid N=n} = \frac 1 n \TP(t) (1 - \TP(t)) .
\end{equation}
Consequently the C-ROC curve $R_{Z,\bX}(p) = \widehat \TP(\FP^{-1}(p))$
has conditional variance
\begin{equation}
  \label{e:varROC:mapped}
  \sigma^2(p) = \var{\widehat R(p) \mid N=n} = \frac 1 n R(p) (1 - R(p)).
\end{equation}
For presence-absence data, if the indicators $y_j$ for different pixels $j$
are independent random variables, then the above
expressions are expected to be the large-sample asymptotic variances, but a formal proof of this is
outside of the scope of this manuscript.

In applications one can use the plug-in estimate
$\widehat\sigma^2(p) = n^{-1}  \widehat R(p) (1 - \widehat R(p))$.

\begin{figure}[!hbt]
  \centering
  \centerline{
    \includegraphics*[width=0.3\refwidth,bb=0 0 410 420]{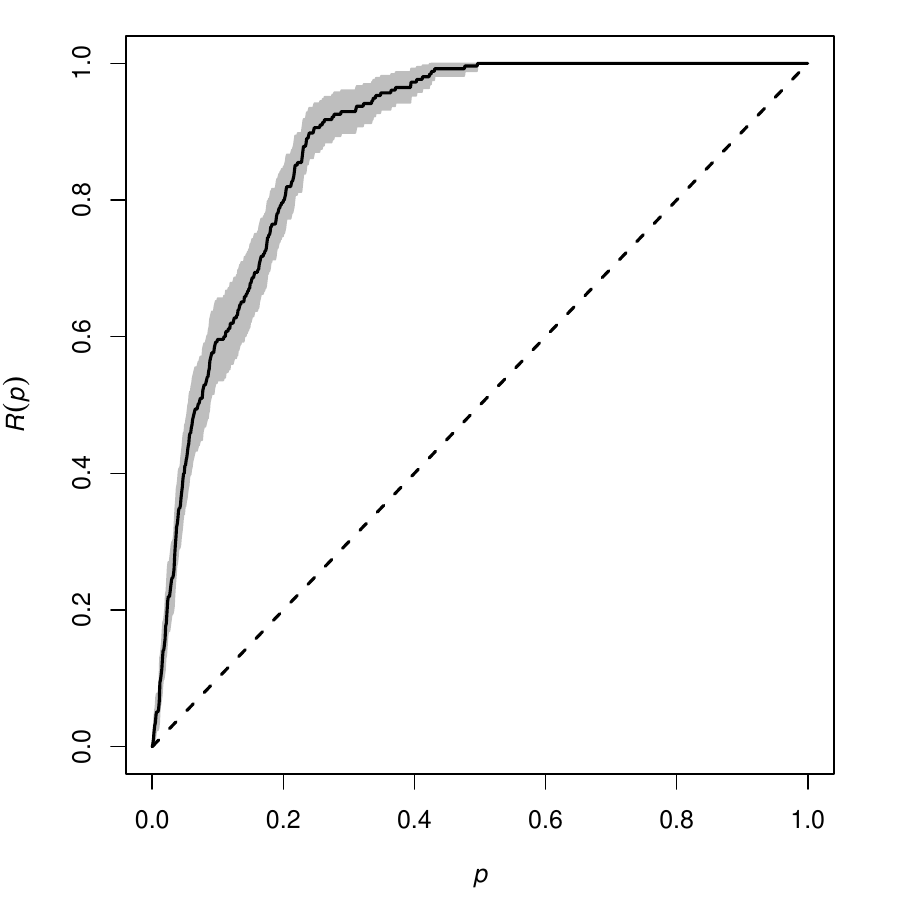}
    \hspace*{4mm}
    \includegraphics*[width=0.3\refwidth,bb=0 0 410 420]{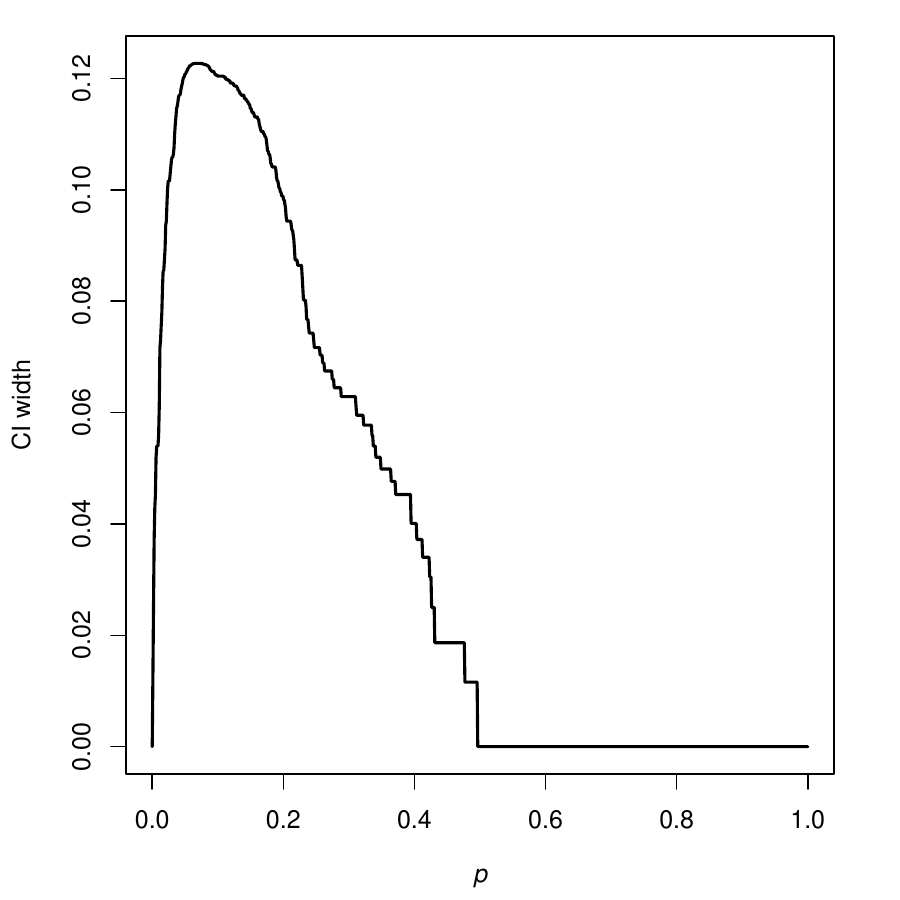}
  }
  \caption{
    \emph{Left:}
    C-ROC curve (calculated by raw method)
    for the Murchison gold data and the distance-to-nearest-fault covariate
    (solid lines) and pointwise approximate $95\%$ confidence bands (grey shading)
    calculated using the plug-in estimate of \eqref{e:varROC:mapped}.
    \emph{Right:} width of confidence bands.
  }
  \label{F:murROCrawCI}
\end{figure}

The left panel of
Figure~\ref{F:murROCrawCI} shows pointwise approximate 95\% confidence bands
for the C-ROC curve for the Murchison gold data using the
distance to nearest fault as the covariate.
Note the optical illusion that the confidence intervals
appear to be extremely narrow; this arises because the C-ROC curve
has a steep slope.
The right panel shows the actual width of the confidence bands
as a function of $p$. 

\subsubsection{Case-control data}
\label{ss:ROC:covariate:casecontrol}

If the cases and controls are mapped spatial point patterns,
assume that they are generated by two independent Poisson point processes.
Then the observed values of $Z$ at the cases and the controls
are i.i.d. samples of size $n$ and $m$ from distributions
with unknown c.d.f's $F$ and $G$ respectively.

Then conditional on the sample sizes and following \cite{hallhyndfan04}
the large-sample asymptotic
variance of the C-ROC curve is 
\begin{equation}
  \label{e:asymvar}
  \sigma^{2}(p) = \var{\widehat{R}(p) \mid N=n, M=m}  = n^{-1} R(p)(1-R(p))
  + m^{-1} \frac{f(G^{-1}(1-p))^{2}}{g(G^{-1}(1-p))^{2}} p(1-p),
\end{equation}
where $f$ and $g$, respectively, are the densities corresponding to $F$ and $G$.

\citet{hallhyndfan04} proposed estimating $f$ and $g$ by kernel smoothing,
and integrating these function estimates to obtain estimates
of $F$ and $G$. This yields a smooth estimate of the C-ROC curve,
and also a plug-in estimate of the variance using \eqref{e:asymvar},
as detailed below.

Let $\kappa$ be the template smoothing kernel,
a probability density function on the
real line, and let $\mathcal{K}$ be the corresponding cumulative distribution
function. Then the smoothed empirical estimators are
\begin{equation}\label{eq:cumest}
  \tilde{F}(t) =
  \frac{1}{n h_1}
  \sum\limits^{n}_{i=1}
  \mathcal{K}\left( \frac{t - Z(x_{i})}{h_{1}} \right),
  \quad\quad 
  \tilde{G}(t) =
  \frac {1}{m h_2}
  \sum\limits^{m}_{i=1} 
  \mathcal{K}\left( \frac{t - Z(x_{n+i})}{h_{2}} \right),
\end{equation} 
where $h_{1}$ and $h_{2}$ are bandwidths.
Methods for bandwidth selection are discussed by
\citet{altmlege95,lloy98}.
The smooth estimate of the C-ROC curve is 
\begin{equation}
  \label{e:roc.smooth}
  \tilde{R}(p) = 1 - \tilde{F}_{Z}(\tilde{G}^{-1}(1-p)).
\end{equation}
To obtain the estimate of $\sigma^2(p)$, the squared gradient term
$f(G^{-1}(1-p))^{2}/g(G^{-1}(1-p))^{2}$ is computed using 
kernel estimates $\tilde f$ and $\tilde g$ with bandwidths $h_f$ and $h_g$ which may be
different from the bandwidths $h_1, h_2$ used above.
The plug-in estimator of  $\sigma^2(p)$ using these smooth estimates is
\begin{equation}
  \label{e:asymvar:est}
  \widetilde\sigma^{2}(p) =
  n^{-1} \widetilde R(p)(1- \widetilde R(p))
  + m^{-1} \frac{\widetilde f(\widetilde G^{-1}(1-p))^{2}}{\widetilde g(\widetilde G^{-1}(1-p))^{2}} p(1-p).
\end{equation}
Asymptotic $(1-\alpha)$-level confidence bands are given by
$\tilde{R}(p) \pm z_{\alpha/2}\widetilde{\sigma}(p)$, where $z_{\alpha}$ is
the upper $1-\alpha$ point of the standard normal distribution. The
coverage of the confidence bands crucially depends on the bandwidths
used in calculating $\widetilde{R}$ and $\widetilde{\sigma}^2$. The optimal
bandwidth for $\widetilde{R}(p)$ is described in \citet{hallhynd03}, while
\citet{hallhyndfan04} make detailed recommendations on bandwidth
selection for computing $\widetilde{\sigma}^2$. 

\begin{figure}[!hb]
  \centering
  \centerline{
    \includegraphics*[width=0.4\refwidth,bb=0 0 410 420]{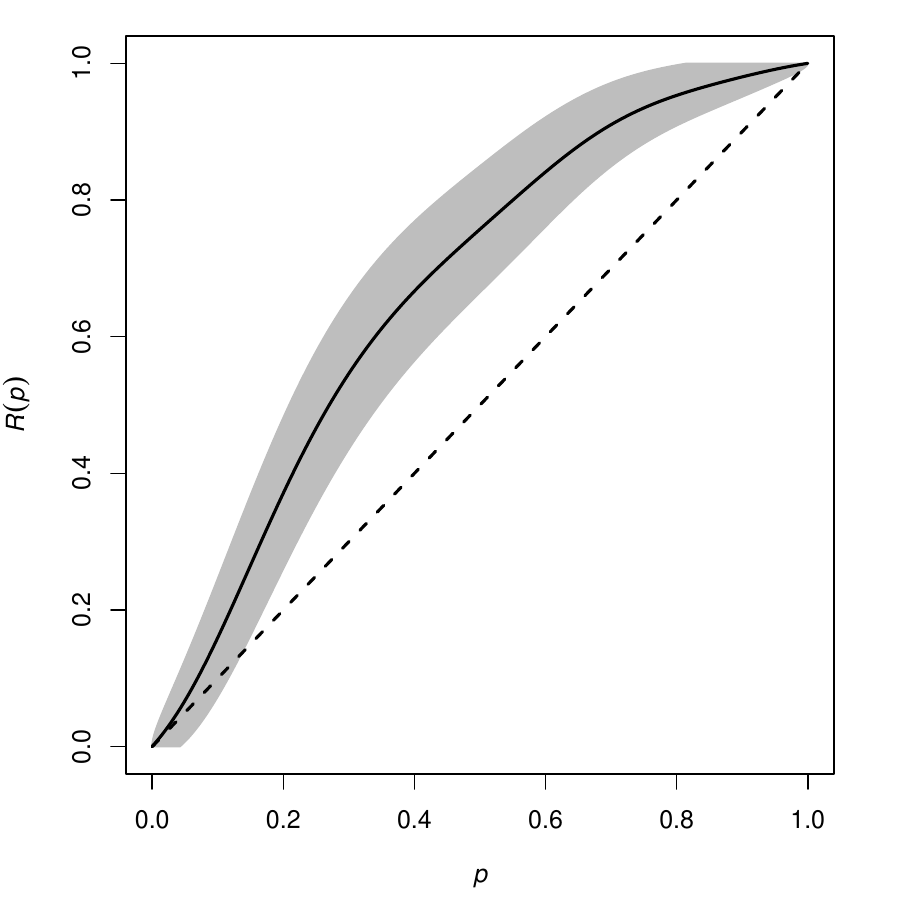}
  }
  \caption{
    Kernel smoothing estimate (solid line) and
    pointwise $95\%$ confidence bands (grey shading)
    calculated using \eqref{e:asymvar:est},
    for the C-ROC curve
    for the mucosa case-control data against vertical coordinate.
  }
  \label{F:mucosaROCsmoothCI}
\end{figure}

Figure~\ref{F:mucosaROCsmoothCI} shows a smoothed estimate of the C-ROC curve
for the mucosa data (treating ECL cells as cases, and other cells as controls)
using the $y$ coordinate as the covariate. Grey shading shows
pointwise $95\%$ confidence bands for the true C-ROC curve
based on the estimate of asymptotic variance \eqref{e:asymvar:est} and truncated to $[0,1]$.
The bandwidths $h_1, h_2$ were determined by 
Silverman's rule of thumb \citep[eq.\ (3.31), p.\ 48]{silv86}
applied separately to the $y$-coordinates of the ECL cells and the other cells,
respectively.

\subsection{Interpretation of C-ROC and C-AUC}

This section discusses the interpretation of the ROC and AUC
for spatial point pattern data based on a spatial covariate $Z$,
as defined above. Similar comments apply to the cases of
presence-absence and spatial case-control data.

\subsubsection{Invariance under rescaling and transformation}
\label{S:monotone}

\begin{lemma}
  \label{L:monotone}
  Let $Z$ be a spatial covariate on a domain $W$,
  and $h:\R\to\R$ a strictly increasing transformation.
  Define the covariate $Z^\dag$ by $Z^\dag(u) = h(Z(u))$ for $u \in W$.
  Then the covariates $Z$ and $Z^\dag$ have the same C-ROC curve, that is
  $R_{Z,\bx}(p) \equiv R_{Z^\dag, \bx}(p)$.
\end{lemma}

\begin{proof}
  The ROC curve for $Z^\dag$ is a plot of $\TP^\dag(t)$ against
  $\FP^\dag(t)$ for all $t \in \R$, where
  from \eqref{eq:TPhat:cts} and \eqref{eq:FPhat:cts}
  \[
    \TP^\dag(t) = \frac 1 n \sum_{i=1}^n \indicate{Z^\dag(x_i) > t},
    \quad\quad
    \FP^\dag(t) = \frac 1 {|W|} \int_W \indicate{Z^\dag(u) > t} \dee u.
  \]
  For any $t\in\R$ and $u \in W$ we have
  $
  Z^\dag(u) > t \mbox{  iff  } Z(u) > h^{-1}(t) = s, %
  $
  say,   
  because $h$ is 1--1. Hence $\TP^\dag(t) = \TP(s)$ and $\FP^\dag(t) = \FP(s)$.
  Hence the ROC curves
  $R_{Z,\bx}(p)$ and $R_{Z^\dag, \bx}(p)$
  are identical.
\end{proof}

In particular Lemma~\ref{L:monotone}
applies to numerical rescaling of the covariate:
$R_{Z,\bx}(p) = R_{aZ+b, \bx}(p)$ for any
scale factor $a > 0$ and any constant $b$.

\begin{lemma}
  \label{L:areapreserving}
  The C-ROC is invariant under rescaling of the spatial coordinates,
  and under any area-preserving transformation of the spatial domain.
\end{lemma}

\begin{proof}
  More generally, let $\psi: W \to W^\ast$ be a 1--1 transformation such that
  $|\psi(B)| = c \; |B|$ for any measurable subset $B \subseteq W$,
  where $c$ is a constant. Here $\psi(B)$ denotes the
  image of $B$ under the transformation,
  $\psi(B) = \{ \psi(x): \, x \in B\}$.
  
  The data points $x_i$ in $W$ are transformed to
  data points $y_i = \psi(x_i)$ in $W^\ast$, and
  the covariate $Z$ on $W$ is transformed to $Z^\ast$ on $W^\ast$
  defined by
  \[
    Z^\ast(v) = Z(\psi^{-1}(v)), \quad\quad v \in W^\ast.
  \]
  Consider the C-ROC curve based on 
  \[
    \TP^\ast(t) = \frac 1 n \sum_{i=1}^n \indicate{Z^\ast(y_i) > t},
    \quad\quad
    \FP^\ast(t) = \frac 1 {|W^\ast|} \int_{W^\ast} \indicate{Z^\ast(v) > t} \dee v.
  \]
  We have $Z^\ast(y_i) = Z(\psi^{-1}(y_i)) = Z(x_i)$ so that $\TP^\ast \equiv \TP$.
  Since $\psi$ is area-preserving, $\FP^\ast \equiv \FP$. Hence
  the C-ROC curves are identical.
\end{proof}

Note that the C-ROC is \emph{not} invariant
under a geographic projection (change of spatial coordinates)
unless it is area-preserving.
In the discrete case, calculating the C-ROC using a grid of square pixels
in latitude-longitude coordinate space would give greater weight to
locations at higher latitudes; the calculation should be performed using
a square grid in a conformal projection.

\subsubsection{Ranking ability}
\label{S:ROCcovar:interpret}

ROC and AUC are often claimed to be measures of
``predictive power'' \citep{lobojimereal07,aust07,fran09}.
However, it is clear that the ROC and AUC based on a covariate $Z$
do not evaluate predictive ability in the conventional sense,
because there is no predictive model in this context. A predictive model
would give predictions of the presence probabilities at each pixel;
a measure of predictive ability would indicate how well
the observed pattern of presences and absences
followed these predicted probabilities.

Rather, the C-ROC and C-AUC based on a covariate $Z$
are measures of the ``ranking ability'' of $Z$,
which is the extent to which presence and absence
pixels in the study region would be segregated from each other, if pixels
were sorted by increasing order of $Z$.
Each point on the ROC curve represents a subdivision of the study region
into two subregions of low and high values of $Z$,
and measures the extent to which this subdivision separates space
into regions of low and high density of points.

Since C-ROC and C-AUC are invariant under monotone transformations of the
covariate, we could say that C-ROC and C-AUC measure the ``ranking power''
of \textbf{any} model $\pi_j = \rho(z_j)$ where $\rho(z)$ is
a monotone increasing function of $z$.
When there is such an underlying model the ROC curve in the right panel of Figure~\ref{F:bei10ROCmur} is also known as the
capture-efficiency curve as mentioned in Section~\ref{S:ROC:model:defn:pixel}.

In practice, the highest density of points could occur at
locations where the covariate $Z$ takes an intermediate value
rather than a low or high value. For example, trees may be most abundant
in locations which have moderate amounts of water, being neither arid nor
swampy. The covariate would then have high value as a predictor in an
appropriate model, but would not have high ``ranking ability''
as measured by C-ROC and C-AUC.
In such a case a more appropriate approach would be to fit a model for the occurrence of points and
then use the model ROC described in Section~\ref{S:ROC:model}.

\subsubsection{Statistical interpretations}

The C-ROC curve is nonparametric in nature; it
does not assume any parametric form of dependence of the point pattern
on the spatial covariate. The only implicit assumption is that larger values of
the covariate are more favorable to the presence of points.

By construction, the C-ROC curve is insensitive to extremes,
and to small sub-populations.
In \ifelseArXiV{Appendix~\ref{SUPP:chorley}}{the supplementary material (section S4)}
we show an example,
using the famous Chorley-Ribble cancer data,
in which the ROC and AUC do not suggest any deviation from the diagonal line,
but the effect of the covariate is statistically significant according to
the likelihood ratio test. This can occur when the effect of the covariate
only applies to a small fraction of the population.

As noted above, the C-ROC curve is a comparison of two distributions,
namely, the distribution of the value of the covariate at a typical point
of the pattern, and the distribution of the value of the covariate at a
uniformly-randomly-selected point in the window.

\begin{lemma}
  \label{L:CDF}
  For a spatial point pattern $\bx = \{ x_1,\ldots,x_n\}$ in a study region $W$,
  and a spatial covariate $Z(u)$ defined for all locations $u \in W$,
  let $\TPhat$ and $\FPhat$ be defined as in \eqref{eq:TPhat:cts} and \eqref{eq:FPhat:cts}.
  Then
  \begin{enumerate}
  \item 
    $F(t) = 1 -\TPhat(t)$ is the cumulative distribution function of
    $Z(x_I)$ where  $x_I$ is a uniformly-randomly-selected data point (i.e.\ where $I$
    is a uniformly random integer between 1 and $n$).
  \item 
    $G(t) = 1 -\FPhat(t)$ is the cumulative distribution function of
    $Z(U)$ where $U$ is a uniformly-randomly-selected location in the window.
  \item
    the C-ROC curve is a reversed P--P plot comparing the distribution of
    $Z(x_I)$ to the distribution of $Z(U)$.
  \end{enumerate}
Similarly, given presence indicators $y_j$ and covariate values $z_j$ for $j = 1,\ldots,J$,
let $\TPhat$ and $\FPalt$ be defined as in \eqref{eq:TPFPhat:covar} and \eqref{eq:FPalt:finite}.
Then 
\begin{enumerate}
  \setcounter{enumi}{3}
  \item 
    $F(t) = 1 -\TPhat(t)$ is the cumulative distribution function of
    the covariate value at a randomly selected `presence' pixel, that is,
    the cdf of $z_I$ where $I$ is a random integer selected with uniform probability
    from the set $\{j: y_j = 1\}$;
  \item 
    $G(t) = 1 -\FPalt(t)$ is the cumulative distribution function of
    the covariate value at a randomly selected pixel, that is, the cdf of
    $z_K$ where $K$ is a random integer uniformly selected from 1 to $J$;
  \item
    the C-ROC curve is a reversed P--P plot comparing the distribution of
    $z_I$ to the distribution of $z_K$.
  \end{enumerate}
\end{lemma}

\begin{proof}
  Let $I$ be a random integer between $1$ and $n$, with equal probability
  for each possible outcome. Then the cumulative distribution function
  of $Z(x_I)$ is
  \[
    F(t) = \Prob{Z(x_I) \le t} 
    = \EE \indicate{Z(x_I) \le t} 
    = \sum_{i=1}^n \Prob{I=i} \indicate{Z(x_i) \le t} 
    = \frac 1 n \sum_{i=1}^n \indicate{Z(x_i) \le t},
  \]
  and statement 1 follows.
  For statement 2, a uniformly random point $U$ in $W$ has
  probability density $f(u) = \indicate{u \in W}/|W|$ in $\R^2$ so that
  the cdf of $Z(U)$ is
  \[
    G(t) = \Prob{Z(U) \le t} 
    = \EE \indicate{Z(U) \le t} 
    = \int_{\R^2} f(u) \indicate{Z(u) \le t} \dee u 
    = \frac 1 {|W|} \int_W \indicate{Z(u) \le t} \dee u,
  \]
  and statement 2 follows.
  Statement 3 is a trivial consequence.
  Statements 4--6 are proved in the same way.
\end{proof}

\begin{lemma}
$\AUC$ can be interpreted according to \eqref{eq:AUC=prob}
as the probability that a typical point of the point pattern has
a higher value of the covariate than a
uniformly-randomly-selected location in the window.
That is, $\AUC = \Prob{Z(x_I) > Z(U)}$,
where $x_I$ is a randomly-selected data point
($I = 1,2,\ldots,n$ with equal probability) and
$U$ is a random point uniformly distributed in the window $W$.
\end{lemma}

\begin{proof}
  Let $R(p)$ be the C-ROC curve, that is, $R(p) = \TP(\FP^{-1}(p))$.
  Then
  \[
    \AUC = \int_0^1 R(p) \dee p 
         = \int_0^1 \TP(\FP^{-1}(p)) \dee p 
         = \int_{-\infty}^\infty \TP(t) \dee \FP(t) .
  \]
  Invoking statements 1 and 2 of Lemma~\ref{L:CDF},
  \[
        \AUC = \int_0^1 \Prob{Z(x_I) > t} \dee \FP(t) 
             = \Prob{Z(x_I) > Z(U)} ,
  \]
  where $I$ and $U$ are as described in Lemma~\ref{L:CDF}.
\end{proof}

The C-ROC curve can also be expressed in terms of the
Probability Integral Transformation.

\begin{lemma}
  The reverse C-ROC curve is the empirical cumulative distribution function
  of the values $v_i = \FP(z_i)$, where $z_i = Z(x_i)$ for $i=1,\ldots,n$.
\end{lemma}

For the proof we notice that
\[
    R(p) = \TP(\FP^{-1}(p)) 
    = \frac 1 n \sum_{i=1}^n \indicate{Z(x_i) > \FP^{-1}(p)} 
    = \frac 1 n \sum_{i=1}^n \indicate{\FP(Z(x_i)) > p} 
    = \frac 1 n \sum_{i=1}^n \indicate{v_i > p} .
\]
Similarly
$
    \lo R(p) = \frac 1 n \sum_{i=1}^n \indicate{v_i \le p} .
$

\subsection{Weaknesses}
\label{S:ROCcovar:weaknesses}

\subsubsection{Dependence on study region}
\label{S:ROCcovar:region}

The ROC based on a covariate has the very important weakness
that it depends on the choice of study region.
(Similar warnings have been issued in the literature \citep{jime12}
about the model ROC discussed in Section~\ref{S:ROC:model} below.)

The horizontal axis of the C-ROC plot is the fraction of area
\emph{in the study region} where the covariate exceeds a particular
threshold. For example in the \textit{Beilschmiedia} data the horizontal
axis is the fraction of area of this particular rectangle, with its
particular topography, where the terrain elevation or terrain slope exceeds
a particular threshold. In the Murchison data, the horizontal axis is
the fraction of area of the mapped rectangle, not the exploration lease, nor
the prospective province. 

This crucial fact is mainly attributable to inhomogeneity.
The C-ROC expresses the ability of the covariate
to segregate \emph{the study region} into
subregions of relatively high and low intensity of points.
This ranking ability can
only be high if there is substantial spatial variation in the
intensity of points within the study region (the point process is
``inhomogeneous''), and if there is spatial variation
in the covariate values within the study region, and furthermore if
high covariate values are associated with high intensity of points
within the study region (or high covariate values are associated with
low intensity of points).

To put it another way, the only case in which the C-ROC for a subregion
is guaranteed to be the same as the C-ROC for the entire study region,
is the case where the covariate has no effect and the C-ROC curve is
the diagonal.

The C-ROC is completely bound to the choice of study region.
Unlike a species distribution model, point process model or other
statistical model, the C-ROC does not represent a scientific ``law''
that can be extrapolated from one study region to another,
or considered to have validity independent of the choice of study region.
Findings based on ROC/AUC cannot be generalised/extended to
other contexts; they are completely bound to the original study region.
In the \textit{Beilschmiedia} data, the C-ROC for the terrain slope covariate
does not express the intrinsic affinity or preference of individual trees
for steep terrain.
In spatial ecology we cannot use ROC analysis to
predict the response of a species to climate change or habitat loss
(or in general, any cases where the covariate values change over time)
\emph{even within the same spatial region}.
In the Murchison data, the C-ROC for distance to nearest
fault does not reveal a set of physical conditions which are always favorable to
the presence of gold deposits. 
In exploration geology we cannot use ROC analysis to identify
the best predictor variables for gold, except in the context of the
original survey region (and not any subregion).

The fact that the C-ROC is bound to a particular study region
only becomes an advantage when the objective is to segregate
this specific region, for example, when it is desired to
identify the most promising locations for mine exploration within
an area that is available for mining, or when it is desired to find
a location for a wind turbine which will minimise harm to wild bird life.

\subsubsection{Restriction to a more homogeneous subset}
\label{S:subset:homogen}

In many real examples, the AUC decreases when we restrict attention
to a sub-region of the original study region.
This typically occurs when the density of points is more homogeneous
in the sub-region.

\begin{figure}[!h]
  \centering
  \centerline{\includegraphics*[width=0.4\refwidth,bb=0 0 410 420]{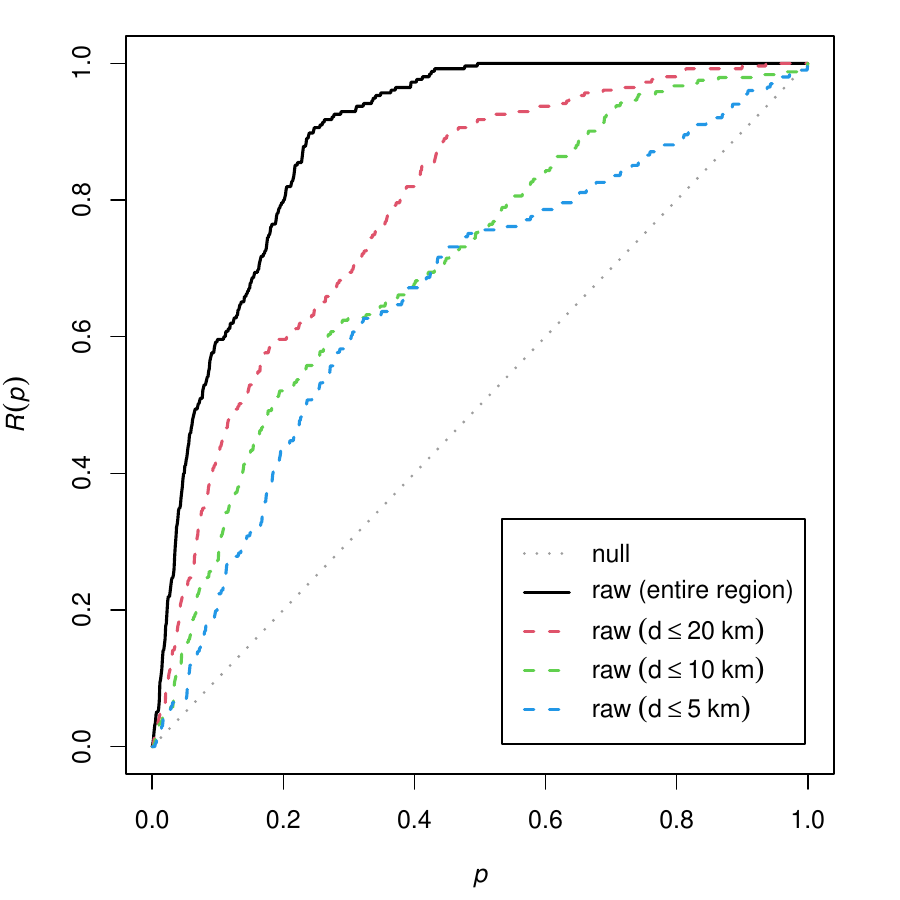}}
  \caption{
    Effect of restriction to a subregion.
    C-ROC curves for the distance-to-nearest fault covariate
    in the Murchison data, computed for the entire study region (solid lines)
    and for regions where the distance to nearest fault is at most 20,
    10 or 5 km (dashed lines). Corresponding AUC values are
    $0.89$, $0.79$, $0.71$ and $0.66$.
  }
  \label{F:murdist:restrict}
\end{figure}

In the Murchison data, the original study region is the
bounding rectangle of the map data. For the distance-to-nearest-fault
covariate, we obtain $\AUC = 0.89$ in this rectangle. The maximum
covariate value (distance to nearest fault) in this rectangle is 128.9 km,
while the maximum covariate value at a gold deposit point is 17.9 km. If we
restrict attention to the region lying at most 20 km from the nearest
fault, the AUC value falls to $0.79$. This value drops further if we consider
regions closer to the faults ($0.71$ for 10 km, $0.66$ for 5 km).
Figure~\ref{F:murdist:restrict} shows
the C-ROC curves obtained by restricting attention to
progressively shorter distances from the faults.

A synthetic example of the same phenomenon is shown in
\ifelseArXiV{Appendix~\ref{SUPP:egeStripExample}}{the supplementary material (section S5)}.

\subsubsection{Simpson's Paradox}

Since the C-ROC and C-AUC are completely bound to the choice of
study region, we can expect paradoxes to arise when the study region
is subdivided. These can be seen as instances of 
Simpson's Paradox \citep{simp51,yule1903}.

Figure~\ref{F:simpson:data} shows a synthetic example. The unit square is
divided into two subregions separated by the dashed lines. The 
intensity of points is uniform in each subregion, but different between
the two subregions. In Figure~\ref{F:simpson:effect} the left panel shows
the C-ROC curve for the $x$ coordinate, calculated for the entire dataset,
with an AUC value of 0.73, indicating that the $x$ coordinate has some
predictive ability to separate the data into high- and low-intensity
regions. However the middle and right panels show the C-ROC curves for the $x$ coordinate
calculated in the two subregions, indicating that within these subregions
the $x$ coordinate has no predictive value at all
(AUC values 0.47 and 0.50).

\begin{figure}[!h]
  \centering
  \centerline{\includegraphics*[width=0.4\refwidth,bb=10 110 420 320]{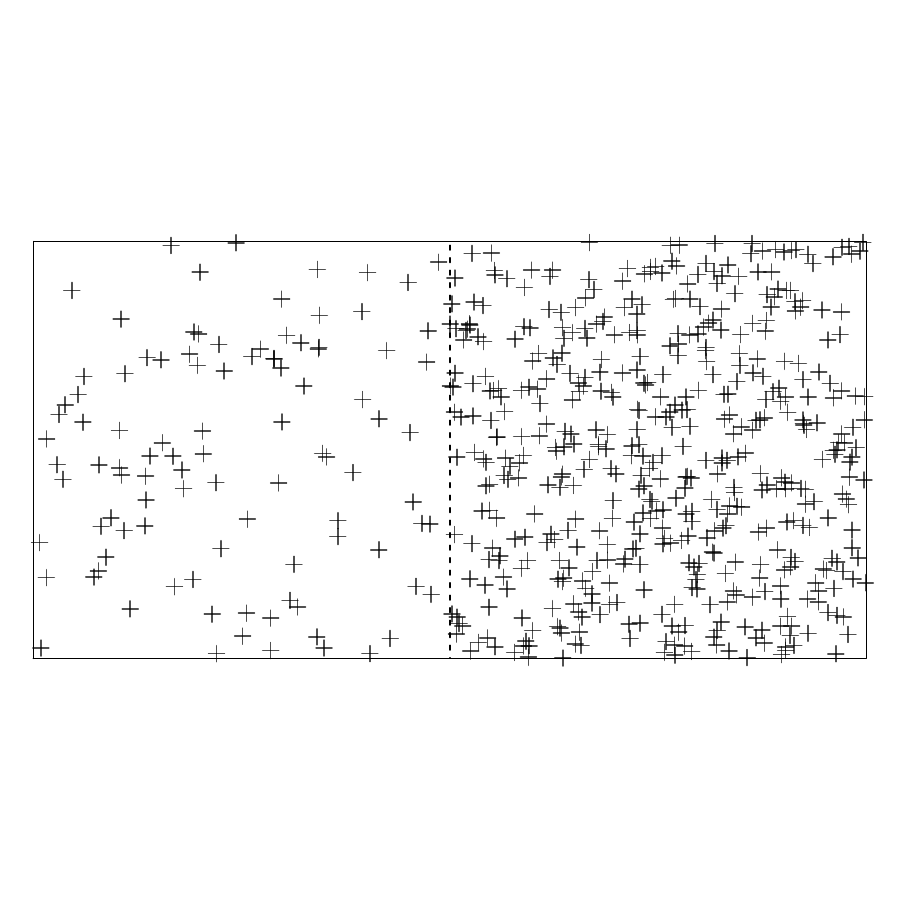}}
  \caption{Synthetic example illustrating Simpson's Paradox for ROC.}
  \label{F:simpson:data}
\end{figure}

\begin{figure}[!h]
  \centering
  \centerline{
    \includegraphics*[width=0.25\refwidth,bb=0 0 410 420]{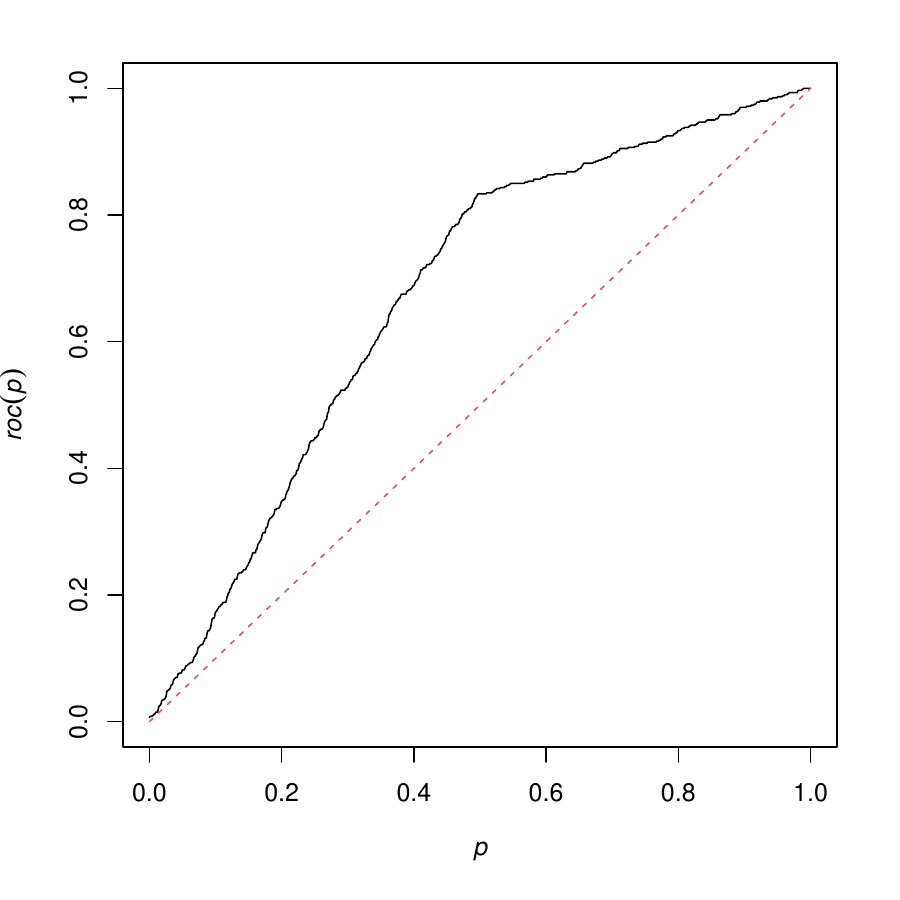}
    \includegraphics*[width=0.25\refwidth,bb=0 0 410 420]{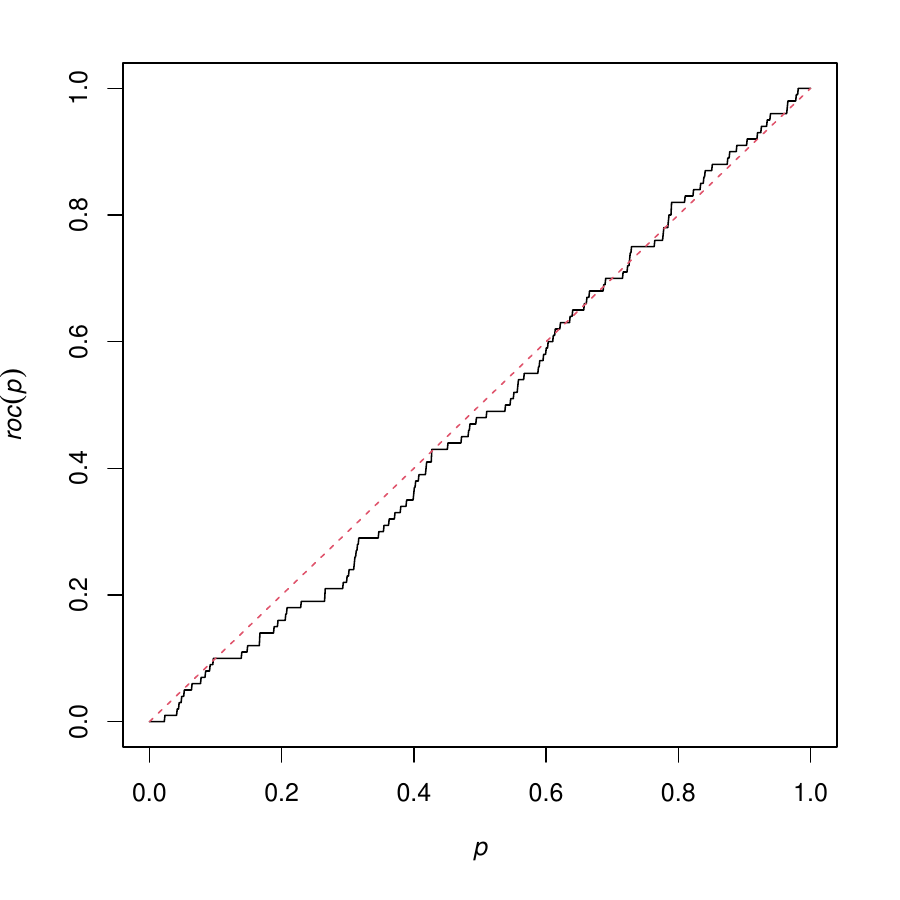}
    \includegraphics*[width=0.25\refwidth,bb=0 0 410 420]{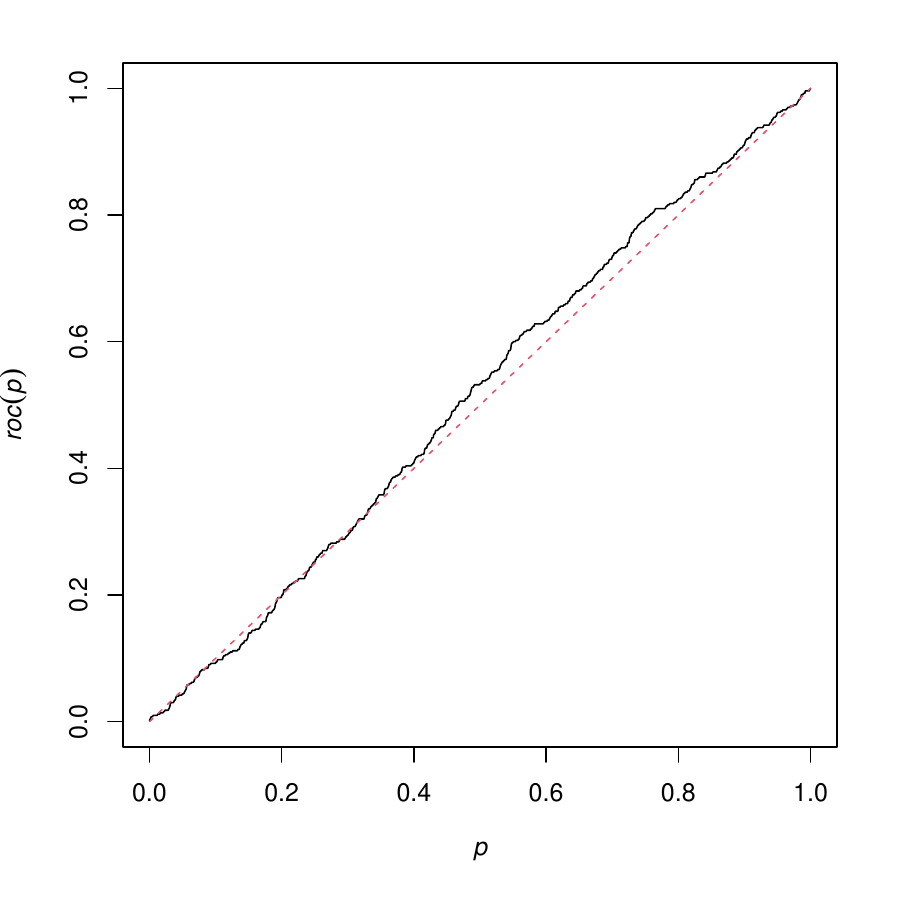}
  }
  \caption{
    Demonstration of Simpson's Paradox
    for the data in Figure~\ref{F:simpson:data}.
    ROC curves for the $x$ coordinate
    calculated for the entire dataset \emph{(Left)},
    for the left half \emph{(Middle)}
    and the right half \emph{(Right)}.
    AUC values are respectively 0.67, 0.48, 0.51.
  }
  \label{F:simpson:effect}
\end{figure}

\ref{app:split_region} describes the connection between the ROC curve for a full study region and
the two ROC curves obtained when splitting the region into two disjoint subsets.

\FloatBarrier

\pagebreak
\section{ROC for a fitted model}
\label{S:ROC:model}

This section examines the particular version of the ROC curve that is most
commonly used in species distribution modelling \citep[p.\ 222]{fran09}
and has been used in mineral prospectivity analysis
\citep{Porwal2010,fabbchun08}.
A predictive model (typically involving several explanatory variables)
is fitted to a presence-absence dataset;
the discriminant variable $S$ is taken to be 
the \emph{fitted presence probability} using the model,
and the ROC curve compares the distributions of $S$
at the ``presence'' and ``absence'' pixels.
We call this the ``model ROC'' (M-ROC).

\subsection{Definition}
\label{S:ROC:model:defn}

\subsubsection{Presence-absence pixel data}
\label{S:ROC:model:defn:pixel}

For presence-absence data (Section~\ref{S:background:spatial:presence})
we assume that a predictive model has been fitted. This may be a
logistic regression, another binary regression, or any model
which specifies the probability of presence $\pi_j$ at each pixel $Q_j$.
Let $\widehat \pi_j$ denote the predicted presence probability at $Q_j$
according to the fitted model. The ``model ROC curve'' (M-ROC) is
conventionally constructed from the estimates
\begin{equation}
  \label{eq:TPFPmodel:finite}
 \TPhat(t) = \frac{
   \sum\nolimits^{J}_{j=1} y_j \indicate{ \widehat \pi_j > t }
}{
  \sum\nolimits^{J}_{j=1} y_j
},
\quad\quad
\FPhat(t) =  \frac{
  \sum\nolimits^{J}_{j=1} (1-y_j) \indicate{ \widehat \pi_j > t }
}{
  \sum\nolimits_{j=1}^{J} (1-y_j)
}.
\end{equation}
Thus $\TPhat(t)$ is the fraction, of those pixels containing data points,
in which the predicted probability of a data point exceeds $t$,
and $\FPhat(t)$ is the fraction, amongst those pixels that do not contain
data points, in which the predicted probability of a data point exceeds $t$.
The M-ROC curve is the reverse P--P plot comparing the distributions
\emph{of the fitted presence probability} $\widehat \pi_j$
over pixels that do and do not contain data points.
If pixels are sorted into increasing order of $\widehat \pi_j$,
then M-ROC and M-AUC measure how well this ordering
segregates ``presence'' from ``absence'' pixels.
In our shorthand, this empirical M-ROC curve based on a fitted model
with presence probability vector $\widehat \bpi$
will be denoted $R_{\widehat \bpi, \by}(p)$.

Geologists use the terms
``fitting rate curve'' \citep{fabbchun08},
``prediction rate curve'' \citep{jcu40323},
``capture efficiency curve'' \citep{Porwal2010}
and ``frequency-area curve'' \citep{fordetal19}
for a plot of the fraction of pixels containing deposits
against the fraction of area
captured by thresholding the fitted probability according to a model.
The area fraction is computed using
both ``deposit'' (presence) and ``barren'' (absence)
pixels; that is, $\FPhat(t)$ in \eqref{eq:TPFPmodel:finite}
is replaced by $\FPalt(t)$ in \eqref{eq:FPalt:finite} with
$z_j$ defined as $\widehat \pi_j$.
\citet{nykaetal15} identify this as a version of the ROC curve,
and point out the practical advantage of using $\FPalt(t)$,
which can be based on a random sample of locations,
instead of $\FPhat(t)$, which requires finding ``true negative'' pixels.

If the true probabilities $\bpi$ were known,
one could define the ``true'' or ``theoretical'' M-ROC curve
$R_{\bpi,\bpi}(p)$ constructed from
\begin{equation}
  \label{eq:TPFPM:finite2}
  \TP(t) = \frac{
             \sum\nolimits^{J}_{j=1} \pi_j \ \indicate{\pi_j > t}
             }{
             \sum\nolimits^{J}_{j=1} \pi_j
             }
 , \quad \quad
  \FP(t) =  \frac{
             \sum\nolimits^{J}_{j=1} (1-\pi_j) \ \indicate{\pi_j > t}
             }{
             \sum\nolimits_{j=1}^{J} (1-\pi_j)
             }.
\end{equation}
Under suitable conditions, $R_{\hat\bpi, \by}(p)$ is a consistent estimate of $R_{\bpi, \bpi} (p)$ if the model is true.

One statistical consideration is the risk of overfitting, since
the fitted presence probability $\widehat \pi_j$ could depend strongly
on the observed response $y_j$, causing $\TPhat(t)$ to be inflated
and $\FPhat(t)$ to be deflated, causing bias in the M-ROC curve.
Overfitting could be a severe problem with nonparametric models, or parametric
models with many degrees of freedom, or small datasets.
A standard remedy is to use the
\emph{leave-one-out} fitted value $\widehat \pi_j^{-j}$ calculated by
fitting the model to all of the data except the data from pixel $j$,
then predicting the probability at pixel $j$ using the covariates in
pixel $j$. Thus $\widehat \pi_j$ in \eqref{eq:TPFPmodel:finite} would be
replaced by $\widehat \pi_j^{-j}$.
In our experience with real data,
using the leave-one-out procedure has a visible but modest effect on
the ROC curve. The change in the AUC value is small,
typically less than 0.01. The AUC may either increase or decrease.
See Figure \ref{F:bei10modelLOO} for an example.

When a model has been fitted, another estimate of the true M-ROC curve
is the \textbf{``model-predicted'' M-ROC curve}, which is the theoretical M-ROC
curve determined by the fitted model. This is based on
\begin{equation}
  \label{eq:TPFP:Mfit:finite}
  \TP(t) = \frac{
             \sum\nolimits^{J}_{j=1} \widehat\pi_j \ \indicate{\widehat\pi_j > t}
             }{
             \sum\nolimits^{J}_{j=1} \widehat\pi_j
             },
             \quad\quad
  \FP(t) =  \frac{
             \sum\nolimits^{J}_{j=1} (1-\widehat\pi_j) \ \indicate{\widehat\pi_j > t}
             }{
             \sum\nolimits_{j=1}^{J} (1-\widehat\pi_j)
             }
\end{equation}
which we shall denote $R_{\hat\bpi,\hat\bpi}(p)$.
Under suitable conditions, the model-predicted M-ROC $R_{\hat\bpi,\hat\bpi}(p)$.
is a pointwise consistent estimator of the true M-ROC $R_{\bpi,\bpi}(p)$
provided the model is true.

\subsubsection{Continuous mapped point pattern data}
\label{S:ROC:model:defn:cts}

For a mapped point pattern $\bx$, instead of a model predicting the
probability of presence in each pixel, we will have a model specifying the
point process intensity $\lambda(u)$ at each spatial location $u$.
The empirical M-ROC will be formed from
\begin{equation}
  \label{eq:TPFPmodel:cts}
  \TPhat(t) = \frac 1 n \sum_{i=1}^n \indicate{\widehat\lambda^{-i}(x_i) > t}
  , \quad\quad
  \FPhat(t) = \frac 1 {|W|}
                \int_W \indicate{\widehat\lambda(u) > t} \dee u
\end{equation}
where $\widehat\lambda^{-i}(x_i)$ denotes the
leave-one-out estimate of intensity at location $x_i$
based on the point pattern data excluding $x_i$.
Here $\TPhat(t)$ is the fraction of data points at which the
fitted intensity value exceeds the threshold, and $\FPhat(t)$ is the fraction
of area in the study region where the fitted intensity function value
exceeds the threshold. 
In our shorthand, the empirical M-ROC curve based on a fitted model
with intensity $\widehat \lambda$ fitted to a point pattern $\bx$
will be denoted $R_{\widehat \lambda, \bx}(p)$.

The theoretical M-ROC will be formed from
\begin{equation}
  \label{eq:TPFPM:cts}
  \TP(t) = \frac{
             \int_W \indicate{\lambda(u) > t} \lambda(u) \dee u
             }{
             \int_W                            \lambda(u) \dee u
             }
, \quad\quad
  \FP(t) = \frac 1 {|W|} \int_W \indicate{\lambda(u) > t} \dee u
\end{equation}
and will be denoted $R_{\lambda,\lambda}(p)$.
When a model has been fitted, the model-predicted M-ROC, denoted $R_{\widehat\lambda,\widehat\lambda}(p)$, is formed similarly, with $\lambda$ replaced by $\widehat\lambda$ in \eqref{eq:TPFPM:cts}.

\subsubsection{Case-control point pattern data}
\label{S:ROC:model:casecontrol}

For spatial case-control point pattern data,
analysed conditionally on the spatial locations
as described in Section~\ref{S:background:model:casecontrol},
a model specifies the probability $\pi_j = \Prob{m_j = 1}$
that location $x_j$ is a case, for $j = 1,\ldots,J$.
The theoretical and empirical M-ROC curves are then constructed using the same
formulas as for pixel data in Section~\ref{S:ROC:model:defn:pixel}:
\eqref{eq:TPFPM:finite2} and
\eqref{eq:TPFPmodel:finite},
with $y_j=m_j=1$ indicating that a point (rather
than a pixel) is a case rather than a control ($y_j=m_j=0$).
The interpretation of the rates, M-ROC and M-AUC
is analogous to the pixel case with case and control
locations taking the role of presence and absence pixels, respectively.
Overfitting may occur in this context and we again recommend the use
of leave-one-out estimates $\widehat\pi_j^{-j}$.

\subsection{Examples}
\label{S:ROC:model:examples}

Consider the logistic regression model \eqref{eq:logistic}
for presence of \textit{Beilschmiedia} in 10-metre pixels
as a function of terrain elevation.
The left panel of Figure~\ref{F:bei10modelconventional} shows the M-ROC curve
for this model, constructed according to standard practice in
statistical ecology \citep[p.\ 222]{fran09};
this is an ROC curve comparing the values
\emph{of the fitted probability of presence}
computed at the %
pixels where trees are present,
against those values computed at
the pixels where trees are absent.
That is, the score $s_j$ at pixel $j$ is defined to be the
fitted presence probability $\widehat \pi_j$;
the observed true positive rate $\TPhat(t)$ and false positive rate $\FPhat(t)$
are computed as in
\eqref{eq:TPFPmodel:finite}
and $\TPhat(t)$ is plotted against $\FPhat(t)$.
Leave-one-out estimates $\widehat \pi_j^{-j}$ of the presence probabilities
could have been used instead, as we recommend,
but the discrepancy is negligible.
The middle panel shows the same calculations for the logistic regression
as a function of terrain slope.

\begin{figure}[!hbt]
  \centering
  \centerline{
    \includegraphics*[width=0.25\refwidth,bb=0 0 410 420]{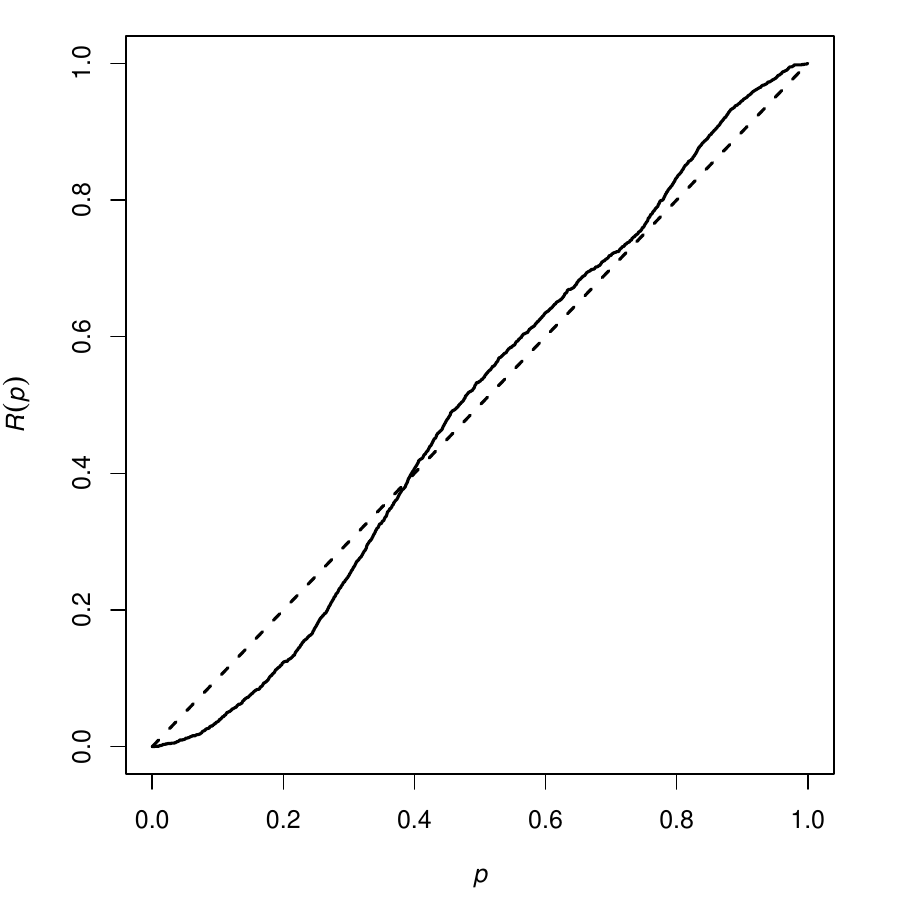}
    \includegraphics*[width=0.25\refwidth,bb=0 0 410 420]{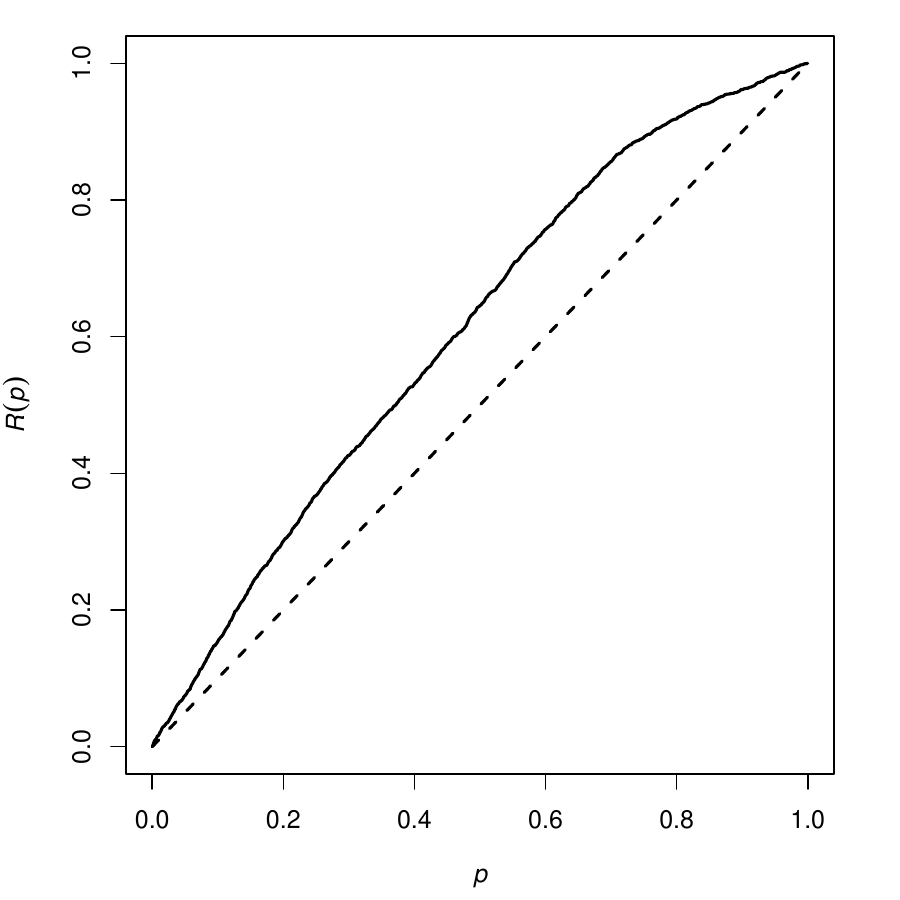}
    \includegraphics*[width=0.25\refwidth,bb=0 0 410 420]{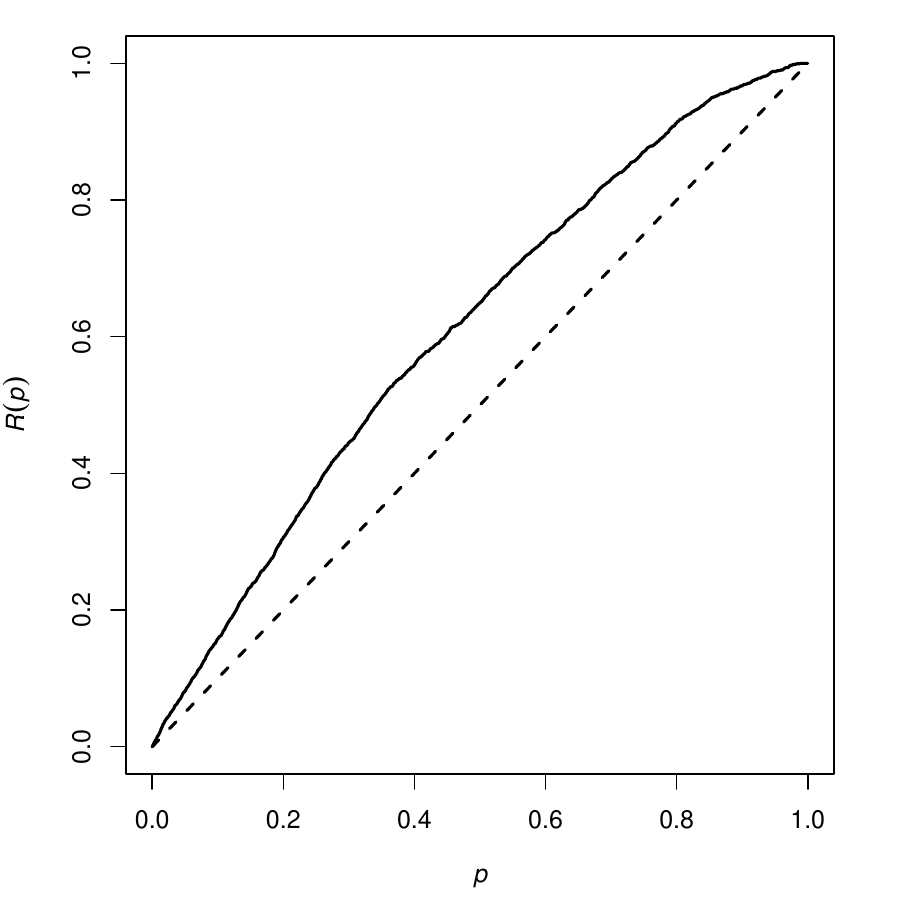}
  }
  \caption{
    M-ROC curves for the logistic regression model
    of \textit{Beilschmiedia} presence against terrain elevation
    (\emph{Left}), against terrain slope (\emph{Middle}),
    and against both terrain elevation and terrain slope (\emph{Right}).
    Pixel size 10 metres.
    AUC values are 0.51, 0.61 and 0.61 respectively,
    Youden statistic values are 0.07, 0.20, 0.17 respectively.
    Note the left and middle panels are identical to
    the left and middle panels of Figure~\protect{\ref{F:bei10ROCmur}}.
  }
  \label{F:bei10modelconventional}
\end{figure}

Note that the left and middle panels of Figure~\ref{F:bei10modelconventional}
are identical to the left and middle panels of Figure~\ref{F:bei10ROCmur},
as a consequence of Lemma~\ref{L:collapse} in Section~\ref{S:ROCmodel:collapse}.
On the other hand,
the right-hand panel of Figure~\ref{F:bei10modelconventional}
is obtained when the model is multiple logistic regression
with additive terms for terrain elevation and terrain slope,
and this does not reduce to a C-ROC.

Interestingly, the right panel shows that the model with
additive terms for elevation and slope has slightly \emph{less} ranking ability
than the model depending only on terrain slope. 
The coefficients of elevation and slope are positive in each fitted model,
but larger slopes tend to occur at smaller elevations (correlation -0.36),
so a model which is expected to produce a better fit actually exhibits
worse ranking ability.
For reference, the fitted presence probabilities for the three logistic regression models
are shown in \ifelseArXiV{Appendix~\ref{SUPP:models}}{Section S2 of the Supplement}.

For the exact spatial locations of the \textit{B.\ pendula} trees,
we may fit the corresponding loglinear Poisson point process models
\eqref{eq:loglinear} and compute the M-ROC curves using
\eqref{eq:TPFPmodel:cts}.
The resulting curves are very similar to those in
Figure~\ref{F:bei10modelconventional} which is expected since these models are the continuous
analogues of the models appearing in Figure~\ref{F:bei10modelconventional}.

\begin{figure}[!hbt]
  \centering
  \centerline{
    \includegraphics*[width=0.3\refwidth,bb=0 0 410 420]{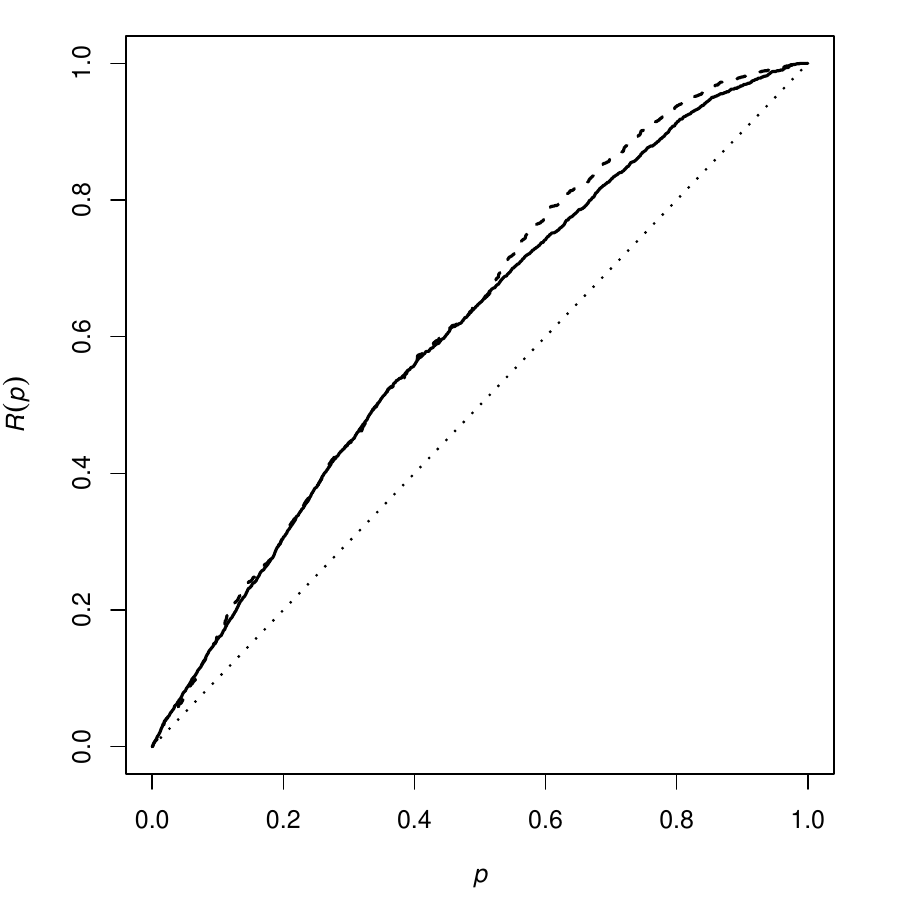}
  }
  \caption{
    M-ROC curves for the logistic regression model
    of \textit{Beilschmiedia} presence against 
    both terrain elevation and terrain slope,
    using the conventional estimate (solid lines)
    and the leave-one-out estimate (dashed lines).
    Pixel size 10 metres.
  }
  \label{F:bei10modelLOO}
\end{figure}

For the model which is multiple logistic regression
with additive terms for terrain elevation and terrain slope,
Figure~\ref{F:bei10modelLOO} shows the M-ROC curves 
computed using the raw estimator applied to the fitted probability of presence
with and without the option of using a leave-one-out estimate.
Somewhat surprisingly, the leave-one-out estimator produces a slightly more optimistic
M-ROC curve with slightly higher AUC value, but as previously noted our experience with real data
is that the leave-one-out estimator has a modest effect on the ROC curve and it may change AUC in
either direction.

\begin{figure}[!hbt]
  \centering
  \centerline{
    \includegraphics*[width=0.3\refwidth,bb=0 0 410 420]{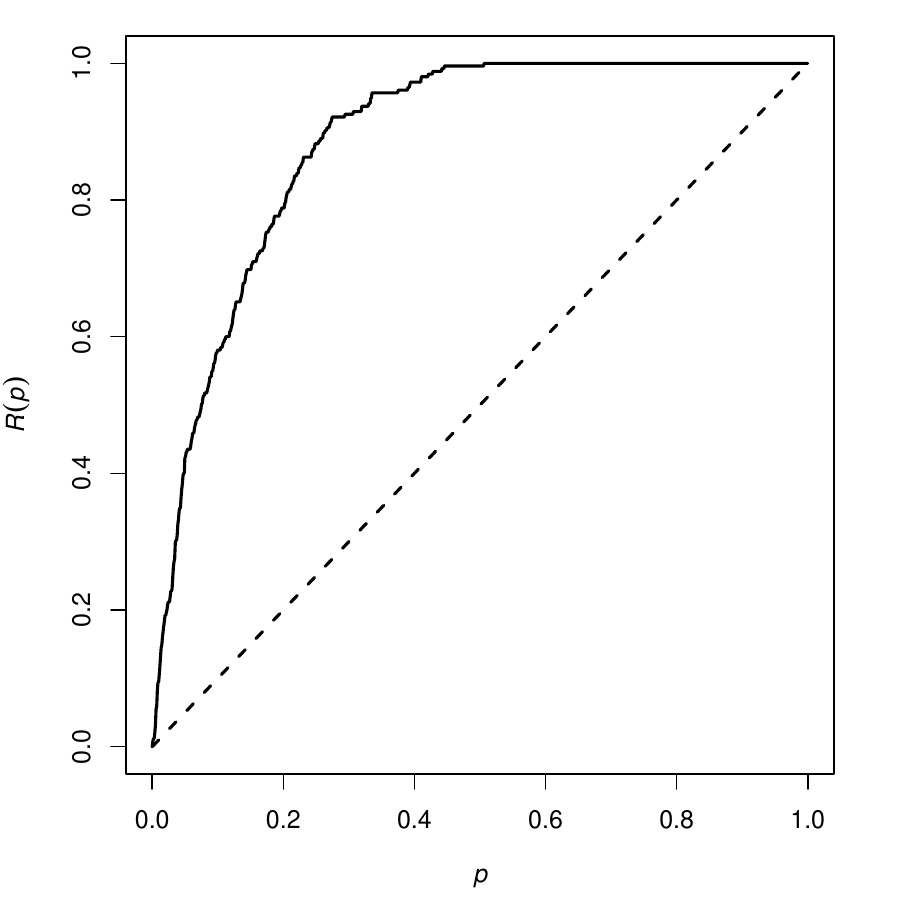}
    \includegraphics*[width=0.3\refwidth,bb=0 0 410 420]{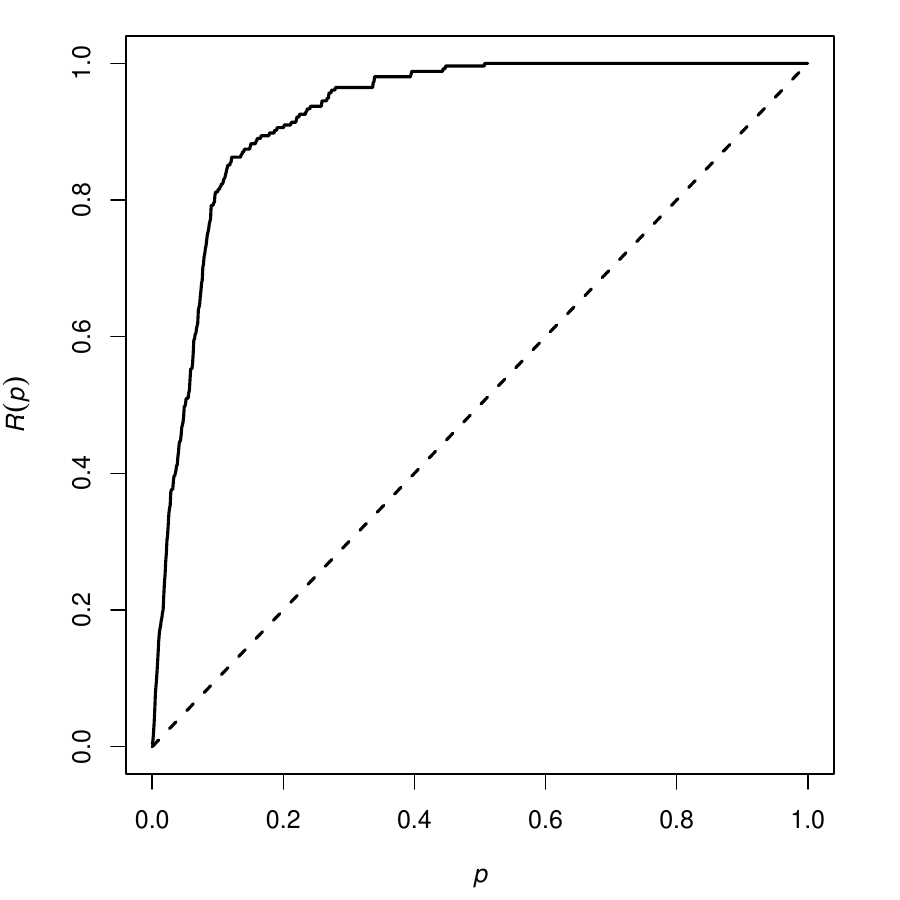}
  }
  \caption{
    M-ROC curves for Poisson point process models
    for the Murchison gold deposits in which the intensity is
    a loglinear function of the distance to nearest fault (\emph{Left})
    or a log linear function of distance to nearest fault and
    greenstone indicator (\emph{Right}).
    AUC values are 0.886 and 0.926 respectively.
    Youden statistic values are 0.647 and 0.742 respectively.
    Values at $p = 0.1$ are 0.580 and 0.813 respectively.
    Note the left panel is identical to
    the right panel of Figure~\protect{\ref{F:bei10ROCmur}}.
  }
  \label{F:murModels}
\end{figure}

Figure~\ref{F:murModels} shows the M-ROC curves for
the loglinear Poisson point process model
for Murchison gold deposits as a function of distance to nearest fault $D(u)$, 
(left panel) 
and as a function of both distance to nearest fault and
greenstone indicator (right panel).
The right panel suggests substantially better performance
for the small values of area fraction $p$ which are of interest
for geological exploration. For example at $p=0.1$ the curves in the
left and right panels have height 0.580 and 0.813 respectively.

\subsection{Variance}

\begin{figure}[!hbt]
  \centering
  \centerline{
    \includegraphics*[width=0.3\refwidth,bb=0 0 410 420]{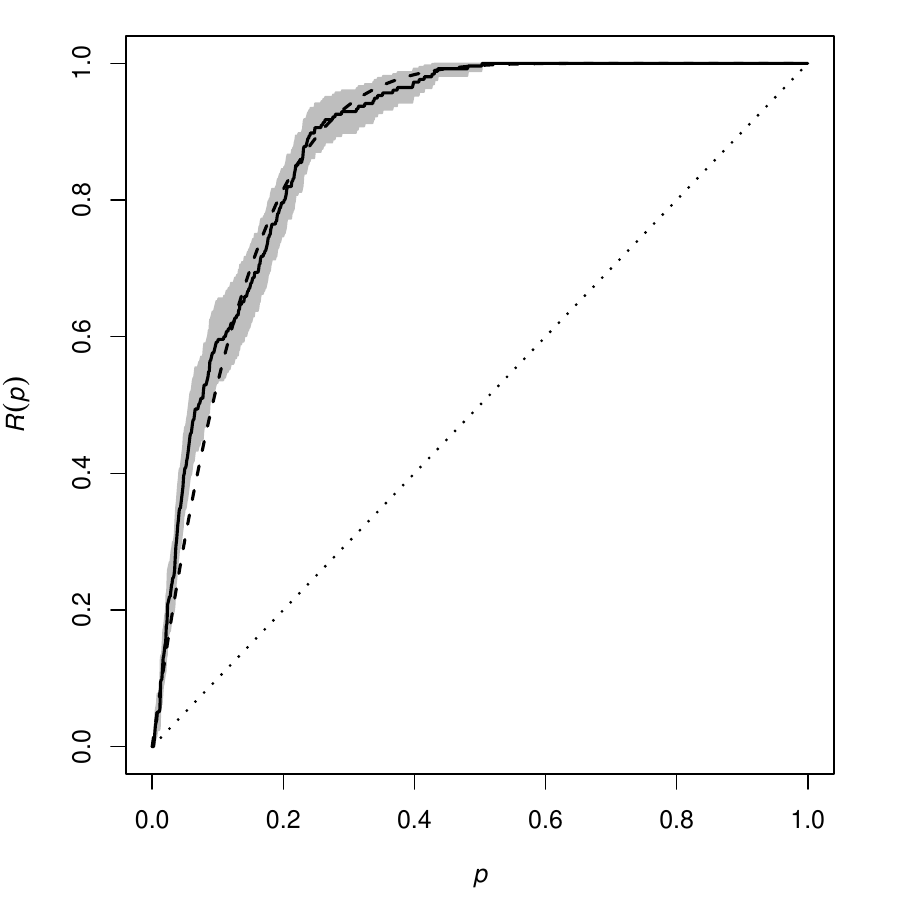}
    \hspace*{4mm}
    \includegraphics*[width=0.3\refwidth,bb=0 0 410 420]{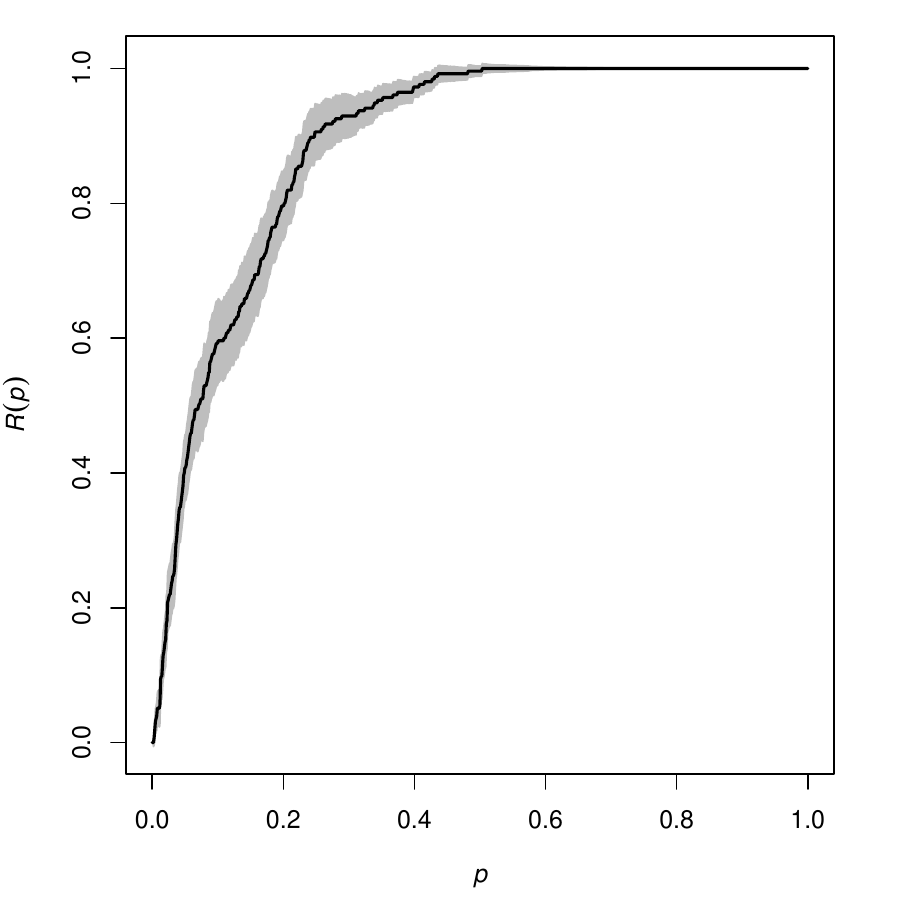}
  }
  \caption{
    \emph{Left:}
    M-ROC curve for the Poisson point process model
    fitted to the Murchison gold data
    assuming intensity is a loglinear function of distance-to-nearest-fault.
    Solid lines: M-ROC curve, raw estimate (leave-one-out).
    Grey shading: pointwise $95\%$ confidence bands 
    calculated using the plug-in estimate of asymptotic variance \eqref{e:varROC:mapped}.
    Dashed lines: M-ROC curve predicted by model.
    \emph{Right:} pointwise 95\% confidence bands
    based on Monte Carlo variance estimate.
  }
  \label{F:murdistROCci}
\end{figure}

The left panel of Figure~\ref{F:murdistROCci} shows ROC curves
for the Poisson point process model fitted to the Murchison gold data
assuming intensity is a loglinear function of distance-to-nearest-fault.
Solid lines show the raw estimate \eqref{eq:TPFPmodel:cts}
of the M-ROC curve for the fitted model,
calculated using leave-one-out estimates of intensity.
Dashed lines show the predicted M-ROC curve $R_{\widehat\lambda,\widehat\lambda}(p)$
of the fitted model, which is based on \eqref{eq:TPFPM:cts}
with $\lambda$ replaced by $\widehat\lambda$.
Grey shading shows asymptotic pointwise $95\%$ confidence bands
around the raw estimate, calculated for the Poisson model
using the plug-in estimate of asymptotic variance \eqref{e:varROC:mapped},
ignoring variability due to parameter estimation.
The right panel of Figure~\ref{F:murdistROCci}
shows pointwise $95\%$ confidence bands
around the raw estimate based on Monte Carlo estimates of pointwise
variance, based on simulations of the fitted Poisson process model,
suggesting that the asymptotic variance approximation is acceptable.
The Monte Carlo approach has the advantage that it can easily be applied
to other point process models and that it takes the extra variability due to parameter estimation
into account.

The extra source of variability associated with the model fit is hard to analyse in any generality because it depends on the kind of model and the method of fitting the model.
In large samples, this source of variability is expected to be negligible relative to the variances given in \eqref{e:varTP:mapped} and \eqref{e:varROC:mapped} so that the same expressions can be used.
If there is only one covariate, and the model intensity is always a monotone increasing function of the covariate, then this source of variability collapses, because the M-ROC is identical to the C-ROC, see Section \ref{S:ROCmodel:collapse}, and the variance of the M-ROC is identical to the variance of the C-ROC.

\subsection{Interpretation of ROC for a fitted model}

The central question is whether these curves can be used
for model criticism or model selection, and in what way.

\subsubsection{M-ROC is insensitive to the model}
\label{S:ROCmodel:monotone}

\begin{lemma}
  The M-ROC for a model with fitted presence probabilities $\hat \pi_j$ is
  \emph{identical} to the M-ROC for another model with fitted presence
  probabilities $\hat \pi^\ast_j$, if $\hat \pi^\ast_j$ is a
  strictly increasing function
  of $\hat \pi_j$, that is $\hat \pi^\ast_j = \Psi(\hat \pi_j)$ where
  $\Psi(p)$ is a strictly increasing function of $p$.
\end{lemma}

\begin{proof}
  For any threshold $t \in \R$ we have
  \[
    \hat \pi_j > t \quad \mbox{ iff } \quad \hat \pi^\ast_j > \Psi^{-1}(t) = s
  \]
  so that the M-ROC curves are identical.
\end{proof}

Thus in Figure~\ref{F:bei10modelconventional} 
each panel would be unchanged if the
fitted presence probabilities were reduced by a factor of 10,
or transformed by the square root, etc.

The M-ROC and M-AUC depend only on the
normalised predictions $\pi_j/\sum_k \pi_k$ or
normalised intensity $\lambda(u)/\int_W \lambda(v) \dee v$,
and cannot detect errors which introduce a multiplicative factor
into the presence probability or intensity.
Hence, the M-ROC is ``not affected by species prevalence''
\citep{manewillorme01}.
Consequently an ROC analysis cannot be used to predict changes in
species abundance. 

This could be considered an advantage when dealing with sampling designs
or modelling techniques which are only able to predict the presence probability
up to an unknown constant factor.
It is a disadvantage insofar as the M-ROC is insensitive to errors
in quantitative prediction (such as over- or under-estimating species
abundance) and cannot be used to compare
the performance of two different models if the models
are related by a monotone transformation.

The same M-ROC will often be obtained for different SDM techniques.

\subsubsection{Collapses to C-ROC}
\label{S:ROCmodel:collapse}

If there is only one explanatory variable $Z$, then the M-ROC curve
is the \emph{same for all models} which depend monotonically on $Z$.
Such models cannot be distinguished from each other on the basis of the
ROC curve or the AUC value.

For any such model, the M-ROC curve is the same as the
C-ROC curve based on the covariate $Z$, and contains information only
about the ranking ability of the covariate, not about the particular model.

\begin{lemma}
  \label{L:collapse}
  Suppose there is only a single spatial covariate $Z$.
  Consider any model which depends only on $Z$, and only
  in a monotone fashion, $\widehat \pi_j = f(z_j)$ where $f$ is a
  strictly increasing function. Then the empirical M-ROC 
  $R_{\widehat \bpi, \by}(p)$
  defined by \eqref{eq:TPFPmodel:finite}
  is identical to the empirical C-ROC $R_{Z, \by}(p)$
  defined by \eqref{eq:TPhat:finite}--\eqref{eq:FPhat:finite}.
\end{lemma}

\begin{proof}
  For any threshold $t \in \R$ we have
  \[
    \hat \pi_j > t \quad \mbox{ iff } \quad z_j > f^{-1}(t) = s
  \]
  so that the ROC curves are identical.
\end{proof}

For example, for presence-absence data, if the model is
logistic regression \eqref{eq:logistic} and if $\widehat\beta_1 > 0$,
then the M-ROC collapses to the C-ROC.
For continuous mapped point pattern data,
if the model is a loglinear Poisson process \eqref{eq:loglinear}
and if $\widehat\beta_1 > 0$, then
the M-ROC $R_{\widehat \lambda, \bx}(p)$
defined by \eqref{eq:TPFPmodel:cts}
is identical to the C-ROC $R_{Z, \bx}(p)$
defined by \eqref{eq:TPhat:cts}--\eqref{eq:FPhat:cts}.

Thus in Figure~\ref{F:bei10modelconventional} the left and middle panels
do not depend on the details of the fitted model,
only on the explanatory variable.

\subsubsection{M-AUC does not measure goodness-of-fit}
\label{S:ROCmodel:badness}

\citet{fielbell97} argue that AUC is a measure of
\emph{goodness-of-fit} of a predictive model.
This is contested by \citet[p.\ 146]{lobojimereal07}.

A straightforward way to see that the M-AUC cannot be a measure of
goodness-of-fit for a fitted model is to consider the case where
the model is true and the fitted model parameters are close to the true values.
For presence-absence data, the M-ROC $R_{\widehat \bpi, \by}(p)$
defined by \eqref{eq:TPFPmodel:finite}
is close to the ``true'' or ``theoretical'' M-ROC
$R_{\bpi,\bpi}(p)$ defined by \eqref{eq:TPFPM:finite2}.
Although the fitted model is almost perfectly correct, the associated M-AUC
is not equal to 1, and could range anywhere between 1/2 and 1.
The M-ROC and M-AUC simply
reflect the spatial non-uniformity of the presence probabilities $\pi_j$.
Similar comments apply for the case of spatial point pattern data.

\subsubsection{Dependent on study region}
\label{S:ROCmodel:subregion}

The M-ROC depends strongly on the choice of the study region:
it is ``region-specific'' \citep{jime12}. 
The justification for this statement is
the same as in Section~\ref{S:ROCcovar:region}.

Even if an SDM is perfectly correct, the AUC of this SDM
will reflect how well the study region can be segregated.
If we restrict attention to a more homogeneous sub-region,
the AUC should be expected to decrease.
See the synthetic example in Figure~\ref{F:theo:data}.

In general, if the study region is divided into sub-regions,
the AUC values computed in the sub-regions will be different
from each other, and different from the AUC computed for
the entire study region, even if the fitted model is the same
in each sub-region.
Hence the AUC cannot be used to measure agreement between model
and data, neither for testing purposes nor for cross-validation.

\subsubsection{Ranking ability}
\label{S:ROCmodel:ranking}

It should now be clear that 
the interpretation of M-ROC and M-AUC as measuring ``predictive ability''
of a fitted model \citep{lobojimereal07,aust07,fran09}
is weak.
Instead one should think of ``ranking ability'' as the qualitative ability to
identify regions of relatively high or low density of points,
rather than the ability to predict numerically the density of points.
The following result connects the ROC to the Lorenz curve,
a common measure of inequality,
which reinforces the interpretation of the ROC curve as a
measure of ranking.

\begin{lemma}
  For a point process model with intensity function $\lambda(u)$,
  the theoretical M-ROC  $R_{\lambda,\lambda}(p)$
  is identical to the reverse Lorenz curve for the function $\lambda(u)$.
\end{lemma}

\begin{proof}
  Suppose $w:W \to [0,\infty)$ is the
  spatially-varying density of ``wealth''.
  Then the Lorenz curve \citep{lore05} is
  the graph of $F_1(t)$ against $F_0(t)$ where
  \[
    F_1(t) = \frac{
               \int_W  w(u) \indicate{w(u) \le t} \dee u
               }{
               \int_W w(u) \dee u
             },
             \quad\quad
    F_0(t) = \frac 1 {|W|} \int_W \indicate{w(u) \le t} \dee u.
  \]
  For a spatial point process model with intensity $\lambda(u)$,
  the theoretical (model-predicted) M-ROC is the graph of 
  $\TP(t)$ against $\FP(t)$ where
  $\TP(t)$ and $\FP(t)$ are given by \eqref{eq:TPFPmodel:cts}.
  If we set $w(u) = \lambda(u)$ for $u \in W$,
  the two curves are equivalent, apart from the reversal of the inequality
  from $w(u) \le t$ to $\lambda(u) > t$.
\end{proof}

In the discrete case of presence-absence data,
a similar statement holds approximately, when all presence probabilities are
sufficiently small.

\subsubsection{M-ROC is more optimistic than C-ROC}
\label{ss:NeymanPearson}

It is a consequence of the Neyman-Pearson Lemma that the
M-ROC is always more ``optimistic'' than the C-ROC for any spatial covariate
in the same context \citep[pp.\ 24--25]{krzahand09}, \citep{egan75,pepe03}.

\begin{lemma}
  For a point process model with intensity function $\lambda(u)$,
  the theoretical M-ROC dominates the theoretical C-ROC
  for any spatial covariate $Z(u)$:
  \begin{equation}
    \label{eq:M>=C}
    R_{\lambda,\lambda}(p) \ge R_{Z,\lambda}(p)
    \quad\mbox{ for all }\quad 0 \le p \le 1.
  \end{equation}
  Consequently, the AUC associated with the M-ROC
  is always greater than or equal to the AUC associated with the C-ROC.
\end{lemma}

\begin{proof}
  For brevity, write $I_A(f) = \int_A f(u) \dee u$ for any
  subset $A \subseteq W$ and any integrable function $f$ on $W$,
  and $L(f,t) = \{ u \in W: f(u) > t\}$ for the upper level set.
  Let $p \in [0,1]$ be fixed.
  From \eqref{eq:TPFPM:cts} we have
  $R_{\lambda,\lambda}(p) = I_{L(\lambda,t)}(\lambda)/I_W(\lambda)$
  where $t$ satisfies $I_{L(\lambda,t)}(1)/I_W(1) = p$.
  Similarly, from \eqref{eq:TP:true:cts} and \eqref{eq:FPhat:cts},
  $R_{Z,\lambda}(p) = I_{L(Z,s)}(\lambda)/I_W(\lambda)$
  where $s$ satisfies $I_{L(Z,s)}(1)/I_W(1) = p$.
  The Neyman-Pearson Lemma states that, 
  given two integrable functions $f(u), g(u) \ge 0$ on $W$,
  considering all subsets $A \subseteq W$ satisfying the constraint that
  $I_A(g) = G$
  where $G$ is a fixed value,
  the subset that achieves the largest value of
  $I_A(f)$
  is of the form
  $A = L(f/g, c)$
  for some constant $c$.
  Applying the Neyman-Pearson Lemma
  to the case $f(u) = \lambda(u)$ and $g(u) \equiv 1$,
  we find that
  $A = L(\lambda,t)$ and $B = L(Z,s)$ satisfy
  $I_A(1) = I_B(1) = p|W|$
  so that
  $I_A(\lambda) \ge I_B(\lambda)$,
  and therefore  $I_A(\lambda)/I_W(\lambda) \ge I_B(\lambda)/I_W(\lambda)$,
  that is, $R_{\lambda,\lambda}(p) \ge R_{Z,\lambda}(p)$.
\end{proof}

\section{Connection to hypothesis tests}
\label{S:tests}

To gain further insight about the ROC and AUC,
we now study their connection to hypothesis tests.

\subsection{AUC and the Wilcoxon-Mann-Whitney test}
\label{S:tests:numerical}

First we consider numerical data.
For the ROC and AUC described in Section~\ref{S:background:ROC},
we noted in Section~\ref{S:background:ROC:P-P} that the ROC
is effectively a comparison between two probability distributions,
namely the distributions of the statistic $S$ in the Positive and Negative
populations.

\citet{hanlmcne82} show that AUC is proportional to the 
Wilcoxon rank sum test statistic (Mann-Whitney $U$ statistic)
of a difference between two groups. 
Assume that independent random samples are taken from the
Positive and Negative groups, with sizes $n_1$ and $n_2$ respectively.
The null hypothesis $\hyp 0$ is that 
the distribution of the statistic $S$ is the same in each group.
The Wilcoxon-Mann-Whitney test statistic $U_1$ is calculated
by pooling the values of $S$ obtained from the two groups,
arranging these values into ascending order,
and computing the sum of the ranks of the values from the Positive group.
Large values of $U_1$ favour the alternative hypothesis that values of $S$
are higher in the Positive group than in the Negative group.

\citet{hanlmcne82} show that $\AUC = U_1/(n_1 n_2)$.
Thus, $\AUC$ is a measure of departure from the null hypothesis
of equality between two groups --- what might be called a
measure of \emph{badness-of-fit} for the \emph{null} model of equality.
Larger values of $\AUC$ indicate greater discrepancy
between the Positive and Negative groups, rather than indicating
greater agreement with a particular alternative hypothesis.

Additionally \citet{handtill01} note that 
the AUC is equivalent to the Gini index, $G = 2 \, \AUC - 1$,
another index of discrepancy between two distributions.

These results may be applied to case-control data.
Consider the mucosa data of Figure~\ref{F:mucosa} consisting of two
groups of cells labelled ECL and Other.
Using the vertical coordinate $y$ as the discriminant variable,
the AUC is 0.726.
The one-sided Wilcoxon rank sum test has $p$-value $4.5 \times 10^{-9}$
indicating very strong evidence against the null hypothesis that the two types
of cells have the same spatial distribution, in favour of the alternative that
ECL cells tend to be relatively more abundant at higher values of $y$.

\subsection{AUC and the Berman-Waller-Lawson test}
\label{S:tests:Berman}

Next we consider spatial point pattern data
(Section~\ref{S:background:spatial:point}).
\citet{berm86} developed two statistical tests for determining whether
a mapped spatial point pattern depends on a spatial covariate.
Independently \citet{walletal92} and \citet{laws93a} developed a
similar test for presence-absence data, to determine whether the
presence probability depends on a spatial covariate.
The Lawson-Waller test is asymptotically equivalent
to Berman's first test as pixel size tends to zero,
so we shall focus on Berman's tests.

\subsubsection{Derivation of Berman's tests}

Suppose that a mapped point pattern $\bx = \{x_1, \ldots, x_n\}$ has been
observed in the study region $W$.
The aim is to determine whether the intensity of points is homogeneous,
or whether the intensity depends on the real-valued spatial covariate $Z(u)$.

The derivation of Berman's tests assumes that the point pattern
is a realisation of a Poisson process $\bX$ in $W$
with some (unknown) intensity function $\lambda(u), \; u \in W$.
The null hypothesis is the assumption of ``no effect'', that is,
the point process intensity is uniform:
\[
  \hyp 0: \; \lambda(u) \equiv \lambda,
\]
while the alternative hypothesis $\hyp 1$ is 
that the intensity is not spatially constant, or more specifically
that the intensity is a non-constant function of $Z$,
\[
  {\hyp 1}^\prime: \; \lambda(u) = \rho(Z(u)),
\]
where $\rho(z)$ is some unknown function that is not constant.

The total number of points in $\bx$ follows a Poisson distribution
with mean
$
  \Lambda = \int_W \lambda(u) \dee u,
$
and given that there are exactly $n$ points, their locations are
i.i.d.\ with probability density $\lambda(u)/\Lambda$.
Let $z_i = Z(x_i)$ be the values of the covariate observed at the
data points of $\bx$.
The key insight \cite[Sec.\ 3.1]{berm86}
is that, under these assumptions, the values $z_1,\ldots,z_n$ 
are i.i.d.\ with cumulative distribution function
\begin{equation}
\label{e:cdf:Z:general}
F(z) = \int_W \indicate{Z(u) \le z} \lambda(u)/\Lambda \dee u.
\end{equation}
Under $\hyp 0$ we have $\lambda(u)/\Lambda\equiv1/|W|$ and the distribution function $F$ is
\begin{equation}
\label{e:cdf:Z:null}
  F_0(z) = \frac{1}{|W|} \int_W \indicate{Z(u) \le z} \dee u,
\end{equation}
which can be interpreted as the distribution function of $Z(U)$
where $U$ is a randomly selected location, uniformly distributed in $W$.
Under $\hyp 1$ the values $z_i$ have some other cumulative distribution
function. Under ${\hyp 1}^\prime$ this distribution is
a weighted distribution
\begin{equation}
\label{e:cdf:Z:gen=null}
  F(z) = \frac{|W|}{\mu} \int_{-\infty}^z \rho(t) \dee F_0(t).
\end{equation}
Under these assumptions, the values $z_1,\ldots,z_n$ are sufficient,
so the test may be conducted using these values alone.

Therefore the problem reduces to testing
whether the probability distribution of
$z_1,\ldots,z_n$ conforms to the null distribution $z_i \sim F_0$.
This can be performed using any of the classical tests of
goodness-of-fit for real random variables
\citep{dagostep86,readcres88}.

Berman's first test (which he called the $Z_1$ test) 
is based on a standardised version of the statistic $S = \sum_i z_i$.
Under $\hyp 0$ it follows from \cite[p.57, eq 3.3]{berm86}, that $S$ has mean
and variance 
\begin{eqnarray}
  \label{e:mu}
  \EE[S] = \mu &=& \lambda \int_W Z(u) \dee u \\
  \label{e:sig2}
  \var{S} = \sigma^2 &=& \lambda \int_W Z(u)^2 \dee u.
\end{eqnarray}
Estimating $\widehat\lambda = n/|W|$, or conditioning on $n$,
we may plug in to \eqref{e:mu}--\eqref{e:sig2}
and compute the standardised statistic $T = (S-\mu)/\sigma$
whose null distribution is approximately standard normal.

In Berman's second test (which he called the $Z_2$ test),
the values $z_i$ are first transformed
to $u_i = F_0(z_i)$, the probability integral transformation
for the null hypothesis. For simplicity we assume that $F_0$ is
continuous and strictly increasing.
Under $\hyp 0$, the values $u_1,\ldots,u_n$
are i.i.d.\ and uniform on $[0,1]$, with mean $\EE[U_i] = 1/2$ and
variance $\var{U_i} = 1/12$. Under $\hyp 0$,
the sample mean $\bar U = \frac 1 n \sum_i u_i$
has mean $\EE[\bar U] = 1/2$ and variance $\var{\bar U} = 1/(12 n)$.
Accordingly the test statistic
\begin{equation}
  \label{e:Berman2}
  V_2 =  \sqrt{12 n} \left( \bar U - \frac 1 2 \right)
\end{equation}
has null distribution approximately standard normal.  

\subsubsection{Connection to C-ROC}

Connections between the C-ROC (Section~\ref{S:ROC:covariate})
and Berman's tests are now clear.
\begin{lemma}
  Assume $F_0$ is continuous and monotone increasing. Then
  \begin{enumerate}
  \item
    The empirical cumulative distribution function of
    the values $u_1,\ldots,u_n$ is
    \[
      \widehat G(p) = \frac 1 n \sum_{i=1}^n \indicate{u_i \le p}
      = \widehat F(F_0^{-1}(p))
    \]
    where $\widehat F(z)$ is the empirical cumulative distribution function
    of $Z(x_1),\ldots,Z(x_n)$.
  \item 
    $\widehat G(p) =  \lo R_{Z,\bz}(p)$.
  \item 
    The test statistic in Berman's second test is related to the AUC by
    \begin{equation}
      V_2 = \sqrt{12 n} \left( \frac 1 2  - \lo\AUC\right)
          = \sqrt{12 n} \left( \hi\AUC - \frac 1 2 \right)
    \end{equation}
  \end{enumerate}
\end{lemma}

A sufficient condition for the continuity and monotonicity of $F_0$
is that $Z$ is Lipschtz continuous with nonzero gradient almost everywhere.

\begin{proof}
  By the assumptions on $F_0$ we have
  $u_i \le p$ if and only if $F_0^{-1}(u_i) \le F_0^{-1}(p)$
  which is equivalent to $z_i \le F_0^{-1}(p)$ so that
  $\widehat G(p) = \widehat F_1(F_0^{-1}(p))$.
  The next statement follows.
  Using Fubini's Theorem it can be verified that
  $\bar u = \int_0^1 (1-\widehat G(p)) \dee p$.
  Consequently
  \[
    V_2 =  \sqrt{12 n} \left( \bar u - \frac 1 2 \right) 
    =  \sqrt{12 n} \left( \frac 1 2  -  \int_0^1 \lo R_{Z,\bz}(p) \dee p \dee p \right) 
    = \sqrt{12 n} \left( \frac 1 2  - \lo\AUC\right).
  \]
  The last equation follows because $\lo\AUC = 1 - \hi\AUC$.
\end{proof}

For practical interpretation, the most important finding is that AUC
is related to the test statistic for a goodness-of-fit test of the
null hypothesis of uniform intensity, that is,
the null hypothesis of ``no effect'' (no dependence on the covariate $Z$),
rather than a test of goodness-of-fit to a fitted model.
The primary role of the covariate $Z$ in Berman's tests is
to define the \emph{alternative} hypothesis, that is, the type of departure
from the null hypothesis of uniformity against which the test is
designed to be sensitive. 

Berman's tests can also be viewed in the context of a parametric model
as significance tests for an effect term in the model \citep{berm86}.
Assume that the intensity is a loglinear function of $Z$
as in \eqref{eq:loglinear}. Then the uniformly most powerful test
of ${\hyp 0}^\dag: \beta_1 = 0$ against ${\hyp 1}^\dag: \beta_1 < 0$
is Berman's first test.
A similar statement holds for presence-absence data for the
logistic regression model \eqref{eq:logistic}.
As Berman notes, this test is not optimal against
other kinds of departure from the null hypothesis.
In an effort to improve performance against some alternatives,
Berman proposes the second test based on $V_2$,
and also recommends that the analyst
should plot and inspect the empirical CDF $\widehat G(p)$ as a diagnostic.

Some practical examples follow.
For the \textit{Beilschmiedia} pattern in the left panel of Figure~\ref{F:bei}
using the terrain elevation covariate,
Berman's first and second tests have two-sided $p$-values of
$0.466$ and $0.014$ respectively.
The discrepancy between these test outcomes arises because
the loglinear model \eqref{eq:loglinear} is inappropriate;
this is explained carefully by %
\citet{berm86}.
Using the terrain slope covariate, the corresponding $p$-values are
effectively zero.

For the Murchison data using the distance-to-nearest-fault covariate,
Berman's first and second tests have $p$-values effectively zero.

\subsection{Insensitivity to dependence on other covariates}
\label{S:tests:insensitivity}

The C-ROC curve for a covariate $Z$ represents the ``effect'' of $Z$,
more precisely, the extent to which an increase
in the value of $Z$ is associated with increased probability of presence
or increased point process intensity. The C-ROC 
is entirely insensitive to the possible effects of other covariates
which may be present.

\begin{figure}[!htb]
  \centering
  \centerline{
    \includegraphics*[width=0.31\refwidth,bb=10 10 420 420]{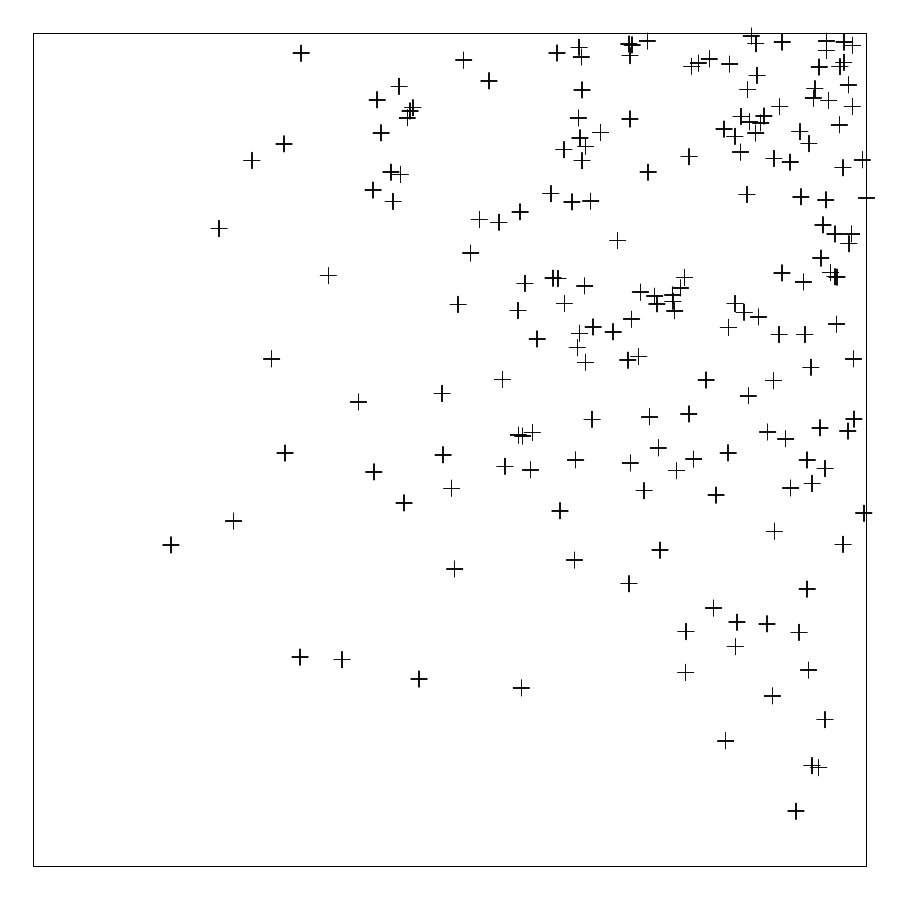}
    \hspace*{7mm}
    \includegraphics*[width=0.3\refwidth,bb=0 0 405 420]{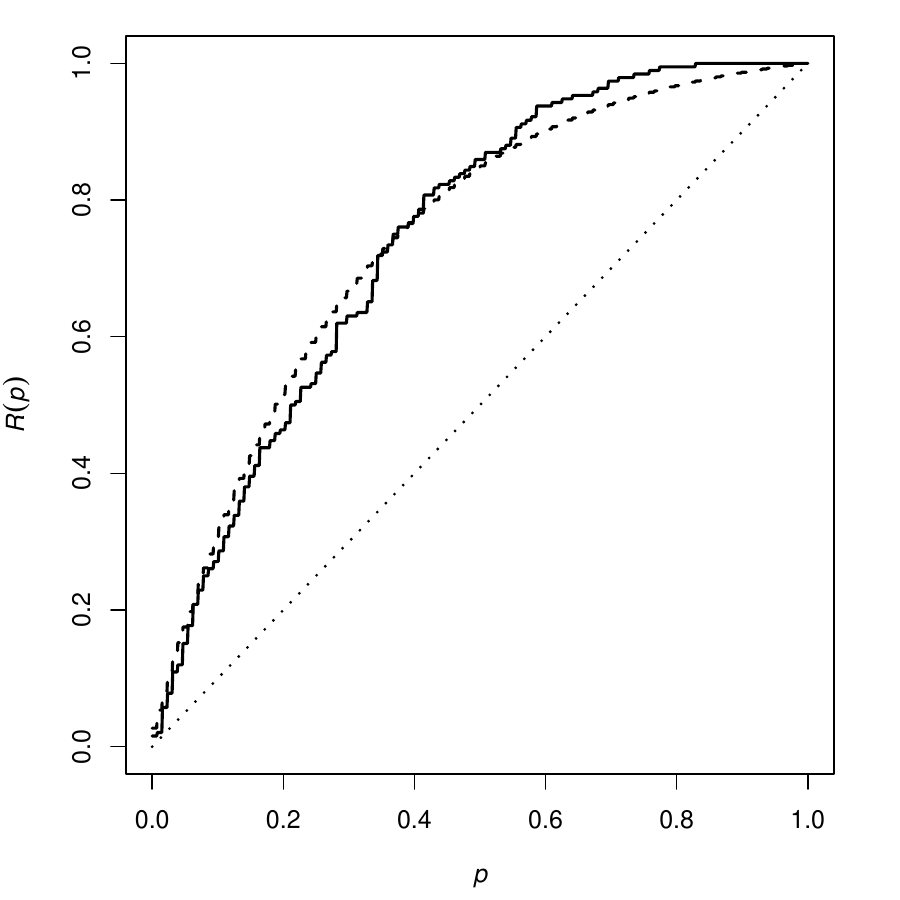}
  }
  \caption{
    Insensitivity of C-ROC to other covariates.
    \emph{Left:}
    Simulated point pattern with intensity depending on
    both the $x$ and $y$ coordinates.
    \emph{Right:} empirical C-ROC curve based on the $x$ coordinate.
  }
  \label{F:corner}
\end{figure}

Figure~\ref{F:corner} shows a synthetic example.
The left panel is a point pattern generated according to the
intensity function $\lambda(x,y) = 100 x^2 y$ in the unit square.
The right panel shows the empirical M-ROC curve for the fitted
loglinear Poisson model
with intensity $\lambda(x,y) = \exp(\beta_0 + \beta_1 x)$ depending only
on the $x$ coordinate, which is also the empirical C-ROC curve
based on the $x$ coordinate. The corresponding AUC is 0.744.
According to some writers, these results suggest a high level
of ``goodness of fit'' of the fitted model. However, this is false for two
reasons.

Firstly, as explained above, while the M-ROC curve
does suggest that the point process intensity depends on $x$, it does not
carry any information about the correctness or incorrectness of the 
functional form of the fitted model, because it collapses to the C-ROC curve
based on the covariate $x$.

Secondly, the fact that the true intensity
also depends on the $y$ coordinate, is completely hidden in this analysis.
A more appropriate test described in \citet[Sec.\ 10.3.5]{baddrubaturn15},
using the $y$ coordinate as the
covariate for the test statistic, yields a $p$-value less than 0.0001,
indicating strong evidence against the fitted model.

\subsection{Youden index and Kolmogorov-Smirnov test}
\label{S:tests:KS}

The Youden index  was mentioned in
Section~\ref{S:background:ROC:performance} as an alternative
to the $\AUC$ as a measure of `performance'. The Youden index
is closely related to the Kolmogorov-Smirnov test
\citep{kolm33test,smir48}. In the general context of ROC and AUC
for numerical data (Section~\ref{S:tests:numerical}),
the Kolmogorov-Smirnov test of the null hypothesis that $F \equiv F_0$
is based on the statistic $\sqrt n D$, where $D$ is
the maximum vertical separation between
the graphs of $\widehat F(z)$ and $F_0(z)$:
\begin{equation}
  \label{eq:KS.stat_new}
  D = \max_z | \widehat F(z) - F_0(z) |
  = \max_p | \widehat F(F_0^{-1}(p)) - p | .
\end{equation}
The two-sided Youden index $J$ is equivalent to $D$ as defined above.
The one-sided Youden index $J_{+}$ defined in \eqref{eq:Youden}
is equivalent to the one-sided maximum deviation
$D_{+} = \max_z ( \widehat F(z) - F_0(z) | )_{+}
= \max_p ( \widehat F(F_0^{-1}(p)) - p )_{+}$
and corresponds to the test statistic $\sqrt n D_{+}$ for the one-sided
Kolmogorov-Smirnov test \citep{baddetal21thresh}.

Assuming that $\widehat F(z)$ is the empirical cumulative
distribution function obtained from $n$ independent realisations from $F$,
the null distributions of $\surd n \, D$ and $\surd n \, D_{+}$ are known
exactly. In the application to spatial point pattern data or spatial
presence-absence data, this assumption is valid if the point process is
Poisson, or the presence-absence indicators for different pixels
are independent. 

Thus the performance measures (AUC, Youden) for ROC curves
are related to \textbf{tests of the presence of an effect}
(Berman-Waller-Lawson, Kolmogorov-Smirnov),
rather than measures of ``goodness of fit'' or a model.
The AUC and Youden index are measures of the size of the effect,
and are not adjusted for sample size.

For the \textit{Beilschmiedia} pattern in the left panel of Figure~\ref{F:bei}
using the terrain elevation covariate, the Kolmogorov-Smirnov test
has $p$-value less than $0.001$. Likewise for the terrain slope covariate.
Likewise for the Murchison data using the distance-to-nearest-fault covariate.

\subsection{Other tests based on the spatial CDF}
\label{S:tests:CDF}

In addition to the Wilcoxon-Mann-Whitney and Kolmogorov-Smirnov tests,
there are other hypothesis tests which have greater power against
certain alternative hypotheses. These include the 
Cram\'er-Von Mises test \citep{cram28,mise31}
with test statistic
\begin{equation}
  \label{eq:CvM.stat_new}
  \omega^2 = n \int_{-\infty}^\infty [\hat F(z) - F_0(z)]^2 \dee F_0(z)
  = n \int_0^1 [\hat F(F_0^{-1}(p)) - p]^2 \dee p 
\end{equation}
and the Anderson-Darling test \citep{andedarl52,andedarl54}
with test statistic
\begin{equation}
  \label{e:AD.stat_new}
  A = n \int_{-\infty}^\infty 
  \frac{
    (\hat F(z) - F_0(z))^2
  }{
    F_0(z) (1 - F_0(z))
  }
  \dee F_0(z)
  =
  n \int_0^1
  \frac{
    (\hat F(F_0^{-1}(p)) - p)^2
  }{
    p (1-p)
  }
  \dee p.
\end{equation}
Again the null distributions of these test statistics are
known exactly, under the same assumptions as above.

For the \textit{Beilschmiedia} data in the left panel of Figure~\ref{F:bei}
using the terrain elevation and terrain slope covariates,
the Kolmogorov-Smirnov, Cram\'er-von Mises and Anderson-Darling tests all have
$p$-values less than $0.001$. Likewise for the Murchison data
against distance to nearest fault.

\subsection{M-ROC and M-AUC for a fitted spatial model}
\label{S:tests:fitted}

\citet{fielbell97} proposed that the AUC is a measure of
goodness-of-fit for a species distribution model fitted to
presence-absence data. In their context, this refers to the M-AUC, 
calculated using the fitted presence probabilities as the discriminant variable.
The claim is that larger values of the M-AUC correspond to a better fit.

This is logically incompatible with the usual formulation of a
goodness-of-fit test, in which the \emph{null} hypothesis postulates
that the data conform to the fitted model under consideration;
the test statistic is a measure
of \emph{departure} from this null hypothesis; and large values of the
test statistic indicate that the null hypothesis should be \emph{rejected}.
The model is deemed to be a good fit if the result of the goodness-of-fit test
is \emph{not} statistically significant.

Sections~\ref{S:tests:Berman}--\ref{S:tests:CDF} above
dealt with the C-AUC, but we argue below that the conclusions
also apply broadly to the M-AUC.
Then M-AUC is not a measure of goodness-of-fit of a fitted model of interest;
rather it is a measure of departure from the null hypothesis of uniformity,
and is related to a test of this null hypothesis. 

A true test of goodness-of-fit, for a species distribution model
or spatial point process model, would normally be based on a covariate
\emph{that is not used in the model}. In Section~\ref{S:tests:Berman}
we saw that the AUC is related to the Berman-Waller-Lawson
test of the null hypothesis of \emph{uniform} intensity
(intensity not depending on any covariates) against the
alternative that the intensity depends on a covariate $Z$. In that context,
the covariate $Z$ is extraneous to the null model
whose goodness-of-fit is being tested.

Extending the results of Sections~\ref{S:tests:Berman}--\ref{S:tests:CDF}
to apply to the M-AUC is technically complicated, because
the discriminant variable $S = \widehat p$ is data-dependent,
so there is an additional source of variability associated with sampling error.
Detailed analysis will depend on details of the model and
the fitting algorithm.

This extra source of variability can be ignored entirely in the case where
the model involves only one explanatory variable $Z$, and where the
model states that the probability of presence is a monotonically increasing
(or monotonically decreasing) function of $Z$. In that case,
by Lemma~\ref{L:collapse}, the M-ROC collapses to the C-ROC based on $Z$,
and the results of Section~\ref{S:tests:Berman} apply.

Similarly in a logistic regression or loglinear Poisson process model
with multiple covariates, the M-ROC curve depends only on the direction
$\widehat\btheta/\| \widehat\btheta \|$ of the fitted coefficient vector
$\widehat\btheta$, not on its magnitude $\| \widehat\btheta \|$.
In other cases, the effect of sampling variability can be neglected
to first order if we use the leave-one-out predictor $\widehat \pi_j^{-j}$.

\section{Connection to resource selection functions}
\label{S:rhohat}

This section explores the relationship between ROC curves
and ``resource selection functions''.
The findings give guidance on how to interpret or ``read'' the ROC curve,
and how to use the ROC curve for model checking.

\subsection{Motivation}

Inspection of the C-ROC or M-ROC curves is an indirect and impractical way
of investigating the effect of covariates.
Covariate values are not displayed in graphs of the ROC.
For the C-ROC (Section~\ref{S:ROC:covariate})
it would be possible to annotate the plot, 
adding the covariate values as tickmarks on the axes.
However, for the M-ROC (Section~\ref{S:ROC:model}),
if the model involves more than one covariate, it becomes impractical
to depict the mapping from covariate values to predicted values
of presence probability.

The dependence of the presence/absence events
on the spatial covariates is described by a species distribution model (SDM).
In most SDM's, the presence probability $p$ at a given pixel
is modelled as depending on the covariate values $Z_1, \ldots, Z_m$
at the same pixel, through a function
\begin{equation}
  \label{eq:SDM}
  p = \rho(Z_1, Z_2, \ldots, Z_m).  
\end{equation}
In various applications, the function $\rho$ in \eqref{eq:SDM} is called a
\emph{resource selection function} or \emph{mineral potential map}.

For ecological applications, Manly \citep[p.\ 177]{manlmcdothom93}
defined a \emph{resource selection function} (RSF) as ``any function
that is proportional to the probability of use by an organism''.
Any function proportional to $\rho$ in \eqref{eq:SDM} could be
called an RSF, and estimation of an RSF is equivalent to
estimation of $\rho$ up to a constant factor.
The resource selection function value reflects
preference for, or avoidance of, particular environmental conditions.
\citet[Sec.\ 1]{boycetal02} say 
``RSF models can be powerful tools
in natural resource management, with applications
for cumulative effects assessment, land-management planning,
and population viability analysis''.
While RSF and SDM are closely related,
the ecological literature seems to associate the terms RSF and SDM
with \emph{nonparametric} and \emph{parametric} estimation, respectively.

In applications to mineral exploration,
the function $\rho$ in \eqref{eq:SDM} is an index of the \emph{prospectivity}
\citep{bonh95,Carranza2009,badd18iamg,knoxgrov97}
or \emph{mineral potential} \citep{Porwal2010,fordetal19},
the predicted relative frequency of undiscovered
deposits as a function of geological and geochemical
covariates.

Similarly, the dependence of a spatial point pattern on
spatial covariates is described by a point process model.
In many such models, the intensity function of the model 
is a function of the covariates, in the form
\begin{equation}
  \label{eq:PPM}
  \lambda(u) = \rho(Z_1(u), \ldots, Z_m(u)), \quad u \in W
\end{equation}

The function $\rho$ in \eqref{eq:SDM} or \eqref{eq:PPM}
can be estimated from data
by parametric model-fitting or non-parametric estimation
in various ways. Parametric modelling and model-fitting
are discussed in \citet[Chap.\ 9]{baddrubaturn15}.
Estimators of resource selection functions
are discussed in \citet{manlmcdothom93}.
In the case of a single explanatory variable $Z$,
nonparametric estimators of the function $\rho$ were
described by \citet{baddchansongturn12,badd18iamg,guan08,guanwang10}.

The left panel of Figure~\ref{F:murRho} shows three different estimates of the
function $\rho$ for the Murchison gold data using the
distance to nearest fault as the covariate $Z$.
We have used the point process formulation.
That is, we assume the point process intensity $\lambda(u)$ at a location $u$
is related to the covariate $Z(u)$ by $\lambda(u) = \rho(Z(u))$,
where $\rho$ is a function to be estimated. 
Thick solid lines show
the kernel smoothing estimate of $\rho$ proposed in \citet{baddchansongturn12}.
Dashed lines show the maximum likelihood estimate
for the parametric model $\rho(d) = \exp(\beta_0 + \beta_1 d)$
under the Poisson assumption
corresponding to \eqref{eq:loglinear}. Thin solid lines show
a monotonic regression estimate proposed by \citet{sage82}
and adapted to this context by \citet[Sec.\ 5]{badd18iamg}.
The broad agreement between these estimates suggests that
the parametric model would be adequate.

\begin{figure}[!h]
  \centering
  \centerline{
    \includegraphics*[width=0.35\refwidth,bb=0 0 410 420]{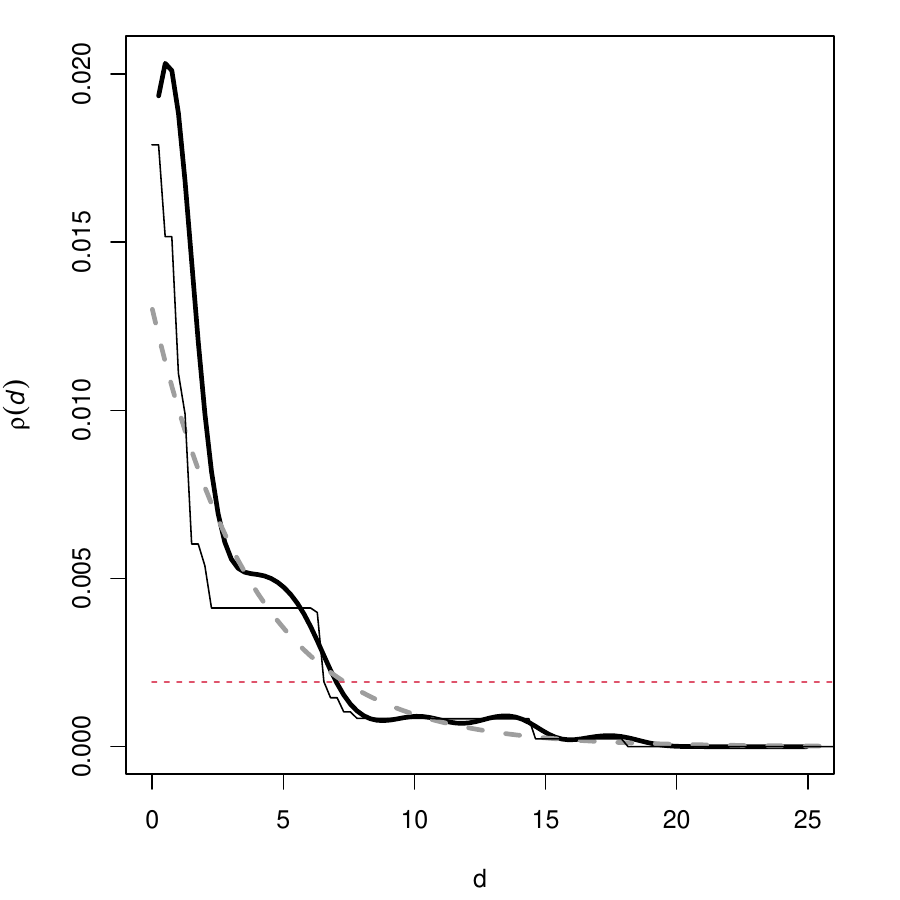}
    \includegraphics*[width=0.35\refwidth,bb=0 0 410 420]{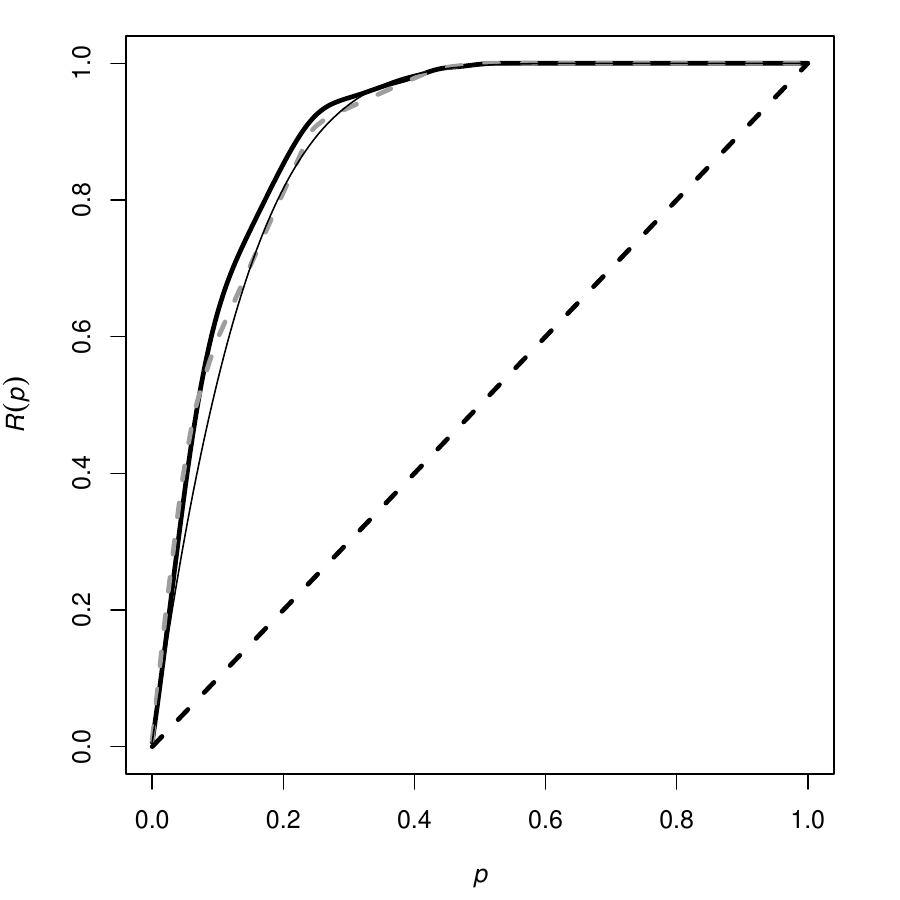}
  }
  \caption{
    \emph{Left:}
    Estimates of the function $\rho(d)$ expressing the 
    intensity of gold deposits as a function of distance to the nearest
    fault in the Murchison data.
    \emph{Right:} corresponding C-ROC curves, computed from left panel using
    Proposition~\protect{\ref{P:dR=rho}}
    equation~\protect{\eqref{eq:R=int.rho}}.
    Thick solid lines: kernel smoothing estimate
    \protect{\citep{baddchansongturn12}}.
    Thin solid lines: monotonic estimate 
    \protect{\citep{sage82,badd18iamg}}.
    Dashed lines: maximum likelihood estimate of loglinear
    function $\rho(d) = \exp(\beta_0 + \beta_1 d)$
    corresponding to logistic regression of presence-absence responses
    (\protect{\citealp[Chap.\ 9]{baddrubaturn15}},
    \protect{\citealp{baddetal10}}).
  }
  \label{F:murRho}
\end{figure}

Related concepts include the``continuous Boyce index''
proposed by \citet{hirzetal06}, which 
is an estimator of $\rho(z)/\Lambda$.

\subsection{Connection between $\rho$ and C-ROC}
\label{S:rhohat:ROC}

Here we establish a connection between $\rho$ and the C-ROC curve.
We shall focus on the continuous space (point process) case;
the discrete case (presence-absence data) is similar.

For convenience we impose regularity assumptions.
For a spatial covariate $Z$, we shall assume that 
$\FP(z)$ is differentiable. A sufficient condition
is that $Z: W \to \R$ be differentiable with nonzero gradient
\cite[Appendix A.1]{baddchansongturn12}.

\begin{prop}
  \label{P:dR=rho}
  Let $\lambda(u)$ denote the intensity function of a point process
  $\bX$, and assume that the intensity is a function
  of a single spatial covariate $Z$, that is, 
  \begin{equation}
    \label{eq:lam=rho}
    \lambda(u) = \rho(Z(u)), \ u\in W.
  \end{equation}
  Assume the regularity assumptions stated above.
  Then the theoretical C-ROC curve based on $Z$ is differentiable, with slope
  \begin{equation}
    \label{eq:dR=rho}
    \frac{{\rm d}}{{\rm d} p} R_{Z,\lambda}(p) = \frac{1}{\kappa}\rho(\FP^{-1}(p))
  \end{equation}
  where $\kappa = \Lambda/|W|$ is the mean intensity of $\bX$,
  and $\FP^{-1}$ is the left-continuous inverse function of
  $\FP(t)$ defined in \eqref{eq:FPhat:cts}.
  Here $\Lambda$ is the expected number of points falling in $W$,
  \begin{equation}
    \label{eq:Lambda}
    \Lambda = \int_W \lambda(u) \dee u 
            = \int_W \rho(Z(u)) \dee u 
            = \int_{-\infty}^\infty \rho(z) \dee F_Z(z) 
            = - \int_{-\infty}^\infty \rho(z) \dee\FP(z).
  \end{equation}
  The function in \eqref{eq:dR=rho}
  is a probability density on $[0,1]$.
  Conversely the theoretical C-ROC curve is determined by $\rho$ through
  \begin{equation}
    \label{eq:R=int.rho}
    R_{Z,\lambda}(p) = \frac 1 \kappa \int_0^p \rho(\FP^{-1}(v)) \dee v
    = - \frac 1 \kappa \int_{-\infty}^{\FP^{-1}(p)}  \rho(z) \dee \FP(z).
  \end{equation}
\end{prop}

\begin{proof}
  For any realisation $\bx$ of the point process $\bX$
  consider the corresponding values $z_i = Z(x_i)$ of the covariate
  at the data points $x_i \in \bx$. Then \cite[p.\ 22]{dalevere88}
  the values $z_i$ constitute a
  point process on the real line,
  with intensity function $h(z) = \rho(z) g_W(z)$
  where $g_W$ is the derivative of the unnormalised spatial cumulative
  distribution function of $Z$ on $W$,
  that is, $g_W = G_W^\prime$
  where $G_W(z) = |W| (1 - \FP(z))$.
  Consequently
  \[
    \TP(z) = \frac 1 \Lambda \int_z^\infty h(z) \dee z 
    = \frac 1 \Lambda \int_z^\infty \rho(z) g_W(z) \dee z
  \]
  where
  \[
    \Lambda = \int_W \lambda(u) \dee u = \int_{-\infty}^\infty \rho(z) g_W(z) \dee z
  \]
  so that $\TP$ has derivative
  \[
    \TP^\prime(z) =
    \frac{{\rm d}}{{\rm d} z} \TP(z) = - \frac 1 \Lambda \rho(z) g_W(z)
    = \frac{|W|}{\Lambda} \rho(z) \FP^\prime(z).
  \]
  Using $R_{Z,\lambda}(p) = \TP(\FP^{-1}(p))$, 
  the derivative of $R_{Z,\lambda}(p)$ is
  \begin{equation}
    \frac{{\rm d}}{{\rm d} p} R_{Z,\lambda}(p)
    =  \frac{\TP^\prime(\FP^{-1}(p))}{\FP^\prime(\FP^{-1}(p))} 
    = \frac{
      \frac{|W|}{\Lambda}
      \rho(\FP^{-1}(p)) \FP^\prime(\FP^{-1}(p))
    }{
      \FP^\prime(\FP^{-1}(p))
        } 
    =  \frac{|W|}{\Lambda} \rho(\FP^{-1}(p)).
  \end{equation}
\end{proof}

\begin{cor}
  \label{COR:ROCshape}
  Under the assumptions of Proposition~\ref{P:dR=rho},
  \begin{enumerate}
  \item 
    If the C-ROC curve is piecewise linear,
    then the $\rho$ curve is piecewise constant.
  \item
    If $\rho$ is monotone increasing, the C-ROC curve is concave.
  \item
    If $\rho$ is monotone decreasing, the C-ROC curve is convex.
  \end{enumerate}
\end{cor}

\begin{proof}
  The first statement is a direct application of \eqref{eq:dR=rho},
  or see \citet{baddetal21thresh}.
  The remaining statements are trivial applications of \eqref{eq:R=int.rho}.
\end{proof}

\begin{lemma}
  \label{L:concave:model}
  Under the regularity assumptions stated above,
  the model-predicted M-ROC curve is always concave.
\end{lemma}

\begin{proof}
  Let $\lambda(u)$ be the intensity function of the model,
  which may depend on any number of covariates $Z_1,\ldots,Z_m$.
  Define the artificial covariate $Y(u)$ to equal $\lambda(u)$.
  Then trivially $\lambda(u) = \rho(Y(u))$ where $\rho(y) = y$ is
  a strictly increasing real function. By Proposition~\ref{P:dR=rho}
  the model-predicted ROC curve of $Y$ is concave.
  But the model-predicted ROC curve of $Y$ is equivalent to the
  model-predicted model ROC curve. Hence the latter is concave.
\end{proof}

The results above were stated for the ``true'' or ``expected'' ROC curve
and the corresponding ``true'' function $\rho$,
and for the ``true'' or ``predicted'' ROC curve based on a model.
The results may also be exploited when empirical estimates
are available. For example, when empirical estimates of $\rho$ are available,
corresponding estimates of the ROC curve may be computed
using equation \eqref{eq:R=int.rho}.
The right panel of Figure~\ref{F:murRho} shows the
ROC curves computed, using equation \eqref{eq:R=int.rho},
from the estimates of $\rho$ in the left panel.

\subsection{Interpretation of results}

Proposition~\ref{P:dR=rho} provides further insight into properties of
the ROC curve. By assumption, $\rho(z)$ represents a general relationship
that does not depend on the window, and describes
how the point process depends on the covariate at every location.
In contrast, the C-ROC curve is expressed in equation \eqref{eq:R=int.rho}
in terms of $\rho$, $\kappa$ and $\FP$.
By definition, $\kappa$ and $\FP$ depend on the
choice of window $W$. This clarifies how the ROC curve depends on the
choice of window, and explains why instances of Simpson's Paradox can occur.

In principle, $\rho$ could be estimated from data observed within one region,
and the estimated relationship can be used to make predictions
for another region. This implicitly assumes that the same relationship
or ``law'' \eqref{eq:PPM} holds within both regions, which may be plausible
in some applications. Predictions made for a new region could include
predictions of the ROC curve: for this we require only the spatial CDF
of the covariate in the new region, and then we can apply
Proposition~\ref{P:dR=rho}.

\begin{figure}[!h]
  \centering
  \centerline{
    \includegraphics*[width=0.35\refwidth,bb=0 0 410 420]{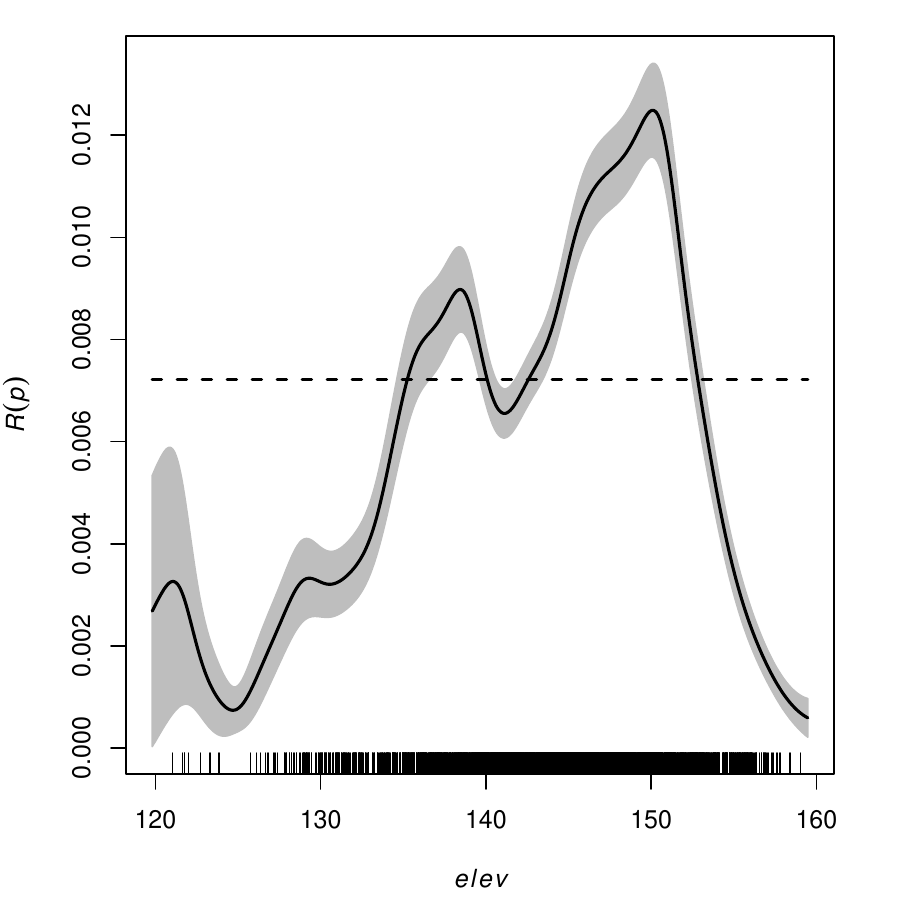}
    \hfill
    \includegraphics*[width=0.35\refwidth,bb=0 0 410 420]{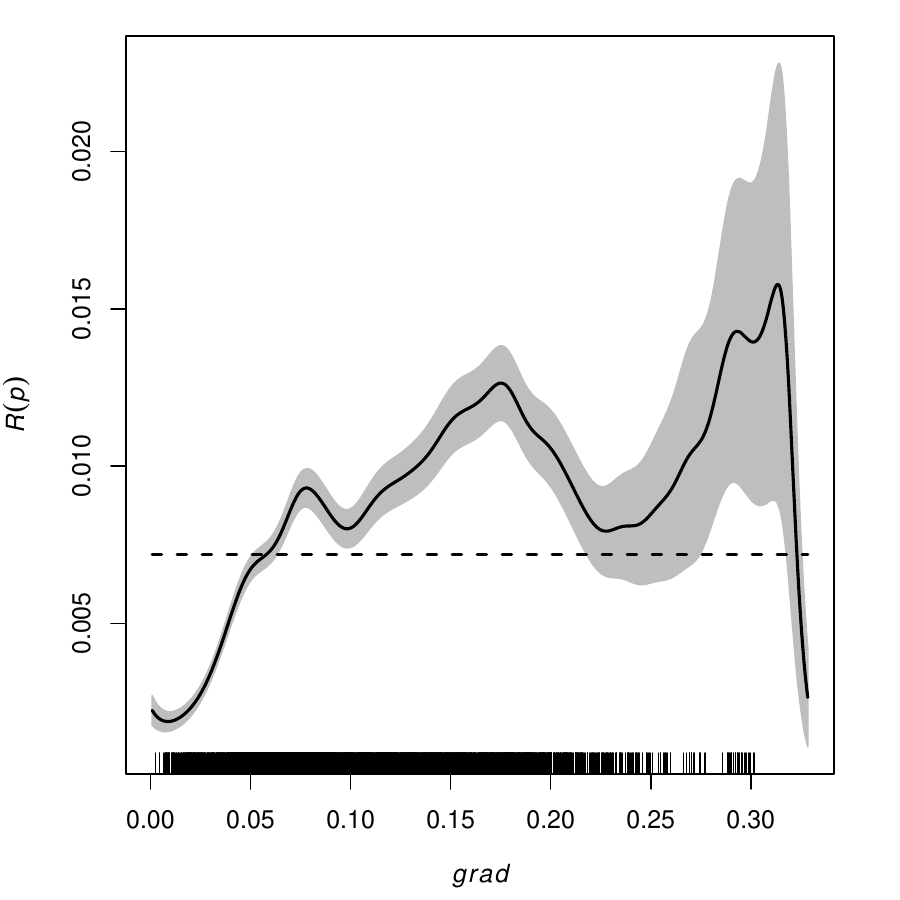}
  }
  \caption{
    Estimates of the function $\rho(d)$ expressing the 
    intensity of \textit{Beilschmiedia} trees as a function of
    terrain elevation (\emph{Left}) or as a function of
    terrain slope (\emph{Right}).
    Solid lines: kernel smoothing estimate
    \protect{\citep{baddchansongturn12}}.
    Grey shading: pointwise 95\% confidence interval.
  }
  \label{F:beiRho}
\end{figure}

Figure~\ref{F:beiRho} shows estimates of $\rho$ for the
\textit{Beilschmiedia} trees as a function of a single covariate,
terrain elevation in the left panel and terrain slope in the right panel.
The ``rug plots'' along the horizontal axes show the observed values of the
covariate $Z$ at the data points, and the grey shading shows the pointwise
asymptotic 95\% confidence interval described in \cite{baddchansongturn12}.
Refer to Figure~\ref{F:bei10ROCmur} for the corresponding ROC curves.

The left panel of Figure~\ref{F:beiRho} indicates that 
forest density is much higher at higher elevations, except for the very
highest elevations. The right panel shows that forest density is very low
on flat terrain, and that there is a slightly greater forest density on
more sloping terrain. These observations 
could be influenced by the topography and hydrography 
of the current study area: the highest elevation in the study area
is achieved on a plateau, and the lowest elevation appears to be a river,
where trees of this type do not grow. 

The left panel of Figure~\ref{F:beiRho} suggests that
forest density is not a monotone increasing function
of terrain elevation. 
In the light of Corollary~\ref{COR:ROCshape},
this is confirmed retrospectively by
the left panel of Figure~\ref{F:bei10ROCmur}. 

These examples militate in favour of using estimates of $\rho$,
rather than the ROC curve, for the initial exploration and modelling
of dependence on covariates.
The ROC is cumulative, and somewhat difficult to interpret,
while $\rho$ is proportional to the derivative of ROC, and is more
straightforward. It is important to assess whether $\rho(z)$ can be
assumed to be a monotone function of $z$.
The true merit of $Z$ may be considerably greater than that
  suggested by AUC if $\rho(z)$ is not monotone.
If there is only a single covariate $Z$, then we should use ROC
  if there's good reason to believe that increasing values of $Z$
  are increasingly favorable. Otherwise, we should use $\rho(z)$ to
  investigate. (Or compare empirical ROC with that
  based on monotone regression of $\rho(z)$ (equivalent to convex
  regression of ROC): discrepancies suggest non-monotone $\rho(z)$.

An advantage of ROC over $\rho$ is that ROC gives a sense of how
  much (area fraction) of the spatial domain is
  involved in each peak or trough of $\rho$.

\section{Extensions of ROC curves}
\label{S:extensions}

\subsection{Restriction to a subset}
\label{ss:restriction}

As mentioned in Sections \ref{S:ROCcovar:region} and \ref{S:ROCmodel:subregion},
both the C-ROC and M-ROC depend on the spatial domain over which
they are calculated.
This important fact can be turned to our advantage.
Calculation of the ROC and AUC can be restricted
to a subset $B$ of the spatial domain $W$,
simply by restricting the domains of summation and integration
in any of the definitions \eqref{eq:TPhat:finite}--\eqref{eq:FPhat:finite},
\eqref{eq:TPhat:cts}--\eqref{eq:FPhat:cts},
\eqref{eq:TPFPmodel:finite},
\eqref{eq:TPFPmodel:cts} or
\eqref{eq:TPFPM:cts}
Restriction to subsets can be useful
in the same way that it is useful to break down an aggregate summary statistic
into summaries for sub-populations.

\begin{figure}[!h]
  \centering
  \centerline{
    \includegraphics*[width=0.4\refwidth,bb=0 0 410 420]{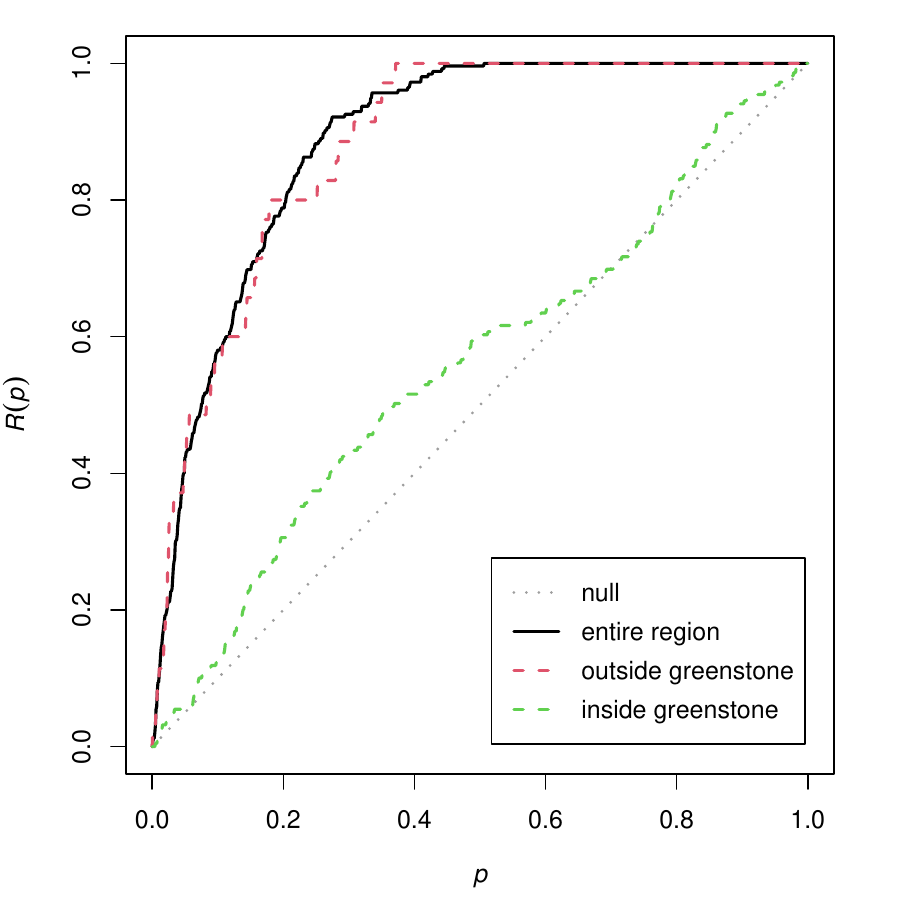}
  }
  \caption{
    Empirical C-ROC curves $\lo R_{Z, \bx}(p)$ for subsets of the Murchison survey.
    \emph{Black:} entire region;
    \emph{red:} outside greenstone outcrop;
    \emph{green:} inside greenstone outcrop.
    Covariate is distance to nearest fault, with small distances
    interpreted as more favorable to gold.
  }
  \label{F:murROCsubsets}
\end{figure}

Figure~\ref{F:murROCsubsets} shows the empirical ROC curves
for subsets of the Murchison survey, namely the subsets inside
and outside the greenstone outcrop. Within the greenstone outcrop, the
distance to nearest fault is not very informative. 
However, outside the
greenstone outcrop, the distance covariate
it is highly informative, and short distances are
especially highly prospective for gold.

These findings are expected if we consider the three-dimensional geological
structure near the surface. Our data are the two-dimensional projected locations
of the gold deposits, and our covariates are the two-dimensional
surface profiles of the solid greenstone and the fault planes.
Assume that, in three dimensions, gold deposits are highly likely
to lie inside greenstone, which in turn is highly likely to occur
near a geological fault.
Then in two dimensions, a location inside the greenstone outcrop
is highly prospective for gold; a location outside the greenstone
outcrop, but close to a geological fault, is more likely to lie above
greenstone below the surface, so it is highly prospective for gold.

\subsection{Weights on data points}
\label{ss:weighted}

For spatial point pattern data, numerical weights $w_i \ge 0$
may be attached to the individual data points $x_i$, $i=1,\dots,n$.
The weights could represent multiplicity or severity
(for example, number of people involved in each road accident),
cost or value (the total amount of gold in each gold deposit,
the estimated volume of each tree). They can also be used to adjust
for uneven sampling effort or uneven probability of detection.

Weighted versions of the ROC curves can be defined by introducing
numerical weights $w_i$ into the calculation of the numerators
of the true positive rates in each of the equations
\eqref{eq:TPhat:cts},
\eqref{eq:TPFPmodel:cts}
and \eqref{eq:TPFPM:cts}.
In the denominators, $n = n(\bx)$ is replaced by $\sum_i w_i$.
In equation \eqref{eq:TPFPM:cts},
the intensity $\lambda(u)$ of the point process $\bX$
is replaced by the intensity of the random measure
assumed to have generated the points and weights
(technical details are beyond the scope of this article).

For an intuitive appreciation of the weighted ROC, suppose that
the data points are the home addresses of disease cases, and the weight
at each address $x_i$ is the number of cases at that address. Then the
weighted ROC of a spatial variable $Z$ for the weighted point pattern
of home addresses is equivalent to the un-weighted ROC of $Z$ for the
individual cases. The weighted AUC is the probability that a
randomly-selected \emph{case} has a higher value of $Z$ than
a randomly-selected spatial location.

\paragraph{Accounting for uneven sampling effort or detection probability}

In some applications the probability of detecting a point
is not constant but spatially varying. This may be a consequence of the
detection technique, or may occur if the space was explored
with uneven sampling effort. This contributes a bias to the observed
``true positive rate'' $\TPhat(t)$.

The bias can be corrected using the Horvitz-Thompson device,
that is, by weighting each observed point by the reciprocal of its
detection probability, $w_i = 1/q_i$ where $q_i$ is the probability that
a point at location $x_i$ would have been detected. This probability must be
known or calculable for all the observed points $x_i$, and must be nonzero
for all spatial locations. This weight is applied to the calculation
of $\TPhat(t)$ while the calculation of $\FPhat(t)$ is not changed.

For the ROC curve of a fitted model,
the fitted model should also take the uneven sampling effort into account.
In logistic regression or loglinear Poisson regression, the model should
include an offset term equal to the logarithm of the sampling effort.

\paragraph{Weighted presence-absence data}

It is also possible to introduce weights for presence-absence data by multiplying the indicator function by an individual weight for each of the presence pixels in the calculation of the true positive rate.
The weight may again indicate e.g. the number of cases in that pixel or the estimated volume of wood in that pixel.

\subsection{ROC relative to a baseline}
\label{S:ROC.baseline}

As previously mentioned for the C-ROC for point pattern data in Section~\ref{ss:rocData.cts} we have that $1-\FP(t)$ is the c.d.f.\ of the covariate value $Z(U)$ at a random location $U$ uniformly distributed in $W$.
However, in some contexts it may be more natural to consider another reference distribution than the uniform.
For example, consider the cases of a disease where the density of the susceptible population is proportional to a given baseline function $b(u)$ for $u\in W$.
Then we define $R_{Z,\bx,b}(p)$ the ``C-ROC relative to baseline $b(u)$'' as the ROC curve with false positive rate
\begin{align}
  \label{eq:FP:cts:b}
  \widehat\FP_{b}(t) &= \frac{
               \int_{W} \indicate{Z(u) > t} b(u) \dee u
               }{
               \int_{W} b(u) \dee u
               }.
\end{align}
and unchanged true positive rate \eqref{eq:TPhat:cts}.
This means that $1-\widehat\FP_b(t)$ is the c.d.f.\ of the covariate value
at a randomly selected location
with nonuniform probability density proportional to $b(u)$.
In the context of disease data, the horizontal axis is the ``cumulative fraction of susceptibles'' rather than the ``cumulative fraction of the survey area''.
Note that this only depends on $b(u)/\int_W b(v) \dee v$, so it is invariant under rescaling of the baseline.
To illustrate the point, consider the mucosa data shown in Figure~\ref{F:mucosa} in Section~\ref{S:background:spatial:casecontrol}.
We treat the ECL cells as the cases of interest and we can use the kernel
estimate of the intensity of non-ECL cells shown in
the left panel of Figure~\ref{F:mucorocDensityOther}
as the baseline proportional to the reference population density.
\begin{figure}[!h]
  \centering
  \centerline{
    \includegraphics*[width=0.52\refwidth,bb=65 110 415 350]{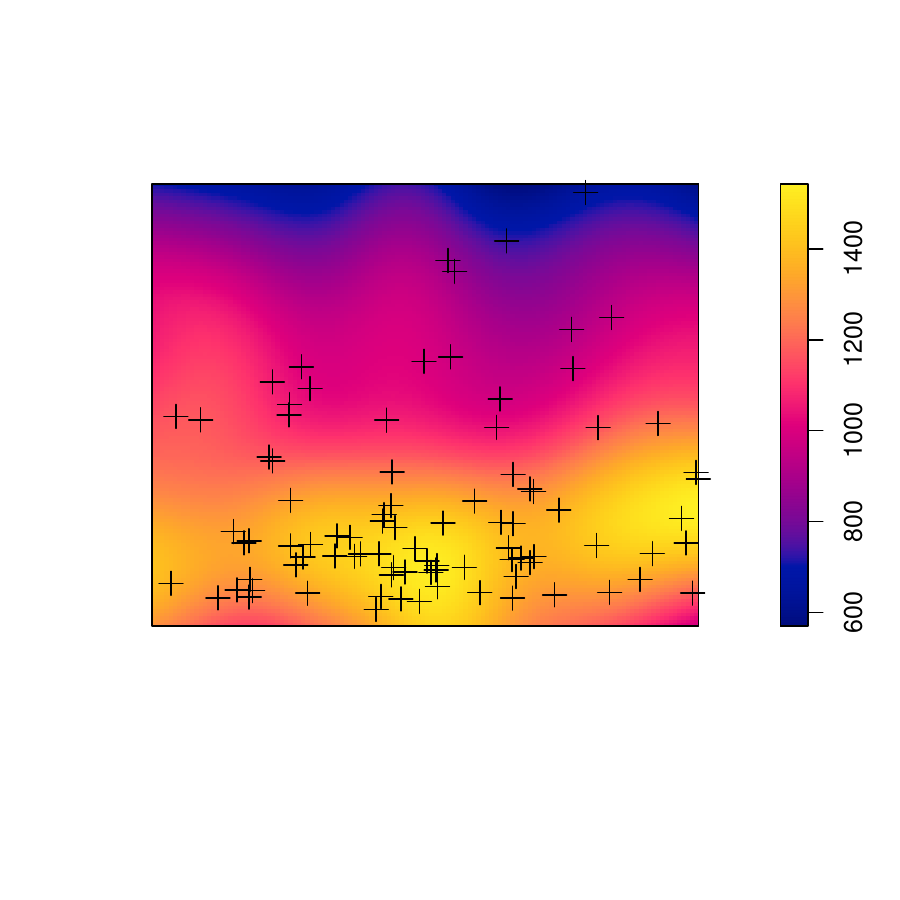}
    \hfill
    \includegraphics*[width=0.35\refwidth,bb=0 0 410 420]{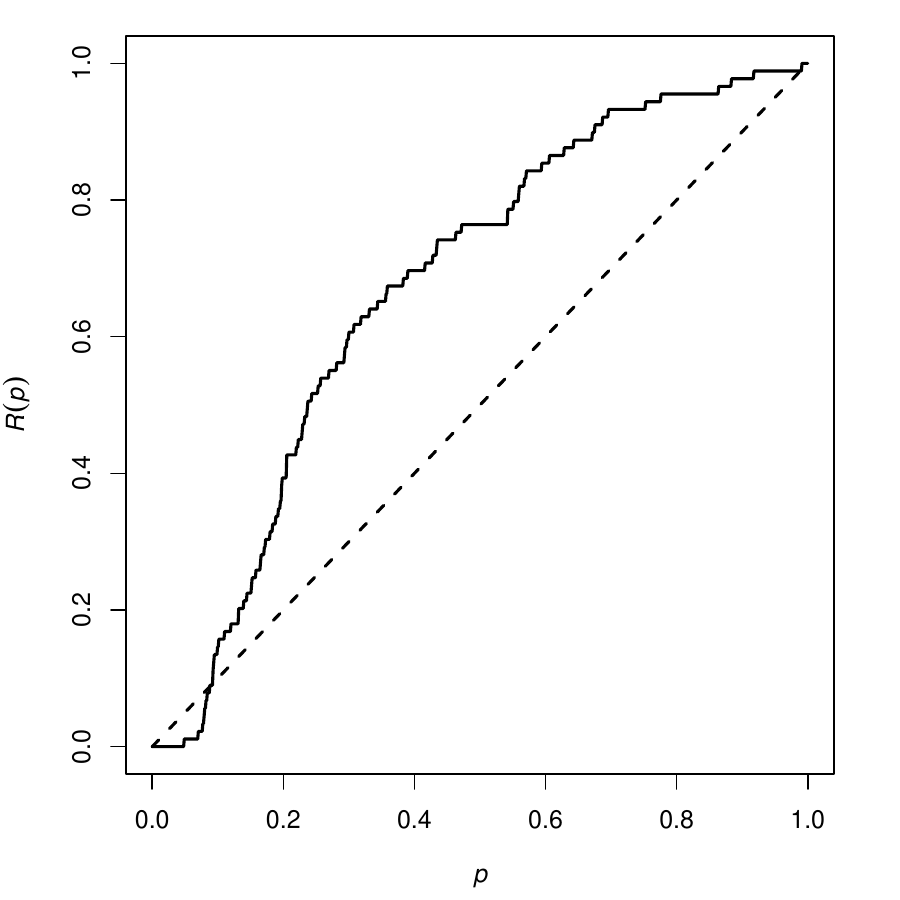}
  }
  \caption{
    Analysis of gastric mucosa data of Figure~\protect{\ref{F:mucosa}}
    using baseline-adjusted C-ROC.
    \emph{Left:} Kernel density estimate of intensity of non-ECL
    cells; ECL cells overlayed as black crosses.
    \emph{Right:} Baseline-adjusted C-ROC curve for
    distance to the stomach wall ($y$-coordinate) treating
    short distances as favorable for ECL cells,
    and using the intensity estimate of non-ECL cells as baseline.}
  \label{F:mucorocDensityOther}
\end{figure}%
The right panel of Figure~\ref{F:mucorocDensityOther} shows $R^<_{Z,\bx,b}$, where $Z$ is distance to stomach wall (y-coordinate), $\bx$ is the pattern of ECL cells, and $b(u)$ is the kernel density estimate of the intensity of non-ECL cells.
It clearly shows that (apart from very short distances) ECL cells appear closer to the stomach wall than cells in general.
Another way to think about this example is that the baseline $b(u)$ appears as the intensity of a
reference model, which stipulates that the intensity of ECL cells is proportional to the kernel density estimate of the intensity of non-ECL cells (or any parametric model for this intensity for that matter).
In general terms when $b(u)$ is the intensity of a reference model the interpretation of the horizontal axis of the ROC curve is ``cumulative fraction of cases/deposits/events predicted to be captured by reference model''.
Other appropriate examples of the baseline $b(u)$ include
(1) the Jacobian of the spatial coordinate system
(so the horizontal axis is 'cumulative fraction of true area');
(2) the spatially-varying cost of exploration
(so the horizontal axis is 'cumulative fraction of total cost of exploration');
and (3) the micro-scale percentage of area that is habitable or usable,
for example the spatially-varying fraction of area occupied by small patches of dry land in a swamp (so the horizontal axis is 'cumulative fraction of
habitable area').

The baseline could also be the fitted intensity of a model
that does not involve the covariate $Z$.
This idea is explored further in Section~\ref{ss:dropROC} below. 
 
For pixel data as in Section~\ref{ss:rocData.pixel} we can make an analogous definition of the ROC relative to baseline weights $b_j \ge 0$, for each pixel $Q_j$.
Define the baseline-adjusted false positive rate by
\begin{align}
  \label{eq:falsepos:b:copy}
  \widehat\FP_b(t) &=  \frac{
             \sum\nolimits^{J}_{j=1} b_j (1-y_j) \ \hat y_j
             }{
             \sum\nolimits_{j=1}^{J} b_j (1-y_j)
             },
\end{align}
where, as before, $y_j$ is the presence-absence indicator and $\hat y_j = \indicate{z_j > t}$ is the classifier.
Combining this with the usual empirical true positive rate $\TPhat(t)$ from \eqref{eq:TPFPhat:covar}
yields the C-ROC curve $R_{Z,\by,\mathbf{b}}$ where $\mathbf{b} = (b_1, \dots, b_J)$ is the vector of baseline weights for each pixel.
In the same way any of the other definitions of ROC curves (M-ROC, model-predicted ROC curves, ...)
can be relative to a baseline by introducing pixel weights or a weight function in the definition
of the false positive rate.

\subsection{Partial ROC curve for dropping or adding an explanatory variable}
\label{ss:dropROC}

When we have fitted a model involving several covariates,
we can construct an array of plots representing the effect
of dropping each existing term in the model.
Each panel is an ROC plot for the covariate in question,
with the baseline being the fitted intensity when this covariate is
removed. An ROC curve lying close to the diagonal line suggests that the
covariate can be dropped from the model.
This is analogous to a partial residual plot, so we shall call it
the \emph{partial ROC}, and the associated AUC value is the \emph{partial AUC}.

\begin{figure}[!hbt]
  \centering
  \centerline{
    \includegraphics*[width=0.3\refwidth,bb=0 0 410 420]{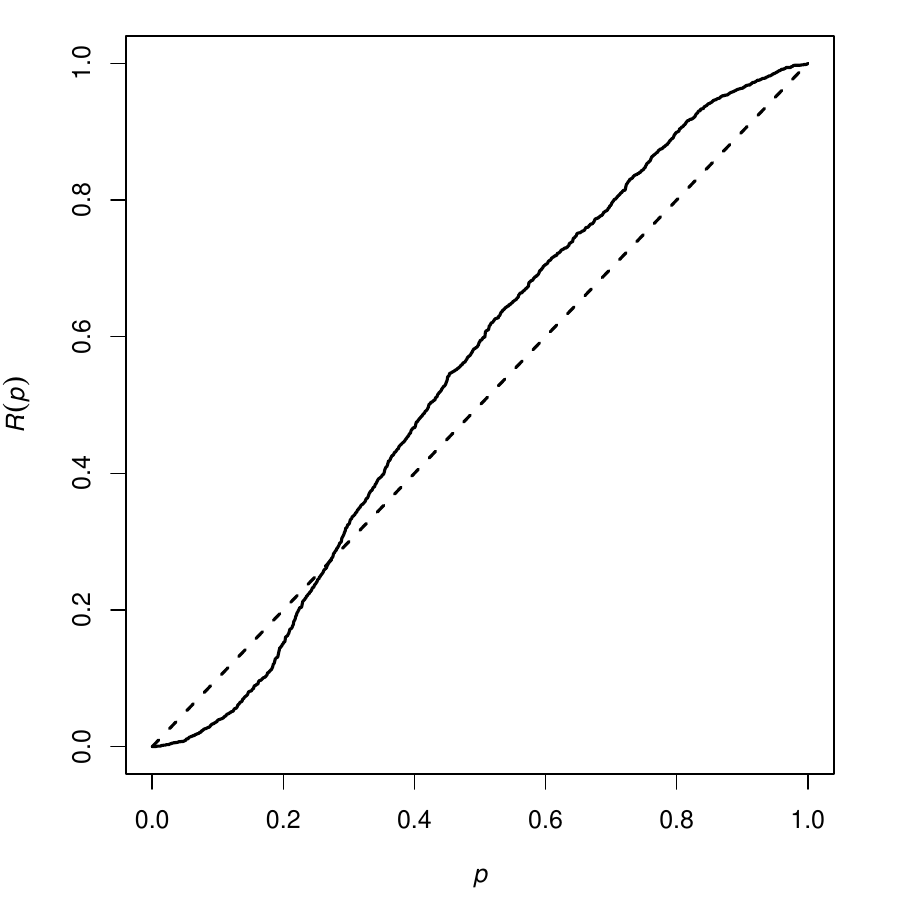}
    \includegraphics*[width=0.3\refwidth,bb=0 0 410 420]{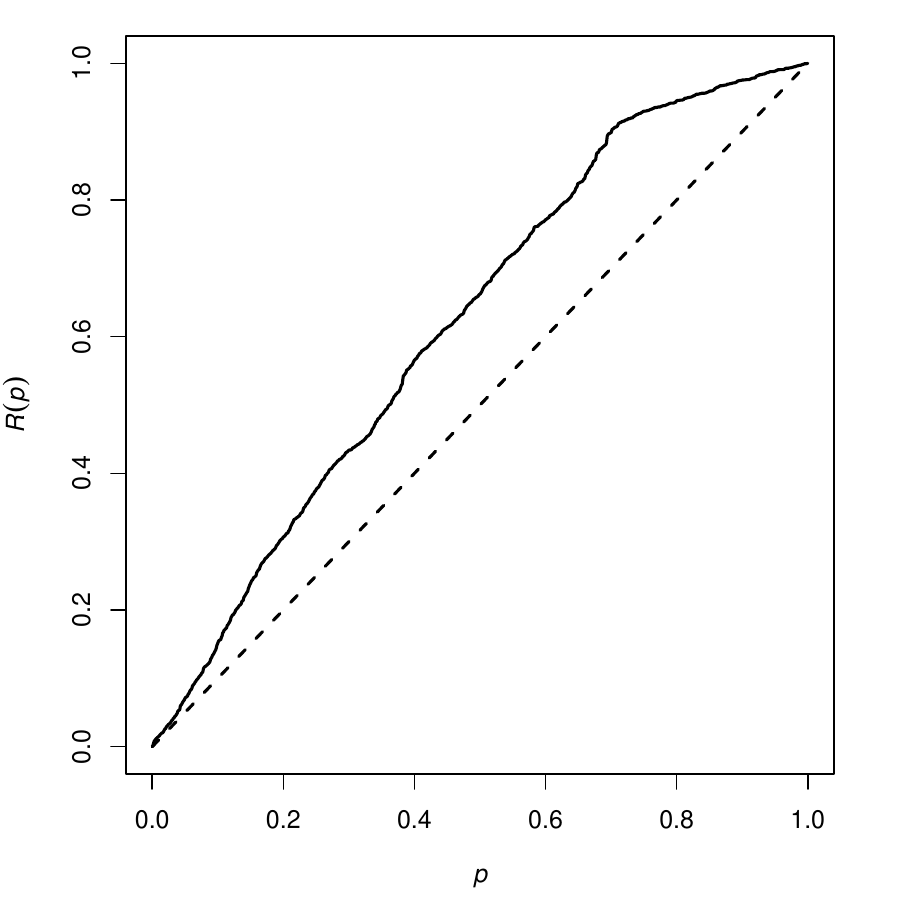}
  }
  \caption{
    Partial ROC curves for dropping covariates from a fitted model.
    Poisson point process, additive loglinear model
    for the \textit{Beilschmiedia} data
    involving the covariates \texttt{Elevation} and \texttt{Gradient}.
    \emph{Left:} ROC curve for dropping \texttt{Elevation}.
    \emph{Right:} ROC curve for dropping \texttt{Gradient}.
  }
  \label{F:beidrop}
\end{figure}

Figure~\ref{F:beidrop} shows the partial ROC curves for dropping the
Elevation and Gradient variables, respectively, from the 
additive loglinear Poisson model
for the \textit{Beilschmiedia} data
(treating large values as favorable to trees, in both panels).
It suggests that the Elevation variable could be dropped
without compromising the model.
corresponding partial AUC values are 0.54 for elevation and 0.62 for slope.
These plots are visually very similar to the C-ROC plots with
constant baseline; maximum discrepancies are 0.072 for elevation,
0.026 for slope.

Similarly, given a fitted model and a list of covariates which were \emph{not}
included in the fitted model, we can construct an array of ROC plots
representing the effect of \emph{adding} each one of the additional covariates
to the fitted model. For each curve, the baseline is the intensity
of the \emph{original} fitted model.
A partial ROC curve lying close to the diagonal line
suggests that the covariate does not improve the ranking ability of the
fitted model and should not be added to the model. 

\begin{figure}[!hbt]
  \centering
  \centerline{
    \includegraphics*[width=0.3\refwidth,bb=0 0 410 420]{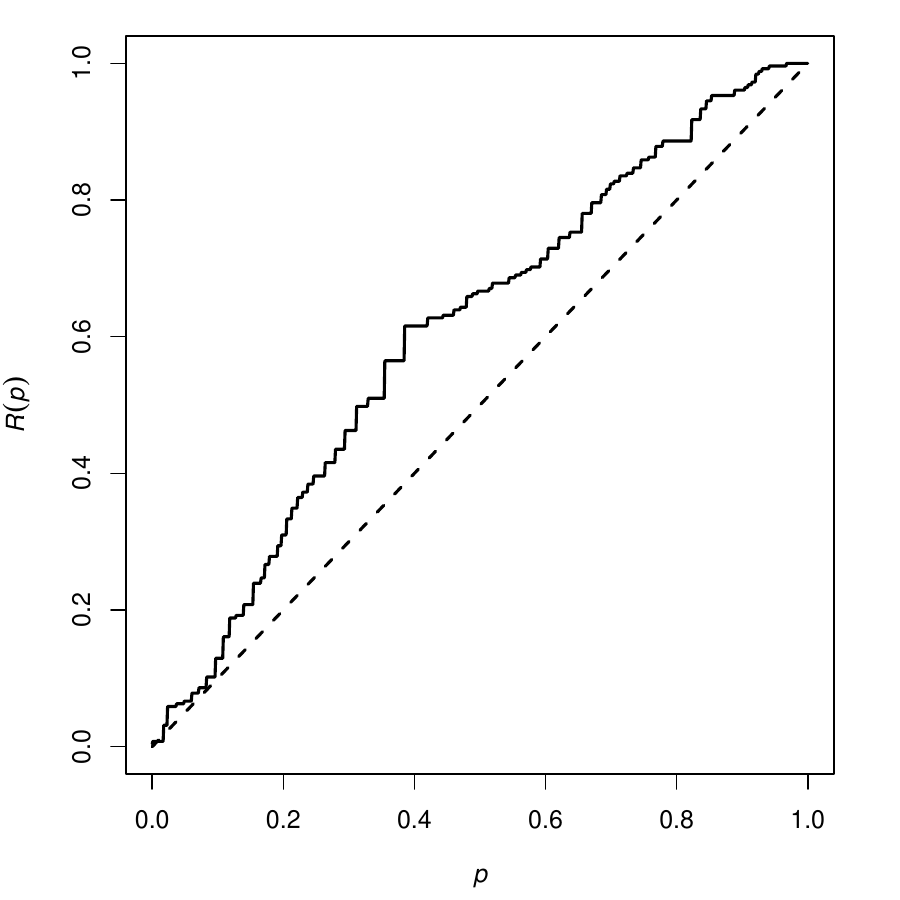}
    \includegraphics*[width=0.3\refwidth,bb=0 0 410 420]{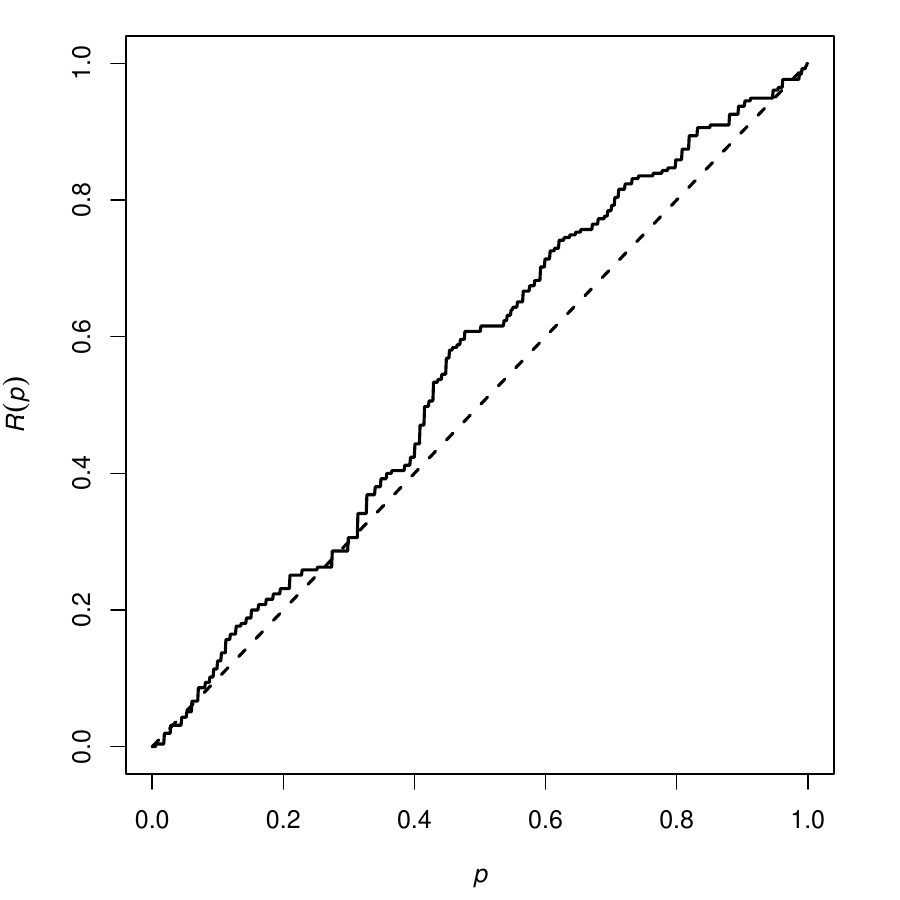}
  }
  \caption{
    Partial ROC curves for adding covariates to a fitted model.
    Poisson point process, additive loglinear model
    for the Murchison gold data 
    involving greenstone indicator and distance to nearest fault.
    \emph{Left:} ROC curve for adding the $x$ coordinate.
    \emph{Right:} ROC curve for adding the $y$ coordinate.
  }
  \label{F:muraddxy}
\end{figure}

Figure~\ref{F:muraddxy} shows the partial ROC curves for adding the
$x$ (Easting) and $y$ (Northing) coordinates, considered as spatial covariates,
to the additive loglinear model for the Murchison data involving
both the greenstone indicator and distance to nearest fault.
Corresponding partial AUC values are 0.61 and 0.55.
The left panel suggests that adding the $x$ coordinate as a covariate
would improve the ranking performance of the model, while the right panel
provides only weak support for adding the $y$ coordinate.
Since $x$ and $y$ coordinates are always available, plots similar to
Figure~\ref{F:muraddxy} could be useful as a model diagnostic.

\section{Model checking using ROC}
\label{S:modelcheck}

There are numerous claims that ROC curves can be used
for model validation \citep{nykaetal15}. Our reading of this literature
is that the term ``validation'' is not being used in the
statistical sense: rather, a model is said to be ``validated''
if its ranking performance is found to be sufficiently high. 

It is shown above that the empirical M-ROC does not itself contain
diagnostic information about the goodness-of-fit of the model. Instead,
the shape of the empirical M-ROC curve is largely indicative of the
spatial variation in the fitted model probabilities or intensity function.

Informal model diagnostics are typically based on comparison between 
observed and predicted values of a statistic. This suggests that
a comparison between the empirical and model-predicted ROC curves
would enable diagnostic checking of models.
This approach could be applied either to C-ROC or M-ROC curves.

For brevity we assume the spatial data are a point pattern $\bx$ in continuous
space (but similar comments apply in other cases).
Suppose a point process model has been fitted to $\bx$ yielding
fitted intensity function $\widehat\lambda$.

\subsection{Model checking using C-ROC}

First consider C-ROC curves. For any point process model with
intensity function $\lambda(u)$, the theoretical C-ROC curve 
$R_{Z, \lambda}$ for a covariate $Z$ is based on the theoretical
true positive rate \eqref{eq:TP:true:cts} and
false positive rate \eqref{eq:FPhat:cts}. For a fitted model with
intensity $\widehat\lambda(u)$, the ``model-predicted'' C-ROC curve
$R_{Z,\widehat\lambda}$ was defined similarly in Section~\ref{ss:rocData.cts}
by replacing $\lambda$ by $\widehat\lambda$ in \eqref{eq:TP:true:cts}.

For any spatial covariate $Z$ we could compare the empirical C-ROC curve
$R_{Z,\bx}$ and the model-predicted C-ROC curve $R_{Z,\widehat\lambda}$.
If the model
is correct, with true intensity $\lambda$, then under suitable conditions
under a large sample regime, both the functions
$R_{Z,\bx}$ and $R_{Z,\widehat\lambda}$
converge to $R_{Z,\lambda}$ as distribution functions
(i.e.\ pointwise at every continuity point of $R_{Z,\lambda}$).
Hence the discrepancy between the empirical C-ROC and model-predicted C-ROC
should be small.

If the model is misspecified,
and the true process has intensity $\lambda^\#$, then under suitable conditions
$R_{Z,\bx}$ converges to $R_{Z,\lambda^\#}$
but $R_{Z,\widehat\lambda}$ converges to $R_{Z,\lambda^\ast}$,
where $\lambda^\ast$ is determined by the behaviour of the model-fitting
procedure applied to the true point process. A discrepancy between the
empirical and model-predicted C-ROC curves
suggests that the model is misspecified. 

This technique can be applied to any covariate $Z$,
whether or not the model involves $Z$.
It may be useful for detecting departures from the model due to effects
that are not included in the model.

\begin{figure}[!htb]
  \centering
  \centerline{
  \includegraphics*[width=0.35\refwidth,bb=0 0 410 420]{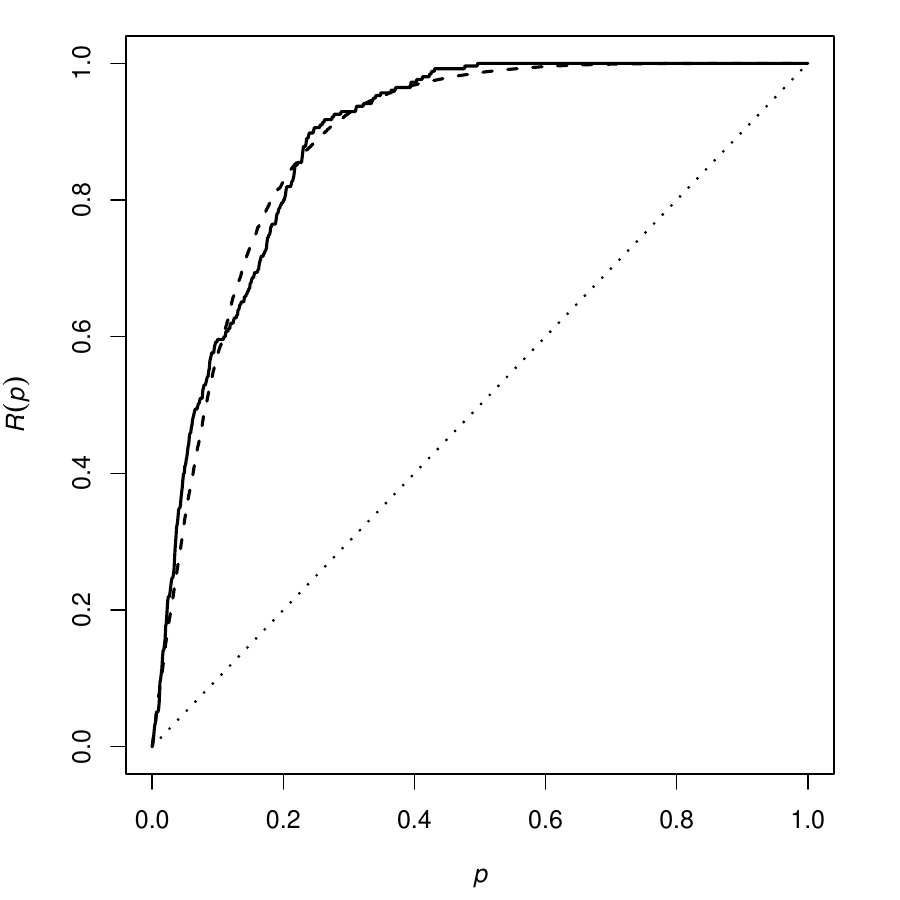}
    \includegraphics*[width=0.35\refwidth,bb=0 0 410 420]{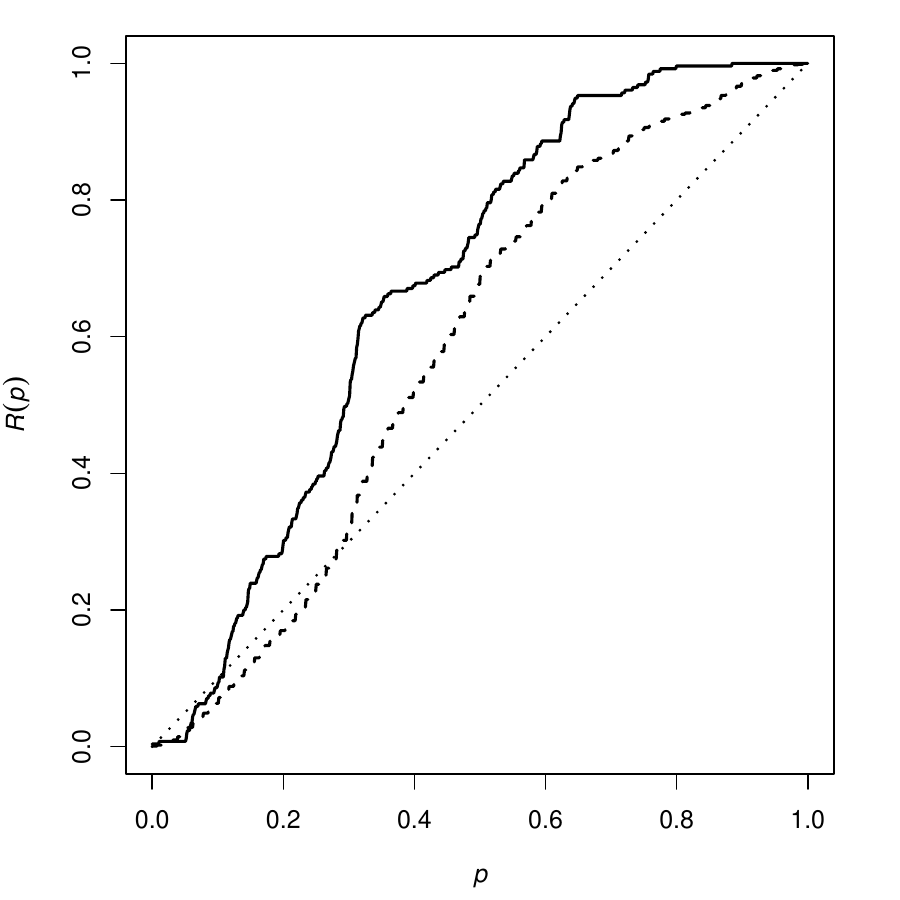}
  }
  \caption{
    Diagnostics using C-ROC.
    Model fitted to the Murchison data in which log intensity is
    an additive linear function of distance to nearest fault and of
    greenstone indicator.
    \emph{Left:}
    Empirical C-ROC curve $R_{Z,\bx}$ (solid lines)
    for the distance-to-nearest-fault covariate $Z$ in the Murchison data,
    and predicted C-ROC curve
    $R_{Z, \widehat\lambda}(p)$ (dashed lines)
    for the same covariate, suggesting a good fit.
    \emph{Right:}
    Empirical C-ROC curve $R_{Z,\bx}(p)$ (solid lines)
    and model-predicted C-ROC curve
    $R_{Z, \widehat\lambda}(p)$ (dashed lines)
    for the Murchison gold data,
    where the covariate $Z$ is the $x$ (Easting) coordinate,
    suggesting a mediocre fit.
  }
  \label{F:murROC:diagC}
\end{figure}

Figure~\ref{F:murROC:diagC} shows two applications of this diagnostic
to the Murchison data. The fitted model is a Poisson point process
in which log intensity is an additive linear function of
distance to nearest fault and greenstone indicator.
The left panel of Figure~\ref{F:murROC:diagC}
shows the empirical C-ROC and the model-predicted C-ROC
for the distance to nearest fault,
suggesting a good fit.
The right panel shows the empirical C-ROC and model-predicted C-ROC
for the $x$ (Easting) coordinate.
This plot suggests that this model does not capture the fact that
mineral deposits are relatively more frequent on the
eastern side of the study region.

\subsection{Model checking using M-ROC}

An alternative technique is to compare the empirical M-ROC $R_{\widehat\lambda,\bx}$
and the model-predicted M-ROC $R_{\widehat\lambda,\widehat\lambda}$
for a fitted model. If the model is true then under suitable conditions
these two curves both converge to the theoretical M-ROC $R_{\lambda,\lambda}$
for the model. As shown in Lemma~\ref{L:concave:model},
this limiting curve is concave.
If the model is misspecified then (with the same notation as above)
$R_{\widehat\lambda,\bx} \to R_{\lambda^\ast,\lambda^\#}$
while $R_{\widehat\lambda,\widehat\lambda} \to R_{\lambda^\ast,\lambda^\ast}$
so that there is a nonzero discrepancy, and the empirical M-ROC
may converge to a function which is not concave. Deviation between the two
curves, or non-convex shape, suggests that the model is not correct.

\begin{figure}[!htb]
  \centering
  \includegraphics*[width=0.35\refwidth,bb=0 0 410 420]{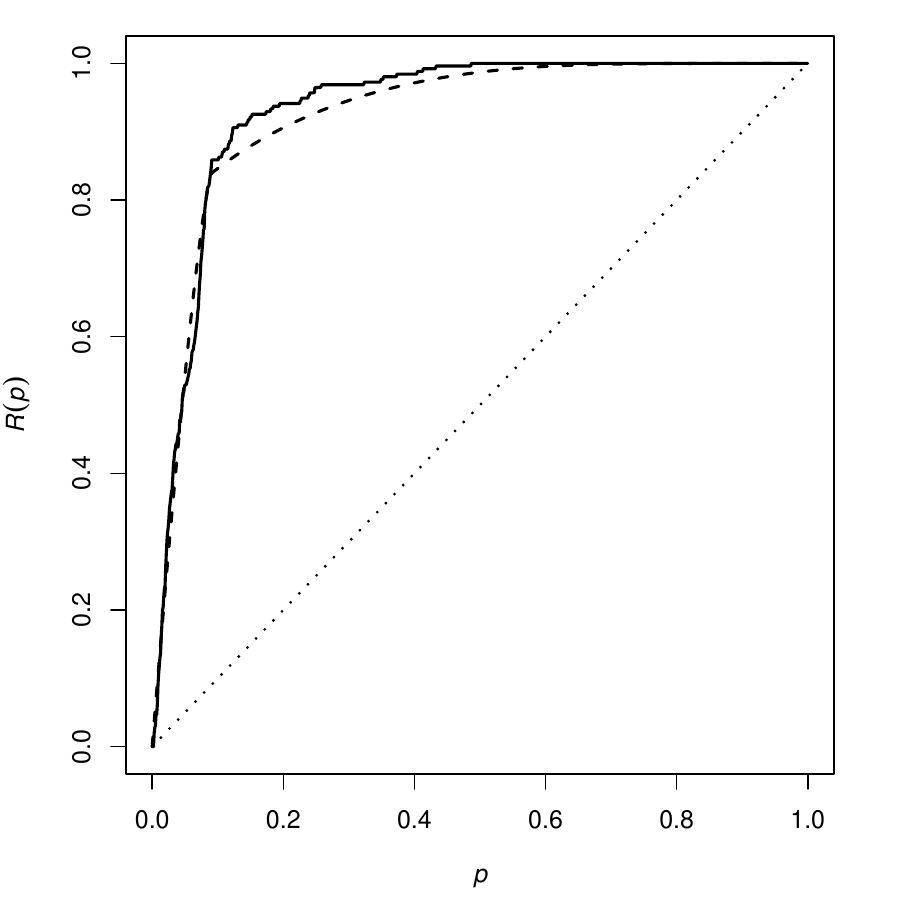}
  \caption{
    Empirical M-ROC curve $R_{\widehat\lambda,\bx}(p)$ (solid lines)
    and model-predicted M-ROC curve
    $R_{\widehat\lambda, \widehat\lambda}(p)$ (dashed lines)
    for the Poisson model of the Murchison gold deposits, in which
    log intensity is a linear function of distance to nearest fault
    and greenstone indicator.
  }
  \label{F:murROC:model.valid}
\end{figure}

Figure~\ref{F:murROC:model.valid} shows the
M-ROC curve $R_{\widehat\lambda,\bx}(p)$ (solid lines)
and model-predicted M-ROC curve
$R_{\widehat\lambda, \widehat\lambda}(p)$ (dashed lines)
for the Poisson model in which
log intensity is a linear function of distance to nearest fault
and the greenstone indicator.
This suggests a good fit.

These diagnostic tools may be useful for detecting departures from a
fitted model. Their weaknesses are the same as those described above
for ROC curves in general.
For example, suppose that the model is correct except for misspecification
by an increasing transformation.
That is, the true process has intensity $\lamtrue(u)$
while the intensity of the model is
$\lambda(u) = \psi(\lamtrue(u))$ where $\psi$ is
an increasing transformation, possibly depending on parameters of the model.
Then the ROC is the same as if the model were completely correct:
$
R_{\psi \circ \lamtrue, \lamtrue}(p)
\equiv
R_{\lamtrue, \lamtrue}(p)$
and when the law of large numbers is a good approximation,
$R_{\hat\lambda, \bx}(p) \equiv R_{\lamtrue, \bx}$.
Comparison of the empirical and model-predicted C-ROC curves
cannot distinguish this kind of misspecified model
from the correct model. Perfect agreement between
$R_{\hat\lambda, \bx}(p)$ and 
$R_{\hat\lambda, \hat\lambda}(p)$ does not prove the model is correct.

\section{Discussion}
\label{S:discussion}

\subsection{Summary of main findings}
\label{S:discussion:findings}

This article has four main findings.
First, we have established that
the current use of ROC curves to evaluate model performance is partly erroneous.
Second, the correct interpretation of ROC curves has been found
by elucidating their connection to other statistical tools.
Third, we have emphasised that there are many versions of the ROC curve,
and have proposed new ROC curves for spatial data,
allowing better interpretation and expanding the scope of application.
Fourth, we have proposed new practical techniques
for valid evaluation of model performance
by comparing empirical and fitted ROC curves.

The first finding concerns the correct interpretation of ROC curves
for a fitted model. In the applications literature, ``the'' ROC curve for a
species distribution model is understood to mean the ROC curve based on the fitted
presence probabilities under the model (Section~\ref{S:ROC:model}). We call this the ``M-ROC''.
We find that the M-ROC is invariant under monotone transformation of model predictions
(probabilities or intensities) (Section~\ref{S:ROCmodel:monotone}).
This implies that the M-ROC cannot distinguish between two models
based on the same covariate (if both models are based only on that covariate, and
are both monotone).
If a model involves only one covariate, then the M-ROC collapses to
the C-ROC, a curve constructed without reference to the model (Section~\ref{S:ROCmodel:collapse})
and therefore has nothing to do with the goodness-of-fit of the model.
We demonstrated that ROC or AUC is not a measure of goodness-of-fit of a fitted model
(Section~\ref{S:ROCmodel:badness}); rather, it 
is a measure of ``badness-of-fit'' of the null model. The use of ROC in this context
contravenes the standard practice
that a test of goodness-of-fit should be based on a covariate that was not included in the model.
The M-ROC is a measure of ``ranking ability'' rather than ``predictive power'' .
The M-ROC is completely bound to the study region
(Sections~\ref{S:ROCcovar:region} and \ref{S:ROCmodel:subregion}),
mostly because of inhomogeneity in the study region.
It is not possible to extrapolate the ROC from one region to another;
restriction to a sub-region can lead to Simpson's paradox;
restriction to a more homogeneous sub-region lowers the AUC;
ROC cannot be used to predict the response to changes in the covariate value
(such as climate change).
ROC and AUC are also insensitive to dependence on other covariates
not included in the model (Section~\ref{S:tests:insensitivity})

The second finding concerns the connections between ROC curves
and other statistical tools, in the context of spatial data.
It is well known that ROC curves are closely related to P-P plots
and other tools for distributional comparison.
In the case of a single covariate, 
we show that the AUC and Youden index are related to the Berman-Waller-Lawson test
and Kolmogorov-Smirnov test, respectively, of the null hypothesis of uniformity
(Section~\ref{S:tests}). That is, they are measures of the degree of departure from the null
model, rather than measures of agreement with the fitted model.
We also established connections between ROC curves and the
function $\rho$ (the resource selection function or prospectivity index)
which expresses the presence probability or intensity as a function of the covariate
(Section~\ref{S:rhohat}). The function $\rho$ can be interpreted as a law or model
predicting the probability of occurrence as a function of the covariates, in a form
which could be extrapolated from one region to another, while the ROC cannot.
These connections provide insight and
implies that these different techniques are not independent pieces of evidence.

The third finding concerns different versions of the ROC curve
in the context of spatial data.
We emphasised that M-ROC is only one particular application of the
ROC concept. We argued that there are many ways to apply/define ROC curves
in this context other than the M-ROC.
We introduced the C-ROC which is based on a covariate rather than a model
(Section~\ref{S:ROC:covariate}). The C-ROC is invariant under monotone transformation of covariate,
and is a measure of ``ranking ability'' rather than ``predictive power''.
We introduced predicted (fitted) and theoretical (true) versions of
the M-ROC and C-ROC (Sections~\ref{S:ROC:model:defn:pixel} and \ref{S:ROC:model:defn:cts}).
We extended all these concepts to spatial point pattern data in continuous space
(Sections~\ref{ss:rocData.cts} and \ref{S:ROC:model:defn:cts}). The 
continuous case is often simpler to understand. We established that the pixel-based ROC
converges as pixel size tends to zero; therefore, ROC curves obtained using different pixel sizes
are approximately consistent. We note that the ROC for presence-absence data effectively
assumes that pixels have equal area; this is sensitive to the choice of coordinate system;
calculating ROC on a square grid in latitude-longitude coordinates would
be ``wrong'' or ``biased''.
We proposed leave-one-out calculation for the M-ROC to
avoid overfitting (end of Section~\ref{S:ROC:model:defn:pixel}).
We extended the M-ROC and C-ROC to spatial case-control point pattern data
(Sections~\ref{ss:casecontrol}, \ref{ss:ROC:covariate:casecontrol},
\ref{S:ROC:model:casecontrol}).
We proposed more extensions of ROC including
the ROC relative to a baseline (Section~\ref{S:ROC.baseline}),
ROC with data weights (Section~\ref{ss:weighted}),
ROC restricted to a subregion (Section~\ref{ss:restriction})
and especially the partial ROC for dropping
an explanatory variable from a fitted model
or for adding an explanatory variable to the model
(Section~\ref{ss:dropROC}).

The fourth finding is a set of new proposed practical techniques
for valid evaluation of model performance
by comparing empirical and fitted ROC curves.
The shape of the empirical C-ROC carries diagnostic information about whether a
monotone model is appropriate (hence whether ROC analysis is appropriate)
(Section~\ref{S:modelcheck}).
The theoretical and model-predicted M-ROC are always concave.
(Corollary~\ref{COR:ROCshape} and Lemma~\ref{L:concave:model}   
in Section~\ref{S:rhohat:ROC}).
The theoretical M-ROC is the most optimistic of all
theoretical ROC curves (Section~\ref{ss:NeymanPearson}).
Model validation is possible by comparing empirical and predicted ROC curves
(either C-ROC or M-ROC) (Section~\ref{S:modelcheck}).
C-ROC has potential as a tool for variable selection (since it is insensitive to
monotone transformations of the explanatory variable).
We also defined partial ROC curves for evaluating the effect of adding another explanatory variable
to an existing model, or the effect of removing one of the explanatory variables from the model.
The dependence of ROC curves on the study region can be turned to advantage,
for example to consider sub-populations or to understand ``interaction'' between variables.

\subsection{Discussion of findings}

The most important implication of this study is that ROC curves for spatial data should not be treated
as providing confirmation of the validity or goodness-of-fit of a model.
They measure the ranking ability of the model, within the study region.
The ROC is bound to the choice of study region.
There is no ability to extrapolate or interpolate the ROC
from one region to another. There is no ability to predict response to changes in the
values of predictor variables, such as response to climate change.
The ROC curve is also insensitive to detecting effects that only involve a small
fraction of the population
(see \ifelseArXiV{Appendix~\ref{SUPP:chorley}}{Section S4 of Supplement}).

On the other hand, the ROC has advantages for \emph{exploratory} data analysis.
The insensitivity of the ROC to monotone transformations
can be exploited for variable selection. It is an advantage
that the ROC can be used before committing to a particular parametric form
of model (for selecting variables, for deciding whether a monotone model
is appropriate, etc.).
Dependence on the study region
can be an advantage when the objective is to segregate
a specific region, e.g.\ to find the best place to put a wind farm to
avoid harm to birds.
The ROC is useful for gauging the magnitude of the ``effect''
of a covariate in a particular spatial region, whereas a statistical model
is a general relationship which could have either a large or small effect in different regions.

Our study presents several new practical opportunities involving ROC curves.
We proposed a technique for model checking by comparing model-predicted and
empirical versions of ROC.
The shape of the ROC curve is diagnostic: if the ROC curve is concave/convex, a monotone model
is appropriate, and the approach is valid. (Note the model-predicted ROC is always concave.)

\subsubsection{Questions for future research}

It would be useful to obtain expressions for the variance of the M-ROC for specific
models fitted by maximum likelihood, especially for logistic regression
and for the loglinear Poisson point process. 

Since ROC curves are closely related to P--P plots  (Section~\ref{S:background:ROC:P-P}),
it may be useful to construct \textbf{Q--Q plots} for the same data,
as these frequently contain complementary information \citep{wilkgnan68}.

\paragraph{Acknowledgements}

We thank Professor Matt Wand for insightful comments,
and Andrew Hardegen, Kassel Hingee and Tom Lawrence who participated
in our early research on this topic.

The Murchison data (Figure~\ref{F:murchison}) are
reproduced by permission of Dr Tim Griffin
of the Geological Survey of Western Australia
and by Dr Carl Knox-Robinson. These data are publicly available
in the \textsf{R} package \texttt{spatstat} \citep{baddrubaturn15}.
Gold deposit and occurrence locations were obtained from a database compiled
by the Geological Survey of Western Australia \citep{MINEDEX}; they include deposits 
of all sizes and are based on map surveys at a scale of 1:50,000
or better \citep{knoxgrov97}. Fault locations were compiled by \citet{watkhick90}.
These data were presented and analysed by \citet{knoxgrov97}
and subsequently by \citet{grovetal00,foxabadd02,baddrubaturn15,badd18iamg}.
Although there exist more up-to-date versions of this survey,
we use the original 1994 version to allow comparison between
different published analyses of the same data.

This project was partially funded by
\textit{The Royal Society of New Zealand}
through Marsden Grant 23-UOO-148.

\addcontentsline{toc}{section}{References}

\bibliography{%
new,%
biblio/agterberg,%
biblio/archaeology,%
biblio/appstat,%
biblio/badd,%
biblio/botany,%
biblio/brillinger,%
biblio/carranza,%
biblio/changepoint,%
biblio/ecology,%
biblio/epidemiol,%
biblio/ford,%
biblio/geoscience,%
biblio/missingdata,%
biblio/prospectivity,%
biblio/speciesdistrib,%
biblio/spatiallogistic,%
biblio/spatstat,%
biblio/stat,%
biblio/stochgeom,%
biblio/thresholdwofe,%
biblio/warick,%
biblio/youdenindex%
}

%% file: appendices.tex
\section*{APPENDICES}
\section{ROC curves for subregions}
\label{app:split_region}

If a region $W$ has been divided into disjoint subsets $W_1,W_2$
and results in ROC curves $R_1(p), R_2(p)$, then the ROC curve
for the entire region $W$ can be reconstructed only if we also know
the areas $a_i = |W_i|$, the numbers of data points $n_i = n(\bx \cap W_i)$
in each subregion, and the spatial cumulative distribution functions
of the covariate in each subregion,
$F_i(z) = |\{u \in W_i: Z(u) \le z\}|/|W_i|$. Then
\begin{equation}
  \label{eq:reconstructR}
  R(p) = \frac{n_1}{n_1 + n_2} R_1(F_1(F^{-1}(p))) +
  \frac{n_2}{n_1 + n_2} R_2(F_2(F^{-1}(p))),
\end{equation}
where
\begin{equation}
  \label{eq:reconstructF}
  F(z) = \frac{a_1}{a_1+a_2} F_1(z) + \frac{a_2}{a_1+a_2} F_2(z) .
\end{equation}

%% file: supplementCore.tex
\section{Models for \textit{Beilschmiedia pendula} trees}
\label{SUPP:models}

As a complement to Figure~\ref{F:bei10modelconventional}
in Section~\ref{S:ROC:model:examples}\ifelseArXiV{,}{of the main article,} 
the left column of Figure~\ref{F:bei.fit} shows contours of the fitted probabilities of the presence
of \textit{Beilschmiedia pendula} in 10-metre pixels, for logistic regression
models depending on terrain elevation, terrain slope, and elevation and slope together.
The right column shows contours of the fitted intensity
for loglinear Poisson point process models,
depending on terrain elevation, terrain slope, and elevation and slope together.

\begin{figure}[!hbt]
  \centering
  \begin{tabular}{rcc}
    & Logistic regression & Loglinear Poisson \\
    \raisebox{2cm}{\mbox{Elevation}} & 
    \includegraphics*[width=0.4\textwidth,bb=15 120 385 310]{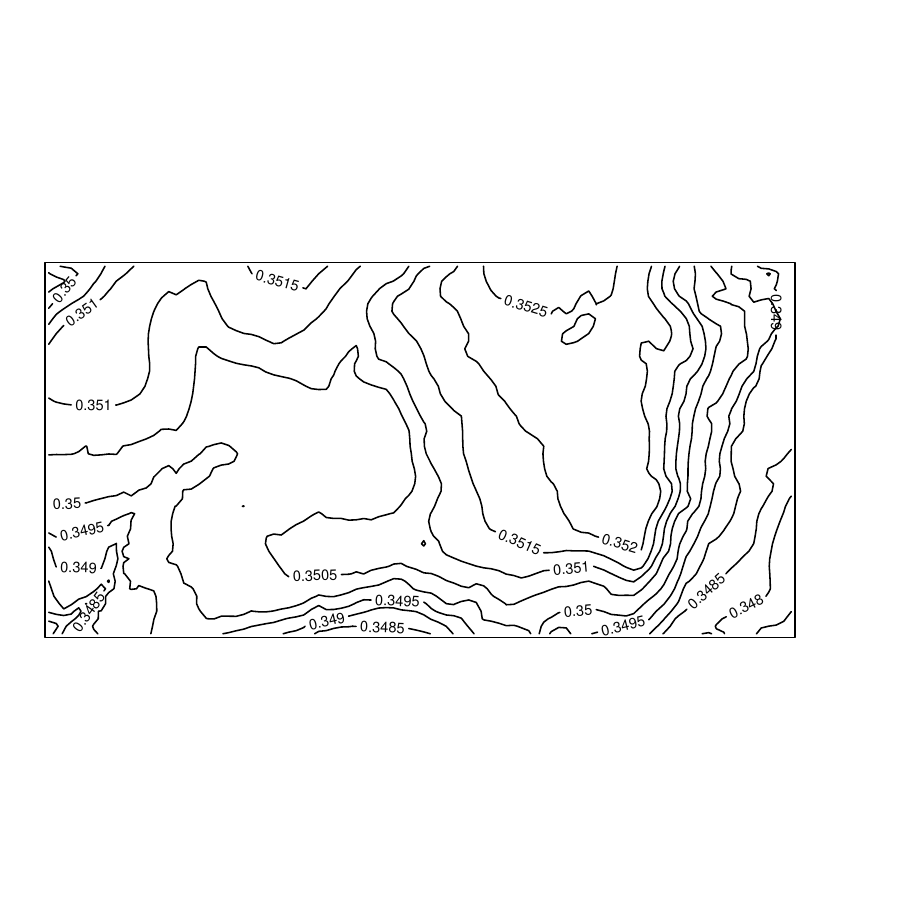} & 
    \includegraphics*[width=0.4\textwidth,bb=15 120 385 310]{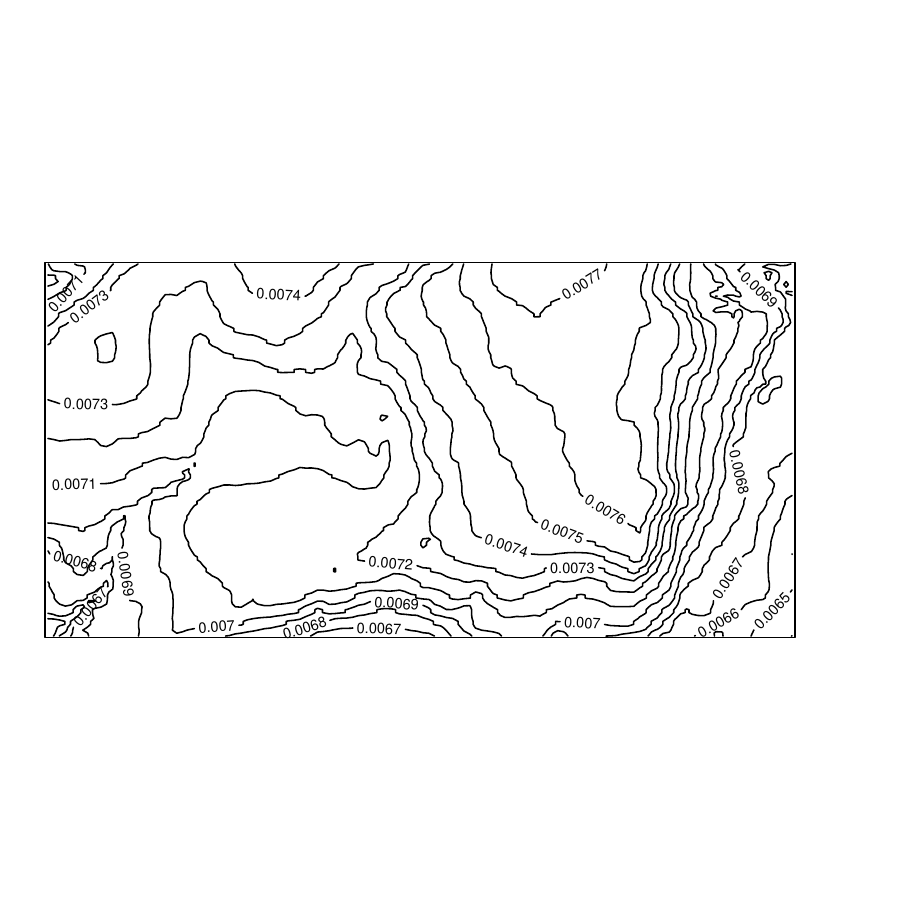} \\
    \raisebox{2cm}{\mbox{Slope}} & 
    \includegraphics*[width=0.4\textwidth,bb=15 120 385 310]{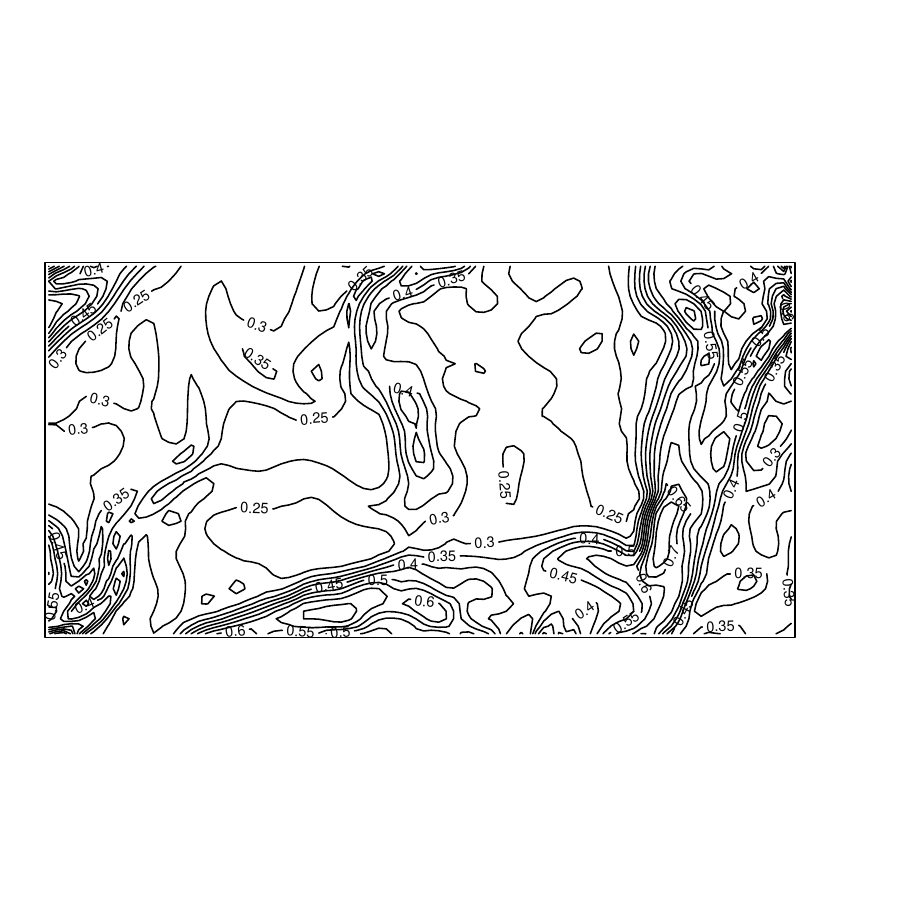} & 
    \includegraphics*[width=0.4\textwidth,bb=15 120 385 310]{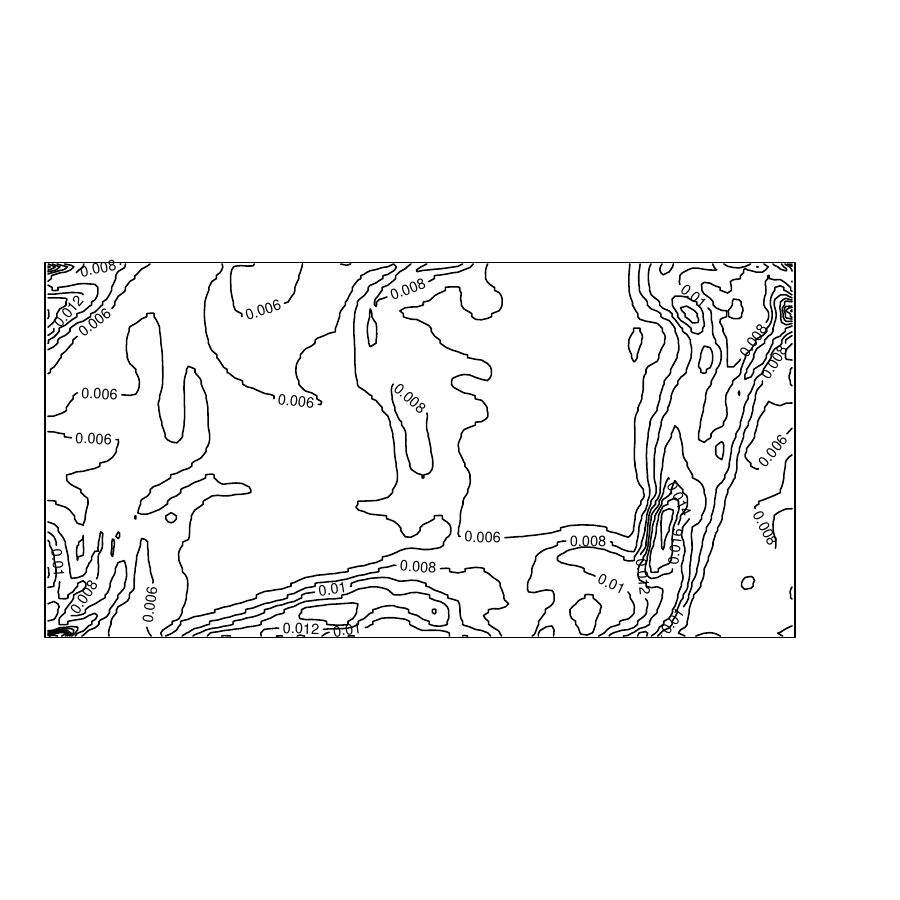} \\
    \raisebox{2cm}{\mbox{Elevation + Slope}} &
    \includegraphics*[width=0.4\textwidth,bb=15 120 385 310]{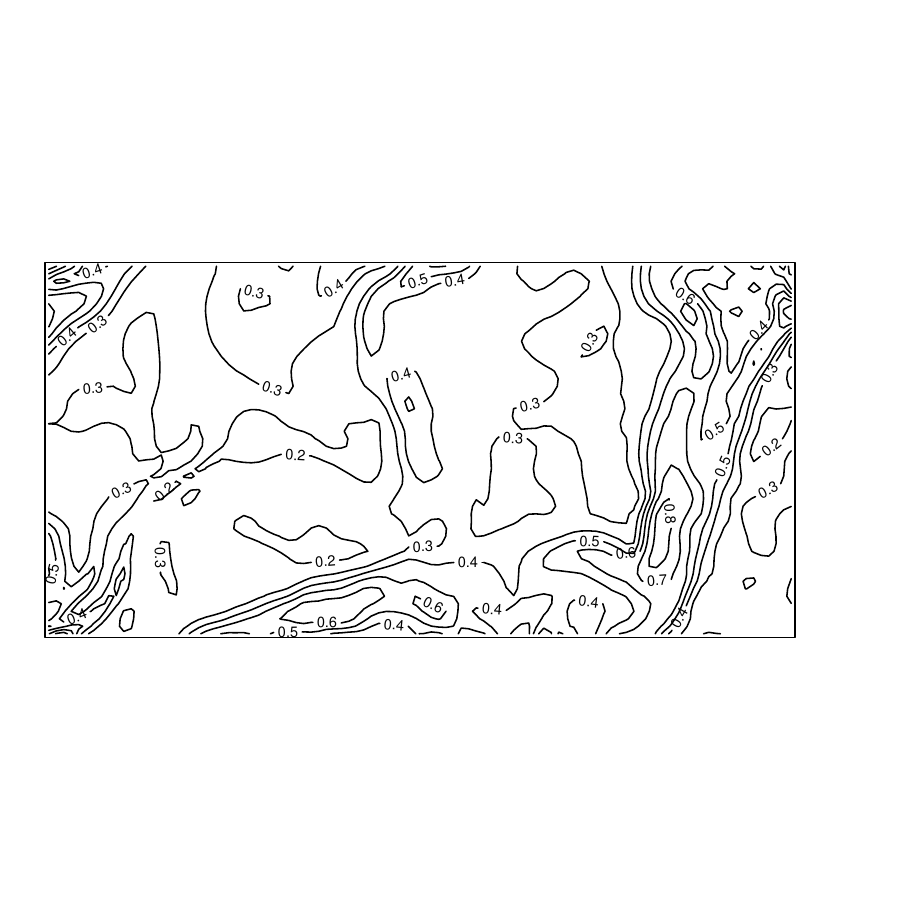} & 
    \includegraphics*[width=0.4\textwidth,bb=15 120 385 310]{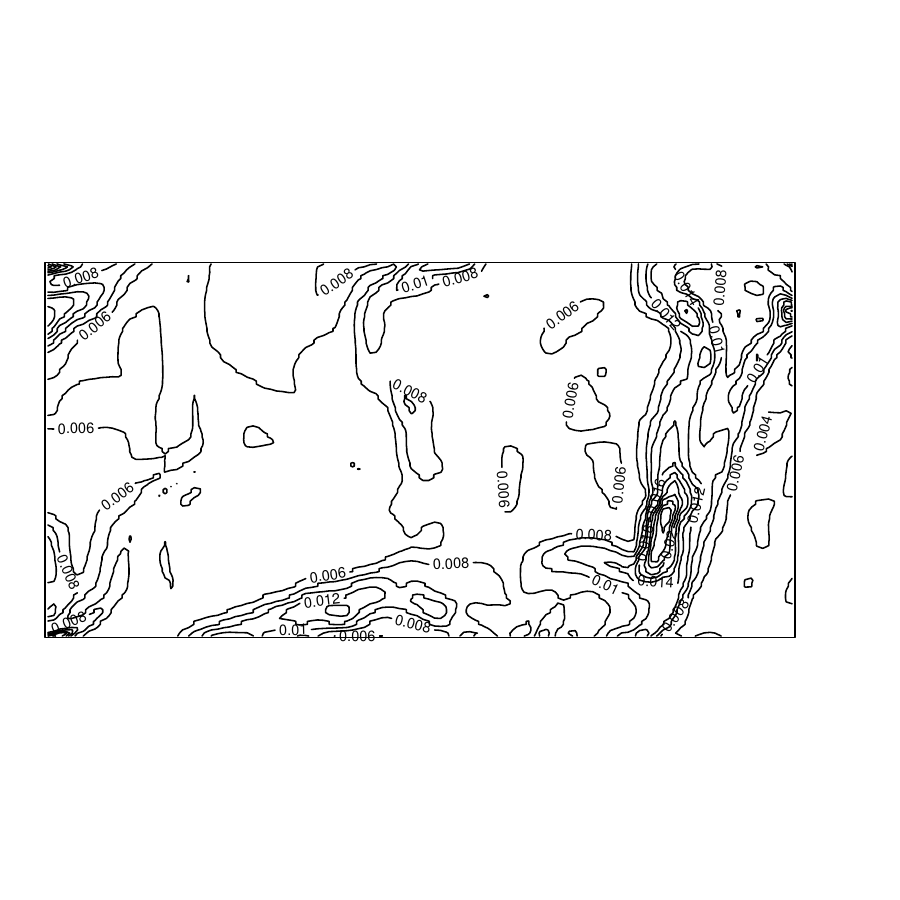}
  \end{tabular}
  \caption{
    Models in which the \textit{Beilschmiedia} presence probability
    depends on terrain elevation and/or slope.
    \emph{Left column} shows contours of fitted probability of presence, $p$,
    within each 10-metre pixel, for logistic regression models.
    \emph{Right column} shows contours of fitted point process intensity $\lambda(x)$
    for loglinear Poisson models.
    \emph{Top row:} models depending on terrain elevation.
    \emph{Middle row:} models depending on terrain slope.
    \emph{Bottom row:} additive models depending on terrain elevation and slope.
  }
  \label{F:bei.fit}
\end{figure}

\section{Weighted ROC for New Brunswick fires}
\label{SUPP:weightedROC}

\begin{figure}[!hb]
  \centering
  \centerline{
    \includegraphics*[width=0.3\textwidth]{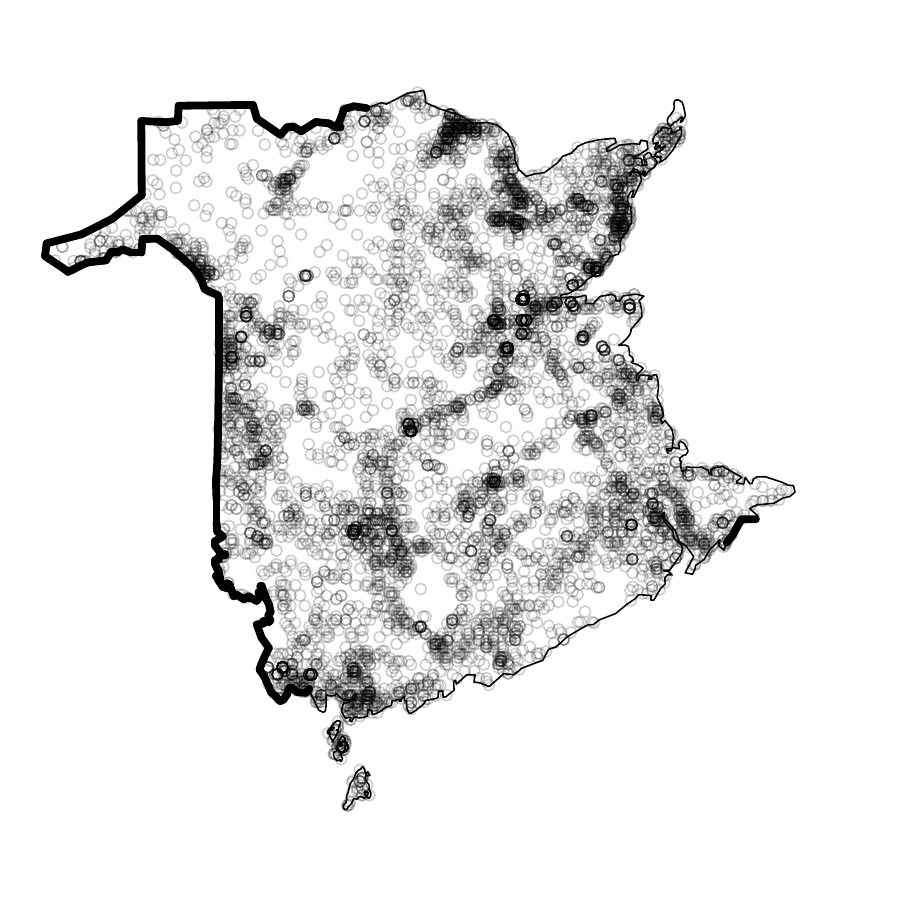}
    \includegraphics*[width=0.3\textwidth]{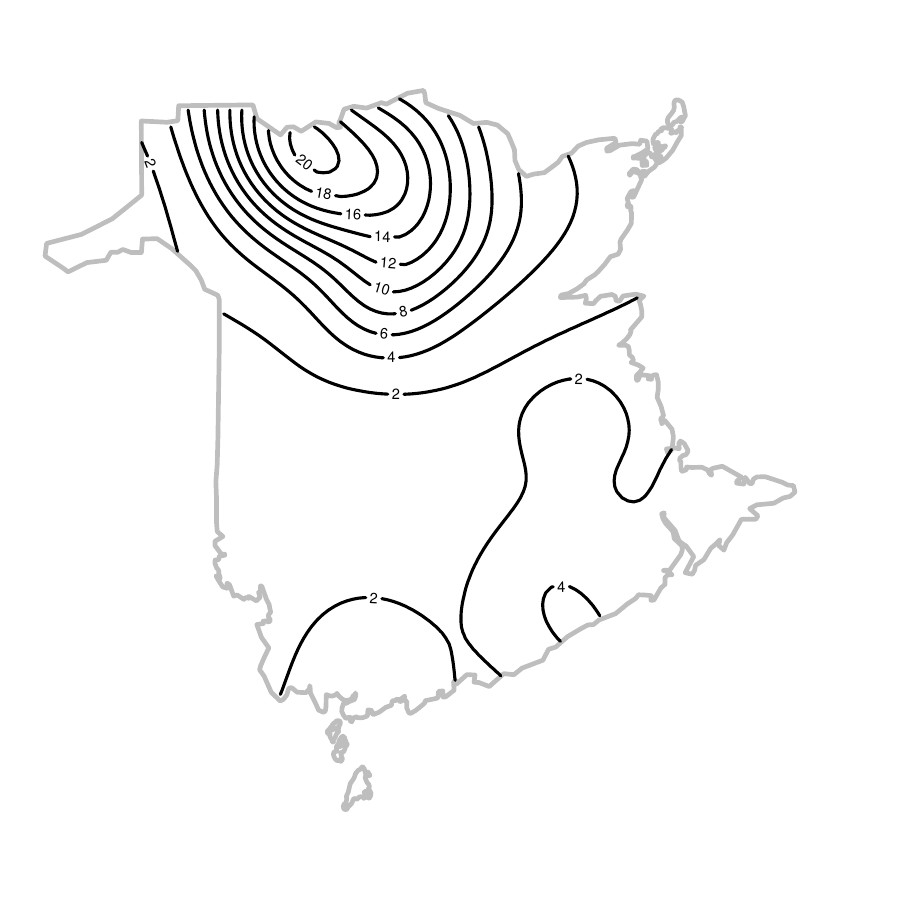}
  }
  \caption{New Brunswick fires.
    \emph{Left:} fire locations (open circles) and
    New Brunswick provincial boundary classified into coastline (thin lines)
    and boundary shared with other territories (thick lines).
    \emph{Right:} contour plot of kernel-smoothed average size of fires.
  }
  \label{AF:nbfires}
\end{figure}

The left panel of Figure~\ref{AF:nbfires}
shows the locations of 7108 fires detected in
the Canadian province of New Brunswick over a 16-year period \citep{turn09}.
The provincial boundary is classified into
coastline and administrative boundary segments.
Information available about each fire includes its final size in hectares.
Fire size distribution is strongly skewed to the right,
with range $0$ to $4871$ hectares, median $0.1$ and mean $3.6$.
The right panel of Figure~\ref{AF:nbfires}
shows a contour plot of a kernel-smoothed spatial average 
of the fire sizes, using the Naradaya-Watson smoother
based on an isotropic Gaussian kernel with standard deviation 30 kilometres.

Figure~\ref{AF:ROCweights} shows unweighted and weighted ROC curves
for the New Brunswick fire locations against distance to
the coastline, assuming short distances are more favorable to fires.
Thin solid lines show the unweighted ROC.
Dashed lines show the weighted ROC, with weights equal to fire size.
Solid lines show the weighted ROC with weights equal to fire size
truncated at 100 hectares.
In this example, the weighted ROC (dashed lines) is strongly influenced by the contribution
from the largest fires. 

\begin{figure}
  \centering
  \centerline{
    \includegraphics*[width=0.35\textwidth,bb=0 0 410 420]{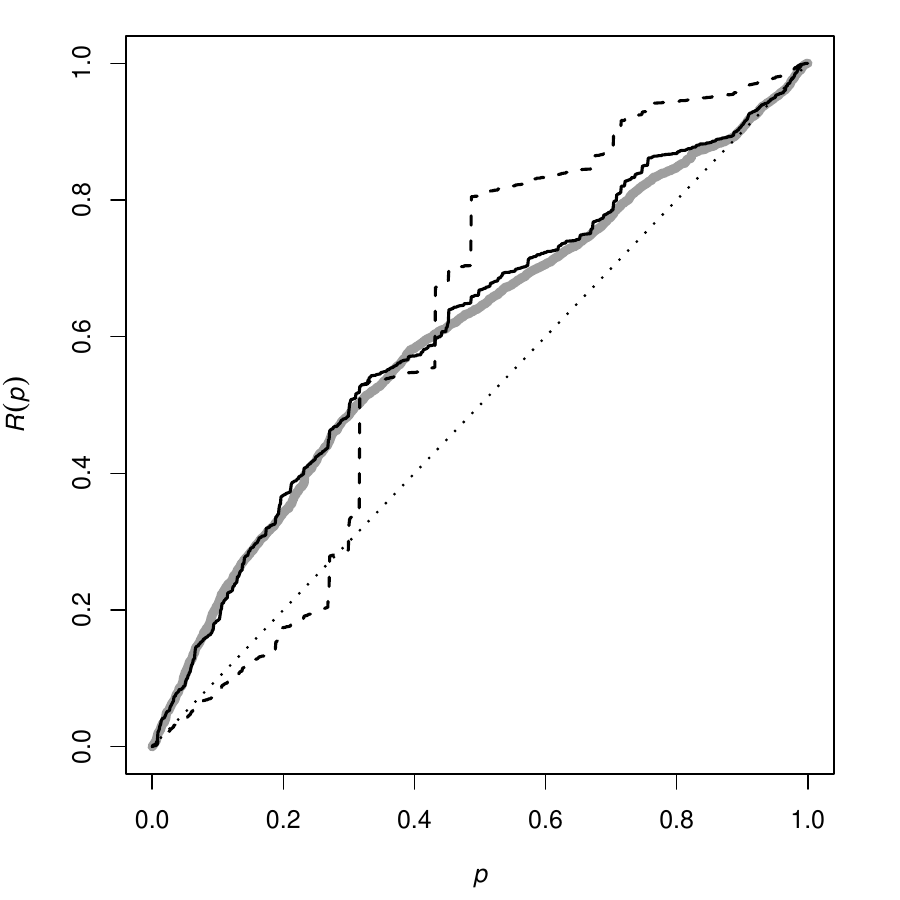}
  }
  \caption{
    Unweighted and weighted ROC curves
    for the New Brunswick fire locations against distance to
    the coastline, assuming short distances are more favorable to fires.
    Thick grey solid lines: unweighted.
    Dashed lines: weight is fire size.
    Thin black solid lines: weight is fire size truncated at 100 hectares.
  }
  \label{AF:ROCweights}
\end{figure}

\clearpage

\section{Poor utility of ROC for Chorley-Ribble cancer data}
\label{SUPP:chorley}

Figure~\ref{AF:chorley}
shows the Chorley-Ribble cancer data of \citet{digg90}
giving the residential locations of new cases 
of cancer of the larynx (58 cases) and cancer of the lung (978 cases)
in the Chorley and South Ribble Health Authority
of Lancashire, England, between 1974 and 1983.
The location of a disused industrial incinerator is also given.
The aim is to assess evidence for an increase in
the incidence of laryngeal cancer close to the 
incinerator. The lung cancer cases serve as a
surrogate for the spatially varying density of the susceptible population. 
Data analysis in \citet{digg90,diggrowl94,baddetal05,hazedavi09} 
concluded there is significant evidence of an incinerator effect.

\begin{figure}[!hb]
  \centering
  \centerline{
    \includegraphics*[width=0.5\textwidth,bb=65 130 395 345]{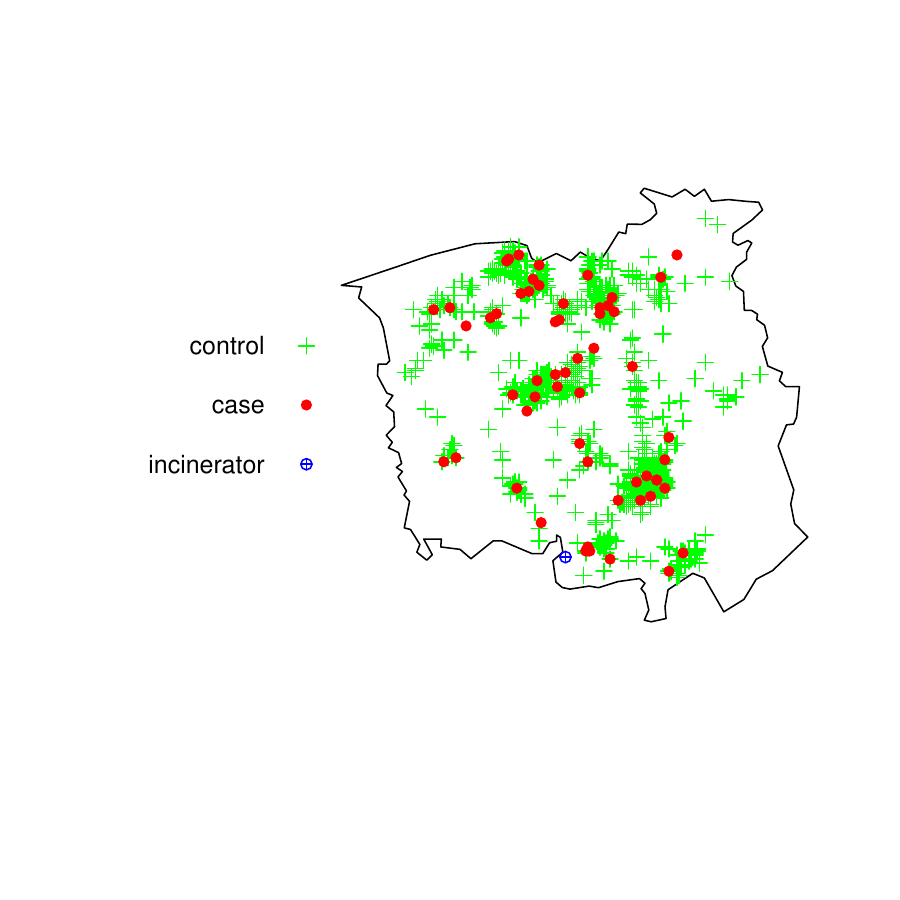}
  }
  \caption{Chorley-Ribble data. 
    Spatial locations of cases of cancer of the larynx ($\bullet$)
    and cancer of the lung ($+$)
    and a disused industrial incinerator ($\oplus$).
    Survey area about 25 kilometres across.
  }
  \label{AF:chorley}
\end{figure}

Figure~\ref{AF:chorleyROC} shows the C-ROC of distance to incinerator
for cases (larynx cancer) against controls (lung cancer) in the Chorley-Ribble
data. The corresponding AUC is 0.54. These suggest no evidence
against the null hypothesis that distance to incinerator has no effect.
The ROC and AUC are not useful in this example,
because the incinerator effect
causes only a small number of excess cases of larynx cancer.

\begin{figure}
  \centering
  \centerline{
    \includegraphics*[width=0.35\textwidth]{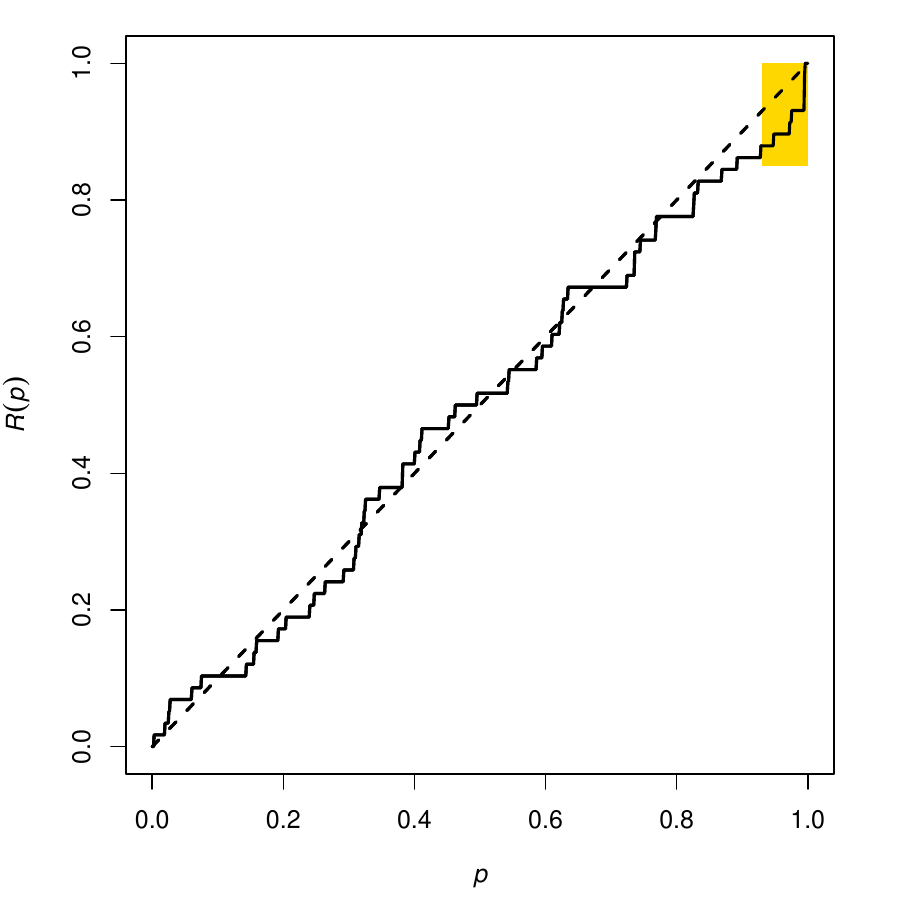}
    \includegraphics*[width=0.35\textwidth]{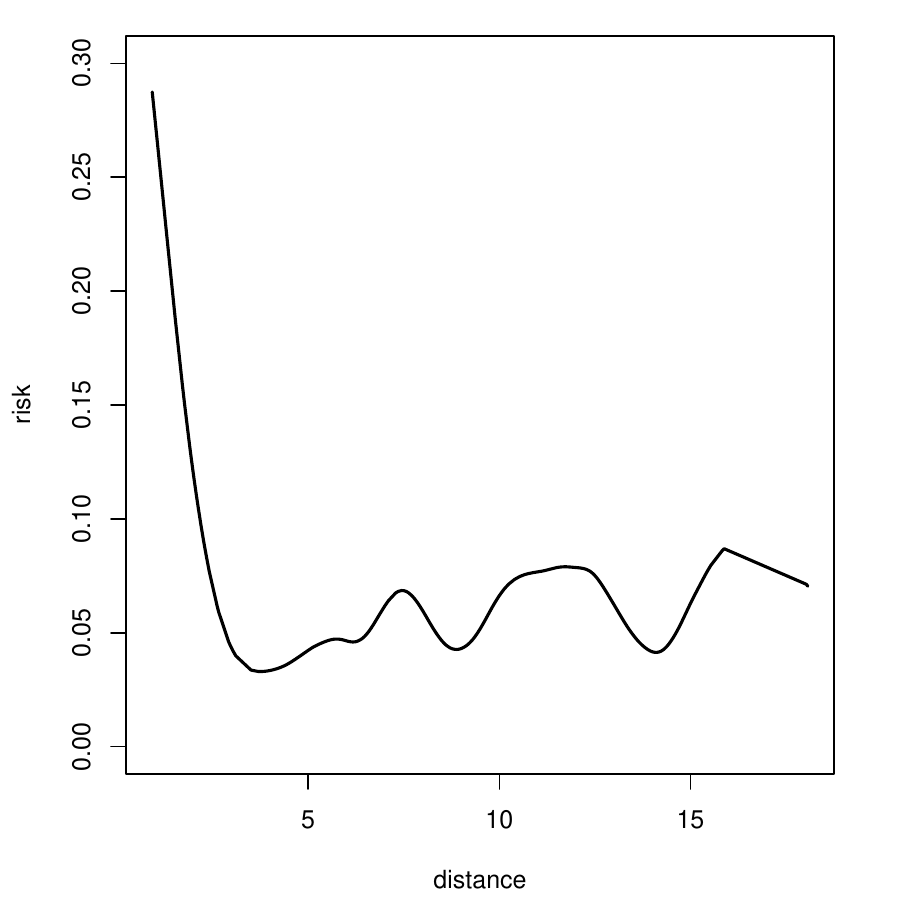}
  }
  \caption{
    Analysis of the effect of distance to incinerator for the
    Chorley-Ribble data.
    \emph{Left:} empirical C-ROC of distance to incinerator.
    \emph{Right:} spline smoothing estimate of probability of a case
    as a function of distance to incinerator.
  }
  \label{AF:chorleyROC}
\end{figure}

\clearpage
\section{Synthetic example of subregion effect}
\label{SUPP:egeStripExample}

Figure~\ref{F:theo:data} shows a synthetic example in which the
original study region is an elongated rectangle  $W = [-10,10] \times [-1,1]$
and the point process intensity is $\lambda(x,y) = 100 \times 2^{-|x|}
= \exp(\log 100 - |x| \, \log 2)$.
  
\begin{figure}[!h]
  \centering
  \centerline{
    \includegraphics*[width=0.9\refwidth]{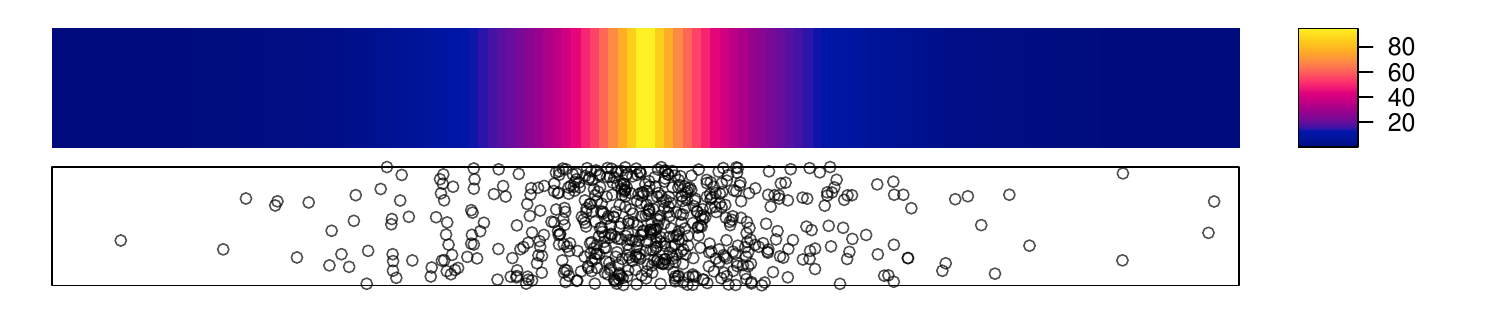}
  }
  \caption{
    Synthetic example illustrating dependence on study region.
    \emph{Top:} intensity function.
    \emph{Bottom:} simulated realisation of Poisson process.
  }
  \label{F:theo:data}
\end{figure}

The left panel of
Figure~\ref{F:theo:ROC} shows the empirical C-ROC for the synthetic data
in Figure~\ref{F:theo:data} using the absolute value of the $x$ coordinate as the covariate.
Also shown is the theoretical C-ROC curve calculated from
\eqref{eq:TP:true:cts} and \eqref{eq:FPhat:cts}.
The right panel shows the simulation envelope of C-ROC curves obtained from 50
simulated realisations of the model which produced Figure~\ref{F:theo:data}.

\begin{figure}[!h]
  \centering
  \centerline{
    \includegraphics*[width=0.3\refwidth,bb=0 0 410 420]{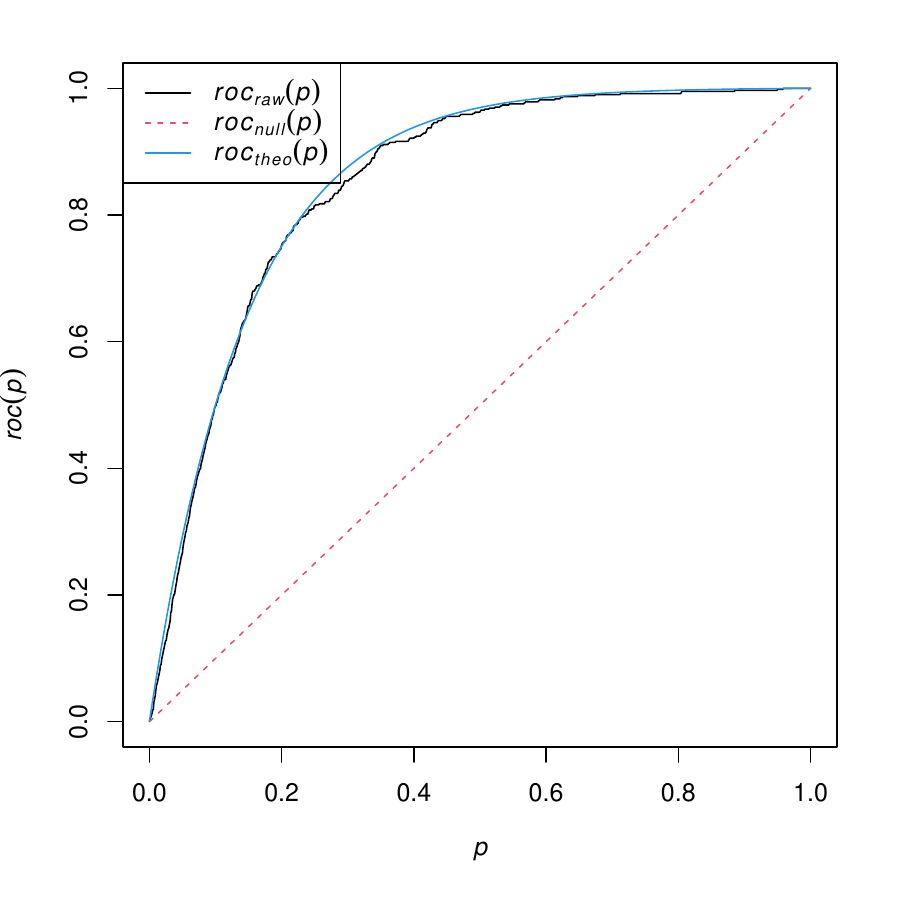}
    \includegraphics*[width=0.3\refwidth,bb=0 0 410 420]{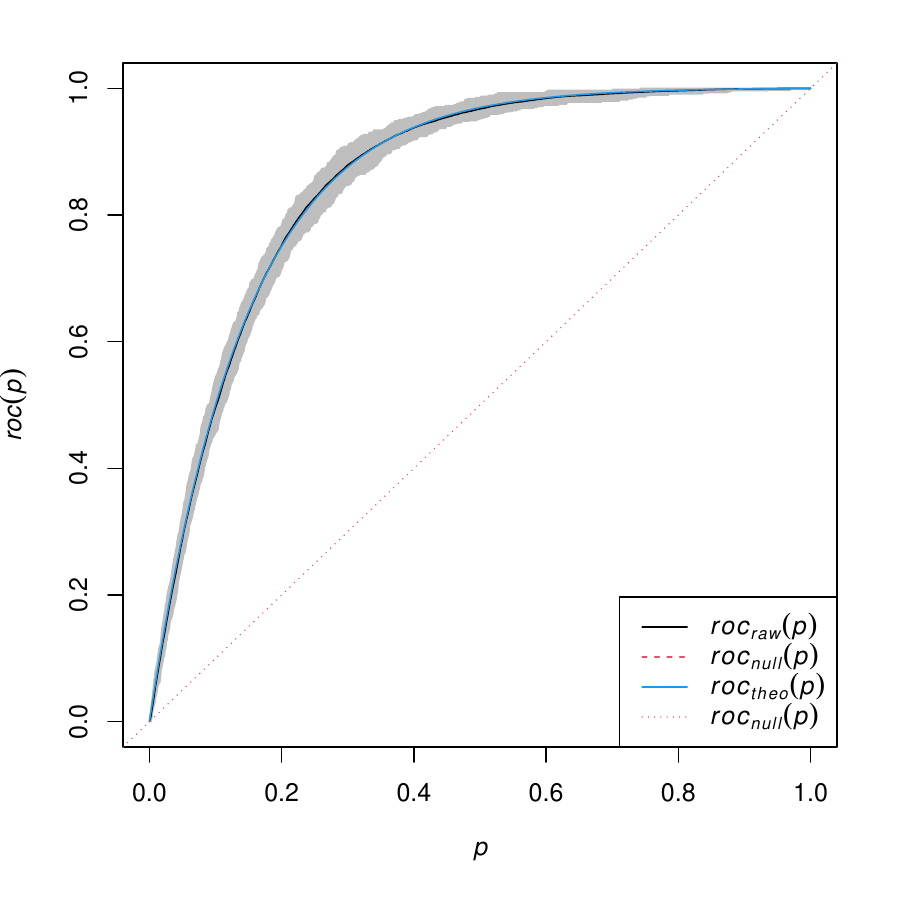}
  }
  \caption{
    \emph{Left:} empirical C-ROC curve for the simulated point pattern
    in the bottom panel of Figure~\ref{F:theo:data}.
    \emph{Right:} envelope of C-ROC curves from 50 simulated realisations
    of the same model.
  }
  \label{F:theo:ROC}
\end{figure}

Figure~\ref{F:theo:depends} shows the C-ROC curves
derived for different subregions $[-a,a] \times [-1,1]$
in the synthetic example.
The left panel shows empirical C-ROC curves, and the right panel
shows the corresponding theoretical C-ROC curves. The C-ROC curve
and AUC value depend greatly on the choice of sub-region. 

\begin{figure}[!hb]
  \centering
  \centerline{
    \includegraphics*[width=0.3\refwidth,bb=0 0 410 420]{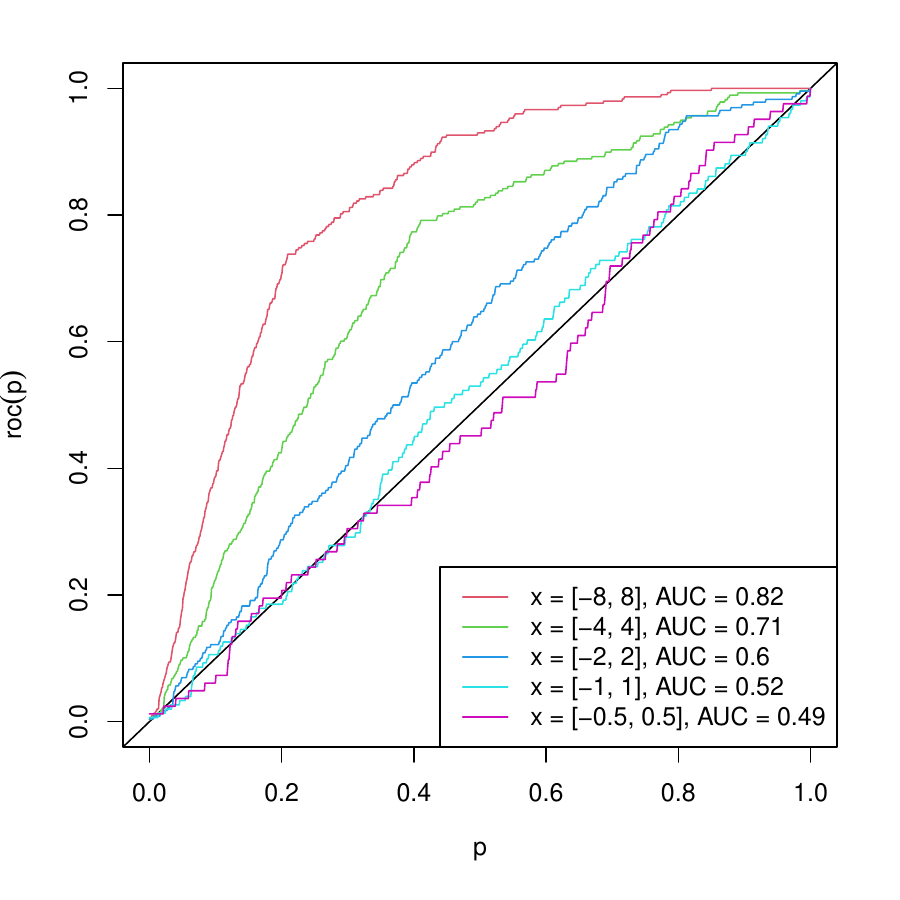}
    \hspace*{3mm}
    \includegraphics*[width=0.3\refwidth,bb=0 0 410 420]{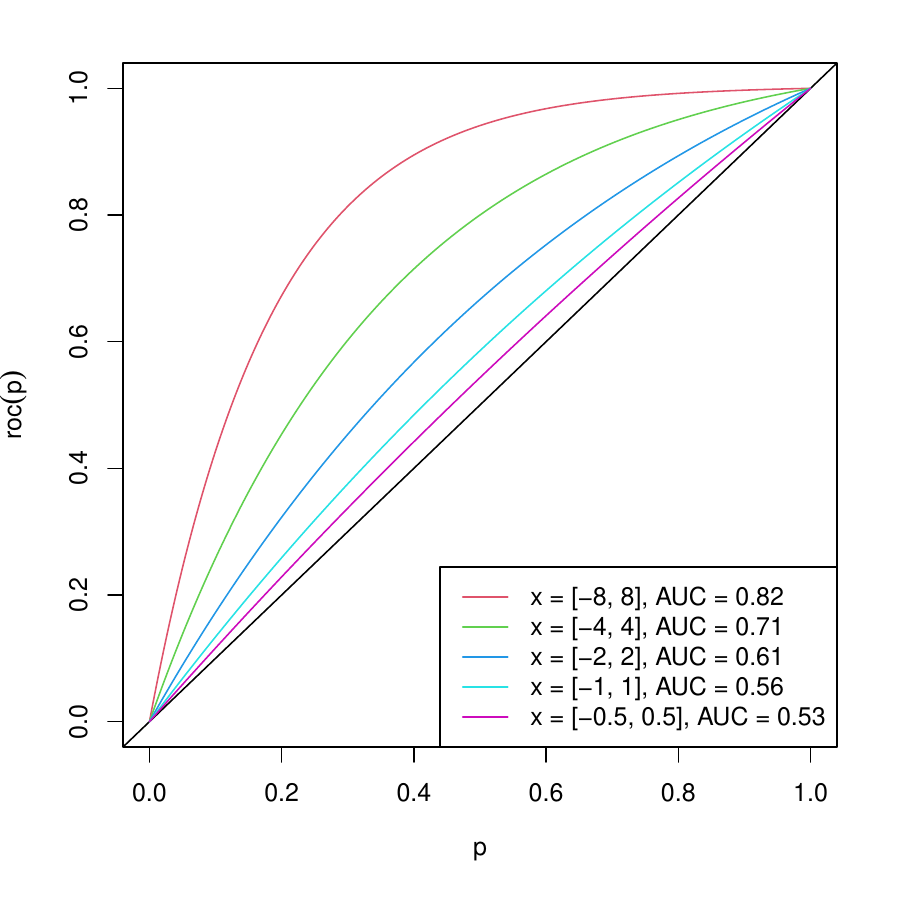}
  }
  \caption{
    C-ROC curves for the synthetic example in Figure~\ref{F:theo:data}
    restricted to different sub-regions.
    \emph{Left:} empirical C-ROC curves for the simulated point pattern
    in the bottom panel of Figure~\ref{F:theo:data}.
    \emph{Right:} theoretical C-ROC curves for the model
    with intensity in the upper panel of Figure~\ref{F:theo:data}.
  }
  \label{F:theo:depends}
\end{figure}